\documentclass{lmcs}
\pdfoutput=1

\usepackage{lastpage}
\lmcsdoi{18}{2}{4}
\lmcsheading{}{\pageref{LastPage}}{}{}%
{Aug.~10,~2020}{Apr.~22,~2022}{}


\usepackage{cmll}
\usepackage[utf8]{inputenc}
\usepackage[T1]{fontenc}
\usepackage{amsmath}
\usepackage{amssymb}
\usepackage{mathtools}
\usepackage{subcaption}
\usepackage{wrapfig}
\usepackage{forest}
\usepackage{xspace}
\usepackage{keycommand} 

\usepackage{tikz}
\usetikzlibrary{matrix,arrows,calc}
\usetikzlibrary{decorations.markings}
\usetikzlibrary{decorations.pathmorphing, backgrounds,
  positioning, fit, shapes.geometric, shapes.multipart,
  decorations.pathreplacing}
\usepackage{tikz-cd}
\tikzcdset{ampersand replacement = \&}

\usepackage{pnLuc}
\usetikzlibrary{backgrounds,fit,positioning}

\usepackage{hyperref}

\usepackage[xcolor,hyperref,notion,paper]{knowledge}
\knowledgeconfigureenvironment{lem,proof,prop,conj,cor,defi,exa,rem,thm,fact}{}

\theoremstyle{plain}
\newtheorem{theorem}[thm]{Theorem}
\newtheorem{proposition}[thm]{Proposition}
\newtheorem{lemma}[thm]{Lemma}

\theoremstyle{definition}
\newtheorem{definition}[thm]{Definition}
\newtheorem{example}[thm]{Example}
\newtheorem{remark}[thm]{Remark}

\usepackage{cleveref}
\crefname{thm}{Theorem}{Theorems}
\crefname{lem}{Lemma}{Lemmas}
\crefname{prop}{Proposition}{Propositions}
\crefname{defi}{Definition}{Definitions}
\crefname{exa}{Example}{Examples}
\crefname{rem}{Remark}{Remarks}
\crefname{theorem}{Theorem}{Theorems}
\crefname{lemmma}{Lemma}{Lemmas}
\crefname{proposition}{Proposition}{Propositions}
\crefname{definition}{Definition}{Definitions}
\crefname{example}{Example}{Examples}
\crefname{remark}{Remark}{Remarks}
\Crefname{thm}{Theorem}{Theorems}
\Crefname{lem}{Lemma}{Lemmas}
\Crefname{prop}{Proposition}{Propositions}
\Crefname{defi}{Definition}{Definitions}
\Crefname{exa}{Example}{Examples}
\Crefname{rem}{Remark}{Remarks}
\Crefname{theorem}{Theorem}{Theorems}
\Crefname{lemma}{Lemma}{Lemmas}
\Crefname{proposition}{Proposition}{Propositions}
\Crefname{definition}{Definition}{Definitions}
\Crefname{example}{Example}{Examples}
\Crefname{remark}{Remark}{Remarks}

\newcommand{\tikzsetnextfilename}[1]{}

\usepackage{trimclip}
\DeclareMathOperator{\WHY}{\textup{?}}
\def\?{\setbox0=\hbox{$\WHY$}\raisebox{.2\ht0}{\clipbox{0pt .2\ht0 0pt
      -.1\ht0}{?}}}

\newcommand{\Contr}{
		\textup{\ooalign{\raisebox{0ex}{$\?$}\cr\raisebox{-.5ex}{$\scriptscriptstyle
					c$}}}}
\newcommand{\Der}{
		\textup{\ooalign{\raisebox{0ex}{$\?$}\cr\raisebox{-.5ex}{$\scriptscriptstyle
					d$}}}}
\newcommand{\Weak}{
		\textup{\ooalign{\raisebox{0ex}{$\?$}\cr\raisebox{-.5ex}{$\scriptscriptstyle
					w$}}}}
\newcommand{\Ax}{\mathbin{\ax}}
\newcommand{\OcAx}{\mathbin{\ocax}}
\newcommand{\Cut}{\mathbin{\cut}}
\newcommand{\Ex}{\mathbin{\mathtt{exc}}}
\newcommand{\One}{\mathbin{\mathbf{1}}}
\newcommand{\BoxR}{\mathbin{\mathtt{box}}}
\newcommand{\Mix}{\mathbin{\mathtt{mix}}}

\newcommand{\Orient}{\ensuremath{\mathsf{o}}}

\newcommand{\emptynet}{\kl[empty proof-structure]{\varepsilon}}
\newcommand{\emptynets}{\kl[empty quasi-proof-structure]{\vec{\varepsilon}}}
\newcommand{\emptylist}{\kl[empty list]{\epsilon}}
\newcommand{\emptylists}{\kl[empty lists]{\vec{\epsilon}}}

\newcommand{\pnrewrite}[1][\quad]{%
  \mathrel{%
    \begin{tikzpicture}[%
      baseline={(current bounding box.south)}
      ]
      \node[%
      ,inner sep=.44ex
      ,align=center
      ] (tmp) {$\scriptstyle #1$\vphantom{\quad}};
      \path[%
      ,draw,<-
      ,decorate,decoration={%
        ,snake
        ,amplitude=0.7pt
        ,segment length=1.2mm,pre length=3.5pt
      }
      ] 
      (tmp.south east) -- (tmp.south west);
    \end{tikzpicture}
  }
}

\newcommand{\FlagType}{\kl[edge type]{\ensuremath{\mathsf{c}}}}
\newcommand{\VertType}{\kl[vertex type]{\ensuremath{\mathsf{\ell}}}}

\newcommand{\SystemFont}[1]{\mathsf{#1}}
\newcommand{\LL}{\ensuremath{\SystemFont{LL}}\xspace}
\newcommand{\MELLPS}{\kl[$\MELL$ proof-structure]{\MELL}}
\newcommand{\MELL}{\ensuremath{\SystemFont{ME}\LL}\xspace}
\newcommand{\MLL}{\ensuremath{\SystemFont{M}\LL}\xspace}
\newcommand{\DiLL}{\ensuremath{\SystemFont{Di}\LL}\xspace}

\newcommand{\polyadic}{\kl[resource]{\ensuremath{\DiLL_0}}}

\newcommand{\one}{\ensuremath{\mathbf{1}}}
\newcommand{\ax}{\ensuremath{\mathtt{ax}}}
\newcommand{\ocax}{\ensuremath{\oc\mathtt{ax}}}
\newcommand{\cut}{\ensuremath{\mathtt{cut}}}

\newcommand{\ie}{\textit{i.e.}\xspace}
\newcommand{\eg}{\textit{e.g.}\xspace}

\newcommand{\resp}{resp.\xspace}

\newcommand{\gluability}{gluability\xspace}

\newcommand{\Gluability}{Gluability\xspace}

\newcommand{\gluable}{gluable\xspace}

\newcommand{\smallproofnets}{0.75}
\newcommand{\beforepn}{-1cm}

\newcommand{\PowerSet}[1]{\mathfrak{P}(#1)}

\newcommand{\NatTransf}{\ensuremath{\mathfrak{T}^{\maltese}}}
\newcommand{\Taylor}[1]{\kl[Taylor expansion]{\TaylorSym(#1)}}
\newcommand{\TaylorSym}{\ensuremath{\mathcal{T}}}
\newcommand{\ProtoTaylor}[1]{\kl[proto-Taylor expansion]
  {\ensuremath{\TaylorSym^{\mathsf{pr}}(#1)}}}
\newcommand{\FatTaylor}[1]{\ensuremath{\TaylorSym^{\maltese}(#1)}}
\newcommand{\Size}[1]{\mathsf{s}(#1)}

\newcommand{\TreeT}{\ensuremath{\mathcal{A}}}
\newcommand{\ForestT}{\ensuremath{\mathcal{F}}}

\newcommand{\ReflexiveTransitive}{\kl[reflexive-transitive closure]{\circlearrowleft}}

\newcommand{\id}{\mathrm{id}}
\newcommand{\sched}{\nu}
\newcommand{\schedtwo}{\mu}

\newcommand{\Nat}{\ensuremath{\mathbb{N}}}

\newcommand{\BoxFunction}{\ensuremath{\mathsf{box}}}

\newcommand{\In}{\ensuremath{\mathbf{in}}}
\newcommand{\Out}{\ensuremath{\mathbf{out}}}

\newcommand{\Formulas}{\ensuremath{\mathcal{F}\!\mathit{orm}}}

\newcommand{\CatFont}[1]{\mathbf{#1}}
\newcommand{\Graph}{\ensuremath{\CatFont{Graph}}}
\newcommand{\Scheduling}{\kl[scheduling]{\CatFont{Sched}}}
\newcommand{\Rel}{\kl[relation]{\CatFont{Rel}}}

\newcommand{\op}[1]{{#1}^{\mathrm{op}}}
\newcommand{\RootedTree}{\ensuremath{\CatFont{RoTree}}}
\newcommand{\GenericFont}[1]{\mathcal{#1}}
\newcommand{\catC}{\ensuremath{\GenericFont{C}}} 
\newcommand{\MELLFunctor}{\kl[mell functor]{\ensuremath{\mathsf{qMELL}}}}
\newcommand{\PolyPN}{\ensuremath{\mathsf{qDiLL_0}}}
\newcommand{\PPolyPN}{\kl[dill functor]{\ensuremath{\mathfrak{P}\PolyPN}}}

\newcommand{\mathnodes}[4]{
	\matrix (m) [
	matrix of math nodes,
	row sep=#1em,
	column sep=#2em,
	minimum width=#3em,
	ampersand replacement = \&]
	{#4}
}
\NewEnviron{commutativediagram}[4]{
	{\centering
		\begin{tikzpicture}[>=stealth]
		\mathnodes{#1}{#2}{#3} { #4 } \BODY
		\end{tikzpicture}
		
	}
}

\knowledge{notion} 
| graph
| graphs

\knowledge{notion}
| empty graph 
| empty graphs

\knowledge{notion}
| subgraph 
| subgraphs

\knowledge{notion} 
| edge
| edges

\knowledge{notion} 
| flag
| flags

\knowledge{notion} 
| vertex
| vertices

\knowledge{notion} 
| tail
| tails

\knowledge{notion} 
| half
| halves

\knowledge{notion}
| loop
| loops

\knowledge{notion} 
| corolla
| corollas

\knowledge{notion} 
| graph morphism
| graph morphisms

\knowledge{notion}
| empty graph morphism
| empty graph morphisms

\knowledge{notion}
| identity graph morphism
| identity graph morphisms

\knowledge{notion} 
| graph isomorphism
| graph isomorphisms

\knowledge{notion} 
| labeled graph
| labeled graphs

\knowledge{notion} 
| colored graph
| colored graphs

\knowledge{notion} 
| ordered graph
| ordered graphs

\knowledge{notion} 
| directed graph
| directed graphs

\knowledge{notion} 
| input
| inputs

\knowledge{notion} 
| output
| outputs

\knowledge{notion} 
| acyclic
| forest
| forests

\knowledge{notion} 
| reflexive-transitive closure

\knowledge{notion} 
| tree
| trees

\knowledge{notion} 
| connected component
| connected components

\knowledge{notion} 
| rooted tree
| rooted trees

\knowledge{notion} 
| root
| roots

\knowledge{notion} 
| path
| paths
| undirected path
| undirected paths

\knowledge{notion}
| path length

\knowledge{notion} 
| connected

\knowledge{notion} 
| cycle
| cycles

\knowledge{notion}  
| directed graph morphism
| directed graph morphisms

\knowledge{notion} 
| pullback
| pullbacks

\knowledge{notion} 
| \MELL formula
| \MELL formulas
| formula
| \(\MELL\) formul\ae{}
| \(\MELL\) formulas
| \(\MELL\) formula
| formul\ae{}
| formulas

\knowledge{notion}
| dual
| dual formulas

\knowledge{notion}
| exponential
| exponentials
| exponential formulas
| exponential formula

\knowledge{notion}
| flattening

\knowledge{notion} 
| atomic
| atomic formula

\knowledge{notion}
| connective
| connectives

\knowledge{notion}
| unit
| units

\knowledge{notion} 
| linear negation
| Linear negation

\knowledge{notion}
| empty list
| empty lists

\knowledge{notion} 
| proof-structure
| proof-structures

\knowledge{notion} 
| undirected graph

\knowledge{notion} 
| structured graph
| structured graphs

\knowledge{notion} 
| structured graph of quasi-proof-structure

\knowledge{notion} 
| quasi-proof-structure
| quasi-proof-structures

\knowledge{notion} 
| $\MELL$ proof-structure
| $\MELL$ proof-structures
| \(\MELLPS\)
| \MELL proof-structure
| \MELL proof-structures
| MELL proof-structure
| MELL proof-structures

\knowledge{notion} 
| $\DiLL_0$ proof-structure
| $\DiLL_0$ proof-structures
| DiLL0 proof-structure
| DiLL0 proof-structures
| polyadic
| resource
 
\knowledge{notion} 
| $\MELL$ quasi-proof-structure
| $\MELL$ quasi-proof-structures
| \MELL quasi-proof-structure
| \MELL quasi-proof-structures
| MELL quasi-proof-structure
| MELL quasi-proof-structures

\knowledge{notion} 
| $\DiLL_0$ quasi-proof-structure
| $\DiLL_0$ quasi-proof-structures
| \DiLL0 quasi-proof-structure
| \DiLL0 quasi-proof-structures
| DiLL0 quasi-proof-structure
| DiLL0 quasi-proof-structures

\knowledge{notion} 
| daimon
| daimons
| Daimon
| Daimons

\knowledge{notion} 
| \(\MELL\) with empty boxes

\knowledge{notion} 
| vertex type

\knowledge{notion} 
| edge type
| edge types

\knowledge{notion} 
| conclusion
| conclusions

\knowledge{notion} 
| conclusion of quasi-proof-structure
| conclusions of quasi-proof-structure

\knowledge{notion} 
| partition
| partitions

\knowledge{notion}
| module
| modules

\knowledge{notion}
| box-function
| box-functions

\knowledge{notion, index = box-function-quasi}
| box-function of quasi-proof-structure

\knowledge{notion}
| box-tree

\knowledge{notion}
| box-forest

\knowledge{notion}
| box
| boxes
| Box
| Boxes

\knowledge{notion}
| depth
| depths
| Depth
| Depths

\knowledge{notion, index = depth-proof-structure}
| depth of proof-structure

\knowledge{notion}
| empty proof-structure
| empty proof-structures

\knowledge{notion}
| empty quasi-proof-structure
| empty quasi-proof-structures

\knowledge{notion}
| cut-free
| Cut-free

\knowledge{notion}
| with cuts

\knowledge{notion}
| proof-structure type

\knowledge{notion}
| cell
| cells

\knowledge{notion}
| hypothesis
| hypotheses
| Hypothesis
| Hypotheses

\knowledge{notion}
| component
| components
| Component
| Components

\knowledge{notion}
| quasi-proof-structure type

\knowledge{notion}
| thick subtree
| thick subtrees

\knowledge{notion}
| thick subforest
| thick subforests

\knowledge{notion}
| tree expansion

\knowledge{notion}
| proto-Taylor expansion

\knowledge{notion}
| Taylor expansion

\knowledge{notion}
| filled Taylor expansion
| \FatTaylor
| filled

\knowledge{notion, index = emptying of proof-structure}
| emptying of proof-structure
| emptyings of proof-structure
| Emptying of proof-structure
| Emptyings of proof-structure

\knowledge{notion, index = emptying of quasi-proof-structure}
| emptying
| emptyings
| Emptying
| Emptyings

\knowledge{notion, index = emptying on}
| emptying on
| emptyings on
| Emptying on
| Emptyings on

\knowledge{notion,index=elementary scheduling}
| elementary scheduling
| elementary schedulings

\knowledge{notion,index=scheduling}
| scheduling
| schedulings
| Scheduling
| Schedulings

\knowledge{notion,index=empty scheduling}
| empty scheduling
| empty schedulings
| Empty scheduling
| Empty schedulings

\knowledge{notion,index=scheduling length}
| length
| scheduling length
| length of scheduling

\knowledge{notion}
| mell functor

\knowledge{notion}
| action of elementary scheduling
| action of the elementary scheduling

\knowledge{notion}
| action of scheduling
| actions of scheduling
| actions

\knowledge{notion}
| relation
| relations

\knowledge{notion}
| applies
| apply
| application

\knowledge{notion}
| nullary splitting for elementary scheduling

\knowledge{notion}
| nullary splitting for scheduling

\knowledge{notion}
| cut-free scheduling
| cut-free schedulings
| Cut-free scheduling
| Cut-free schedulings

\knowledge{notion}
| qDiLL0
| \PolyPN

\knowledge{notion}
| dill functor

\knowledge{notion}
| action of elementary scheduling on dill

\knowledge{notion}
| action of scheduling on dill
| action on dill
| actions on dill

\knowledge{notion}
| glueable

\knowledge{notion}
| full

\knowledge{notion}
| cut-free glueable

\knowledge{notion}
| canonical

\knowledge{notion}
| canonically glueable

\knowledge{notion}
| daimon action
| maltese action

\knowledge{notion}
| nondaimon action
| nonmaltese action

\knowledge{notion}
| dill0 nullary splitting 

\knowledge{notion}
| number of nullary splittings

\input{insbox}

\begin{document}

\title[Gluing resource proof-structures: inhabitation and inverting Taylor]{Gluing resource proof-structures: inhabitation and inverting the {Taylor} expansion}

\author[G. Guerrieri]{Giulio Guerrieri}[a]	
\address{Huawei Research, Edinburgh Research Centre, Edinburgh, United Kingdom}	
\email{giulio.guerrieri@huawei.com}  

\author[L. Pellissier]{Luc Pellissier}[b]	
\address{Université Paris Est Créteil, LACL, F-94010 Créteil, France}	
\email{luc.pellissier@lacl.fr}  

\author[L. Tortora de Falco]{Lorenzo Tortora de Falco}[c]	
\address{Università Roma Tre, Dipartimento di Matematica e Fisica, Rome, Italy}	
\email{tortora@uniroma3.it}  

\begin{abstract}
  A Multiplicative-Exponential Linear Logic (MELL) proof-structure can be
  expanded into a set of resource proof-structures: its Taylor expansion. We
  introduce a new criterion characterizing (and deciding in the finite case) those sets of resource
  proof-structures that are part of the Taylor expansion of some MELL
  proof-structure, through a rewriting system acting both on resource and MELL
  proof-structures.  We also prove semi-decidability of
  the type inhabitation problem for cut-free MELL proof-structures.
\end{abstract}

\maketitle

\section{Introduction}
\label{sect:intro}

\subsubsection*{The Taylor expansion}
Girard's linear logic (\LL, \cite{Girard:1987}) is a
refinement of intuitionistic and classical logic that isolates the infinitary
parts of reasoning 
in two dual modalities: the \emph{exponentials} $\oc$
and $\wn$.
They give a logical status to operations of memory management such as
\emph{copying} and \emph{erasing}: a linear (\ie~exponential-free) proof corresponds---via the
Curry--Howard isomorphism---to a program that uses its argument \emph{linearly},
\ie~exactly once, while an exponential proof corresponds to a program that can
use its argument \emph{at will}.

The intuition that linear programs are analogous to linear functions
(as studied in linear algebra) while exponential programs mirror a more general
class of analytic functions
got a technical incarnation in Ehrhard's work \cite{Ehrhard:2002,Ehrhard:2005} on $\LL$-based denotational semantics for the $\lambda$-calculus.
This investigation has been then internalized in the syntax, yielding the \emph{resource $\lambda$-calculus} \cite{Ehrhard:2003,Ehrhard:2008}, inspired by \cite{Boudol}: 
there, copying and erasing are forbidden 
and replaced by the possibility to apply a 
function to a \emph{bag} of resource $\lambda$-terms which specifies how many times (in a finite number) an argument can be linearly passed to the function, so as to represent only bounded computations.

The \emph{Taylor expansion} \cite{Ehrhard:2008} (more precisely, its support; but here and throughout the paper we do not consider the rational coefficients in the Taylor expansion) associates with every ordinary $\lambda$-term a---generally infinite---set of resource $\lambda$-terms, recursively approximating
the usual application: the Taylor expansion of the $\lambda$-term $MN$ is made up
of all resource $\lambda$-terms of the form $t[u_1, \dots, u_n]$, where $t$ is a
resource $\lambda$-term in the Taylor expansion of $M$, and $[u_1,\dots,u_n]$ is
a bag of $n$ (for any $n \geqslant 0$) resource $\lambda$-terms in
the Taylor expansion of $N$.  Roughly, the idea is to decompose a program into a
set of purely ``resource-sensitive programs'', all of them containing only bounded (although
possibly non-linear) calls to inputs. The notion of Taylor expansion has had
many applications in the theory of the $\lambda$-calculus, \eg{} in the study of
linear head reduction \cite{EhrhardRegnier06}, normalization
\cite{DBLP:conf/fossacs/PaganiTV16,DBLP:conf/csl/Vaux17}, Böhm trees
\cite{DBLP:conf/csl/BoudesHP13,KerinecMP18,BarbarossaManzonetto20},
$\lambda$-theories~\cite{ManzonettoR14}, intersection
types~\cite{MazzaPellissierVial18}. 
In general, understanding the relation
between a program and its Taylor expansion renews the logical approach to the
quantitative analysis of computation started with the inception~of~\LL.

A natural question is the \emph{inverse Taylor expansion problem}: Which sets of resource $\lambda$-terms are included in the Taylor expansion of a same $\lambda$-term? 
How to characterize them?
Ehrhard and Regnier \cite{Ehrhard:2008} defined a simple \emph{coherence} binary relation such that a finite set of resource $\lambda$-terms is included in the Taylor expansion of a $\lambda$-term if and only if all the elements
of this set are pairwise coherent. Coherence is crucial in many structural
properties of the resource $\lambda$-calculus, such as in the proof that in the $\lambda$-calculus normalization and Taylor expansion commute \cite{EhrhardRegnier06,Ehrhard:2008}.

We aim to solve the inverse Taylor expansion problem in the more general context of \LL, more precisely in the \emph{multiplicative-exponential fragment} $\MELL$ of \LL, being aware that for $\MELL$ 
no coherence relation can characterize the solutions (see below).
Our characterization is constructive, in that it allows us to define a \emph{decision} procedure to solve the inverse Taylor expansion problem in the \emph{finite} case.
A side effect of this investigation is apparently unrelated to the notion of Taylor expansion:
we characterize $\MELL$ formulas that are inhabited by cut-free \MELL proof-structures, so as to prove semi-decidability of the \emph{type inhabitation problem} for cut-free \MELL proof-structures (again, see~below).

\subsubsection*{Proof-nets, proof-structures and their Taylor expansion: seeing trees behind graphs}

In $\MELL$, linearity and the sharp analysis of computations naturally lead to
represent proofs in a more general \emph{graph}-like syntax instead of a
term-like or tree-like one.\footnotemark \footnotetext{A term-like object is
  essentially a tree, with one output (its root) and many inputs (its other
  leaves).}  Indeed, linear negation is involutive and classical duality can be
interpreted as the possibility of juggling between different conclusions,
without a distinguished output \cite{Parigot1992}.  Graphs representing proofs
in \MELL are called \emph{proof-nets}: their syntax is richer and more
expressive than the $\lambda$-calculus (which corresponds to an intuitionistic implicative fragment of \MELL).  
Contrary to $\lambda$-terms, proof-nets
are special inhabitants of the wider land of \emph{proof-structures}.
A proof-structure is any ``graph'' that can be build in the language of proof-nets and it need not represent a proof in \MELL.
Proof-nets can be characterized, among
proof-structures, by abstract (geometric) conditions called correctness criteria
\cite{Girard:1987}.

Proof-structures are well-behaved for performing computations: indeed, 
cut-elimination steps can be defined for proof-structures, and proof-nets
can also be seen as the proof-structures with a good behavior with respect to
cut-elimination \cite{Bechet98}. Furthermore, proof-structures can be
interpreted in denotational models and proof-nets can be characterized among them by semantic means \cite{Retore97}.
It is then natural to attack problems in the general framework of proof-structures. 
In our work, correctness plays no role at all, hence we will
consider proof-structures and not only proof-nets. 
\MELL proof-structures are a
particular kind of graphs, whose edges are labeled by $\MELL$ formulas and
vertices  (aka cells) by $\MELL$ connectives, and for which special subgraphs are
highlighted, the \emph{boxes}, representing the parts of the proof-structure
that can be discarded and copied (\ie called an unbounded number of times) during cut-elimination. 
A box is 
delimited from the rest of a proof-structure by exponential modalities: 
its border is made of one
$\oc$-cell, its principal door, and arbitrarily many $\wn$-cells, its auxiliary
doors.
Boxes are either nested or disjoint (they cannot partially overlap), 
so as to add a tree-like structure to \mbox{proof-structures \emph{aside} from their graph-like nature.}

As in the $\lambda$-calculus, 
one can define box-free \emph{resource} (or $\DiLL_0$) \emph{proof-structures}\footnotemark \cite{DBLP:journals/tcs/EhrhardR06}, where
\footnotetext{Aka differential proof-structures \cite{Carvalho16}, differential nets \cite{DBLP:journals/tcs/EhrhardR06,MazzaPagani07,DBLP:journals/lmcs/Carvalho18}, simple~nets~\cite{PaganiTasson2009}.}%
$\oc$-cells make resources available boundedly,
and the \emph{Taylor expansion} of $\MELL$ proof-structures into these resource proof-structures, that recursively copies the
content of the boxes an arbitrary number of times.
In fact, as somehow anticipated by Boudes \cite{Boudes:2009},
such a Taylor expansion operation can be carried on any tree-like
structure.
This primitive, abstract, notion of Taylor expansion can then be pulled back to
the structure of interest (in this case, the resource proof-structures), as shown in \cite{Wollic} and put forth again here.

\subsubsection*{The question of coherence for proof-structures}

The \emph{inverse Taylor expansion problem} has a natural counterpart 
for \MELL proof-structures: given a set $\Pi$ of resource proof-structures, is there a
$\MELL$ proof-structure the Taylor expansion of which includes $\Pi$? 
Pagani and Tasson \cite{PaganiTasson2009} give the following answer: it is possible to
decide whether a finite set of resource proof-structures is a subset of the
Taylor expansion of a same \MELL proof-structure (and even possible to do it
in nondeterministic polynomial time); but unlike the $\lambda$-calculus, the structure of the relation ``being
part of the Taylor expansion of a same proof-structure'' is \emph{much more}
complicated than a binary (or even $n$-ary) coherence. Indeed, for any $n>1$, it
is possible to find ${n+1}$ resource proof-structures such that any $n$ of them
are in the Taylor expansion of some $\MELL$ proof-structure, but there is no
$\MELL$ proof-structure whose Taylor expansion has all the $n\!+\!1$ as elements
(see our Example \ref{ex:not-coherent} and \cite[pp.~244-246]{Tasson:2009}).

In this work, we introduce a new combinatorial criterion, \emph{\gluability},
for deciding whether a set of resource proof-structures is a subset of the
Taylor expansion of some $\MELL$ proof structure, based on a \emph{rewriting system} on lists of $\MELL$ formulas. Our criterion is more general 
and simpler than the one of \cite{PaganiTasson2009}, which is limited
to the \emph{cut-free} case with \emph{atomic axioms} and characterizes only \emph{finite} sets: we do not have these limitations. 
Akin to \cite{PaganiTasson2009}, our criterion yields a \emph{decision procedure} for the inverse Taylor expansion problem when the set of resource proof-structures given as input is \emph{finite}.
We believe that our criterion is a useful tool for studying proof-structures.
We conjecture that it can be used to show that a binary coherence relation exists for resource proof-structures satisfying a suitable geometric restriction. 
It might also shed light on correctness~and~sequentialization.

As the proof-structures we consider are typed, an unrelated difficulty arises: a
resource proof-structure $\rho$ of type $A$ might not be in the Taylor expansion of any cut-free \(\MELL\)
proof-structure, not because it does not respect the structure imposed by the
Taylor expansion, but because 
there is no cut-free \MELL proof-structure of type $A$, and $\rho$ can 
``mask'' this ``untypeability''.\footnotemark 
\footnotetext{Similarly, in the $\lambda$-calculus, there is no closed $\lambda$-term of type \(X \to Y\) with \(X \neq Y\) atomic, but the resource $\lambda$-term \((\lambda f.f)[\,]\) can be given that type: the empty bag $[\,]$ kills any information on the argument.} 
To solve this issue, we enrich the resource (but not \(\MELL\)) proof-structure 
syntax with a ``universal'' proof-structure: 
a special \(\maltese\)-cell (\emph{daimon}) that can have any number of outputs of any types, 
representing information plainly missing (see \Cref{sec:conclusions} for more details and the way this matter is handled by Pagani and Tasson~\cite{PaganiTasson2009}).

\subsubsection*{Our contribution}

This paper is the long version of \cite{CSL2020} and we keep its structure and main results: the \emph{\gluability} criterion that solves the inverse Taylor expansion problem in \MELL proof-structures (\Cref{thm:characterization}), and the way we prove it, which consists in showing that the Taylor expansion defines a \emph{natural transformation} (\Cref{thm:projection-natural}) from the realm of resource proof-structures to the realm of \MELL proof-structures (see \Cref{sect:outline} for an informal explanation).
With respect to \cite{CSL2020}, the main novelties are:
\begin{enumerate}
	\item Following \cite{Wollic}, we introduce rigorous definitions of all the notions of graph theory (\Cref{sect:graphs}) involved in the definition of proof-structures and Taylor expansion. 
	In this way, we can present here a purely graphical definition of $\MELL$ proof-structures (\Cref{sect:proof-nets}), so as to keep Girard's original intuition of a proof-structure as a graph even in \MELL and to avoid \emph{ad hoc} technicalities to identify the border and the content of a box. 
	This is not only an aesthetic issue but also practical, in that our \MELL proof-structures are manageable: sophisticated operations on them can be easily defined.
	For instance, we give an elegant definition of their Taylor expansion by means of \emph{pullbacks} (\Cref{sect:taylor}).
	
	\item With respect to \cite{CSL2020}, here we use daimons in a different and more limited way. In particular, unlike \cite{CSL2020}, our \MELL proof-structures do not contain any $\maltese$-cell.
	Thus, our results refer to a more standard and interesting definition of \MELL proof-structures, 
	and we deal with less syntactic categories than in \cite{CSL2020}, simplifying the presentation and fixing some technical inaccuracies in a few definitions and lemmas in \cite{CSL2020}.
	
	\item Consequently, our results (notably, naturality and \gluability) 
		are more informative, and allow us to \emph{decide} the inverse Taylor expansion problem  in the \emph{finite} case (\Cref{thm:decisionFinite}). 
	
	\item We solve the apparently unrelated \emph{type inhabitation problem} for cut-free \MELL proof-structures (\Cref{thm:inhabitation}), not considered in \cite{CSL2020}:
	our rewrite system \emph{semi-decides} if, for a list $\Gamma$ of \MELL formulas, there is a cut-free \MELL proof-structure of type $\Gamma$.
	We only provide a semi-algorithm: we can guess a rewriting and check its correctness; but there is no bound on its length. 
	Inhabitation problems are well-studied in many type systems for the $\lambda$-calculus, but 	no results are in the literature for \MELL proof-structures.
	Our contribution may shed some light on the open problem of decidability of~\MELL.
	
	\item Akin to \cite{CSL2020}, to simplify the presentation, we first focus on proof-structures restricted to atomic axioms, then \Cref{sect:nonatomic} shows how to lift our results to the non-atomic case. 
	Compared to \cite{CSL2020}, the lift is more elegant and requires less \mbox{\textit{ad hoc} adjustments.} 
\end{enumerate}

\section{Outline and technical issues}
\label{sect:outline}

\subsubsection*{Rewriting}
The essence of our rewrite system is not located 
in proof-structures but in
lists of \MELL formulas (Definition \ref{def:unwinding-paths}). 
Roughly, the rewrite system is
generated by elementary steps akin to rules of sequent calculus read from the
\emph{bottom up}: they act on a list of conclusions, analogous to a monolaterous
right-handed sequent. 
These steps can be seen as morphisms in a category $\Scheduling$ whose objects are lists of \MELL formulas, and are actually more sequentialized than sequent
calculus rules, as they do not allow for commutation. For instance, the
rule corresponding to the introduction of a $\otimes$ on the $i$\textsuperscript{th} formula, is
defined as
$\otimes_i : (C_1,\dots,C_{i-1},A\otimes B, C_{i+1}, \dots,
C_n) \to (C_1,\dots,C_{i-1},A, B, C_{i+1}, \dots,
C_n)$.

\begin{wrapfigure}[4]{r}{3cm}
	\vspace{-3\baselineskip}
	\scalebox{0.8}{
		\centering
		$
		\vspace{\beforepn}
		\tikzsetnextfilename{images/mell-tensor-intro-1}
		\pnet{
			\pnformulae{
				\pnf[A]{$A$}~\pnf[Ab]{$A^{\bot}$}
			}
			\pnaxiom[ax]{A,Ab}
			\pntensor{A,Ab}[i]{$A \otimes A^{\bot}$}
		}
		\pnrewrite[\otimes_1]
		\tikzsetnextfilename{images/mell-tensor-intro-2}
		\pnet{
			\pnformulae{
				\pnf[A]{$A$}~\pnf[Ab]{$A^{\bot}$}
			}
			\pnaxiom[ax]{A,Ab}
		}$}
\end{wrapfigure}
These rewrite steps then act on $\MELL$ proof-structures, coherently with
their type, by modifying (most of the times, erasing) a cell immediately above the conclusion of the proof-structure. 
Formally, this means that there is a functor $\MELLFunctor$ from $\Scheduling$ to the category $\Rel$ of sets and relations, associating with
every list of \MELL formulas the set of \MELL proof-structures with these conclusions, and with
every rewrite step a relation implementing it (Definition \ref{def:FunctorMELL}). 
The rules \emph{deconstruct} the proof-structure, starting from its conclusions. The rule
$\otimes_1$ acts by removing a $\otimes$-cell on the first conclusion, replacing
it by two conclusions.

These rules can only act on specific proof-structures, and indeed, capture a lot
of their structure: $\otimes_i$ can be applied to a \MELL proof-structure $R$ if and
only if $R$ has a $\otimes$-cell in the conclusion $i$ (as opposed to, say, an
axiom). So, in particular, every proof-structure is completely characterized by
any sequence rewriting it to the empty proof-structure.

\subsubsection*{Naturality}
The same rules also act on sets of resource (aka $\DiLL_0$) proof-structures, defining the
functor $\PPolyPN$ from the category $\Scheduling$ of rewrite steps into the category $\Rel$
(Definition \ref{def:FunctorPPoly}). When carefully defined, the Taylor expansion induces
a \emph{natural transformation} from $\PPolyPN$ to $\MELLFunctor$
(\Cref{thm:projection-natural}). By applying this naturality repeatedly, we get
our characterization (\Cref{thm:characterization}): a set of resource
proof-structures $\Pi$ is a subset of the Taylor expansion of a $\MELL$
proof-structure if and only if $\Pi$ is \emph{\gluable}, that is, there is a sequence rewriting $\Pi$ to the singleton of
the~\emph{empty}~proof-structure.

The naturality property is not only a mean to obtain our characterization, but also an interesting result in itself: natural transformations can often be used to express fundamental properties in a mathematical context. 
In this case, the \emph{Taylor expansion is natural} with respect to the
possibility of building a (\(\MELL\) or resource) proof-structure by adding a
cell to its conclusions or boxing it.
Said differently, naturality of the Taylor expansion roughly means that the rewrite rules that deconstruct a \MELL proof-structure $R$ and a set of resource proof-structures in the Taylor expansion of $R$ mimic each other.

\subsubsection*{Quasi-proof-structures and mix}
\label{subsect:intro-quasi-proof-structure}

Our rewrite rules consume proof-structures from their conclusions. The rule
corresponding to boxes in $\MELL$ opens a box by deleting its principal door (a
$\oc$-cell) and its border, while for a resource proof-structure  
it deletes a $\oc$-cell and separates the different copies of the content of the box (possibly) represented by such a $\oc$-cell. 
This operation is problematic in a twofold way.
In a resource proof-structure, where the border of boxes is not marked, it is not clear how to identify such copies.
	On the other side, in a $\MELL$ proof-structure the content of a box is not to be treated as if it were at the same level as what is outside of the box: 
it can be copied many times or erased, while what is outside boxes cannot,
and treating the content in the same way as the outside suppresses this
distinction, which is crucial in \LL.
So, we need to remember that the content of a box, even if it is at depth $0$ (\ie not contained in any other box) after erasing the box wrapping it by means of our rewrite rules, 
is not to be mixed with the rest of the structure  
at~depth~$0$. 

\begin{wrapfigure}[3]{r}{1.75cm}
	\vspace{-\baselineskip}
	\scalebox{0.8}{
		\centering
		$
		\vspace{\beforepn}
		\tikzsetnextfilename{images/root}
		\pnet{
			\pnsomenet[u]{$\pi$}{1.25cm}{0.5cm}
			\pnoutfrom{u.-140}[pi1c]{$\qquad \ \!\!\cdots$}
			\pnoutfrom{u.-40}[u2a2]{$\quad$}
			\pnsemibox{u}
		}$}
\end{wrapfigure}

In order for our proof-structures to provide this information,
we need to generalize them and consider that a proof-structure can have 
a \emph{forest} of boxes, instead of just a tree: this yields the notion of
\emph{quasi-proof-structure} (Definition \ref{def:quasi-proof-structure}).
In this way, according to our rewrite rules, opening a box by deleting its principal door amounts to taking a
box in the tree and disconnecting it from its root, creating a new tree. We draw
this in a quasi-proof-structure by surrounding each \emph{component}---made of the objects having the same root---with a dashed line,  
open from the bottom, remembering the phantom presence of the border of the box,
even if it was erased. 
This allows one to open the box only when it is alone in a component, just surrounded by a dashed line 
(Definition~\ref{def:unwinding}).

This is not merely a technical remark, as this generalization gives a status to
the $\Mix$ rule of \LL: indeed, mixing two proofs amounts to taking two proofs
and considering them as one, without any other modifications. Here, it amounts to
taking two proof-structures, each with its box-tree, and considering them as one by
merging the roots of their trees (see the mix step in Definition \ref{def:unwinding}). 
We embed this design decision up to the level of formulas, which are segregated
in different zones that have to be mixed before interacting. 
Indeed, our rewrite rules actually act on conclusions arranged as a \emph{list of lists} of \MELL formulas.

\subsubsection*{Geometric invariance and emptiness: the filled Taylor expansion}
\label{subsec:intro-fattened}
The use of forests instead of trees for the nesting structure of boxes, where
the different roots are thought of as the contents of long-gone boxes, has an
interesting consequence in the Taylor expansion: indeed, an element of the
Taylor expansion of a proof-structure contains an arbitrary number of copies of
the contents of the boxes, in particular \emph{zero}. If we think of the part at
depth $0$ of a $\MELL$ proof-structure as inside an invisible box, its content
can be deleted in some elements of the Taylor expansion just as any other
box.\footnotemark
\footnotetext{The dual case, of copying the contents of a box, poses no problem
  in our approach.}
As erasing completely the conclusions would
cause the Taylor expansion not preserve the conclusions (which would lead to
technical complications), we introduce the \emph{filled Taylor expansion}
(Definition \ref{def:FilledTaylor}), which contains not only the elements of the usual
Taylor expansion, but also elements of the Taylor expansion where one component
has been erased and replaced by a \(\maltese\)-cell (\emph{daimon}), representing lack of information, apart from the number and types of the conclusions.
Roughly, a $\maltese$-cell is a placeholder for any $\DiLL_0$ proof-structure of given conclusions.

\subsubsection*{\Gluability and cut-free \gluability}
\label{subsec:cut-vs-cut-free-glueab}
Our \gluability criterion is based on local rewritings, which yields a geometric and modular approach to the inverse Taylor expansion problem. From the geometric point of view, there is nothing special in the elementary rewrite step corresponding to the cut rule, and it is natural to prove our \gluability criterion in presence of cuts (\Cref{thm:characterization}.\ref{p:characterization-with-cut}). 
But from the proof-theoretical point of view, the gulf separating cut-free proof-structures from the others shows up: for every \MELL formula $A$, the set of $\MELL$ proof-structures with conclusion of type $A$ is never empty, while this might very well be the case if we restrict to cut-free $\MELL$ proof-structures (see Example~\ref{ex:glueable-non-cut-free} and Remark~\ref{rmk:non-cut-free}). The modularity of our proofs allows to straightforwardly adapt the criterion to the cut-free case (\Cref{thm:characterization}.\ref{p:characterization-cut-free}) and thus to state correctly (and easily solve) the \emph{type inhabitation problem} for (cut-free) $\MELL$ proof-structures (\Cref{thm:inhabitation}).

\subsubsection*{Outline}
\Cref{sect:graphs} recalls some preliminary notions on graph theory.
In \Cref{sect:proof-nets} we define (\MELL and $\DiLL_0$) proof-structures and quasi-proof-structures with atomic axioms.
\Cref{sect:taylor} defines the notion of Taylor expansion.
\Cref{sect:rules} introduces the rewriting rules on lists of lists of formulas and lifts them to \MELL quasi-proof-structures via the functor $\MELLFunctor$.
In \Cref{sec:naturality} we lift the rewriting rules to $\DiLL_0$ quasi-proof-structures via the functor $\PPolyPN$ and we show our first main result: the Taylor expansion induces a natural transformation between the two functors.
\Cref{sec:glueable} proves two other main results: the solution of the inverse Taylor expansion problem (via a \gluability criterion) and the solution of the  type inhabitation problem in \MELL.
\Cref{sec:decidability} proves that the inverse Taylor expansion problem is decidable in the finite case. 
In \Cref{sect:nonatomic} we show how to adapt our method when axioms are not necessarily atomic.
\Cref{sec:conclusions} concludes with some final remarks.


\section{Preliminaries on graphs}
\label{sect:graphs}

\subsubsection*{Graphs with half-edges.}
There are many formalizations of the familiar notion of graph. 
Here we adopt the one due to \cite{Borisov2008}:\footnotemark
\footnotetext{The folklore attributes the definition of graphs
	with half-edges to Kontsevitch and Manin, but the idea can actually be traced
	back to Grothendieck's \emph{dessins d'enfant}.}
a graph is still a set of edges and a set of vertices,
but edges are now split in halves, allowing some of them to be hanging.
Splitting every edge in two has at least four features of particular interest to represent \LL proof-structures:
\begin{itemize}
	\item two half-edges are connected by an involution, thus defining an edge linking two vertices (possibly the same vertex). 
	The fixed points of this involution are ``hanging'' edges, 
	linked to a vertex only on one endpoint:
	they are well suited for representing the conclusions of a proof-structure.
	In this way it is also easy to define some intuitive but formally tricky operations such as 
	grafting/substituting a graph into/for another graph~(see~Example~\ref{ex:graph});
	
	\item given any vertex $v$ in a graph $\tau$, it is natural to
	define the \emph{corolla} of $v$, that is $v$ itself with 
	the half-edges linked to it; $\tau$ is the union of its corollas, glued together by the involution;
	
	\item while studying proof-structures, it is often necessary to treat
	them both as directed and undirected graphs. 
	With this definition of graph, an orientation, a labeling and a coloring, are structures on top of the structure of the undirected graph (see Definition \ref{def:labelled-graphs});
	\item this definition of graph allows 
	a uniform syntax to represent both proof-structures and other structures of interest (\eg, the box-tree of a proof-structure). 
	In this way, we avoid appealing to \textit{ad hoc} conditions in the definitions of proof-structure and Taylor expansion, which are then more compact and rely only on notions from graph theory.
\end{itemize}

\begin{defi}[graph]
\label{def:graph}
	\index{graph}
	\index{flag}
	\index{vertex}
	\index{tail}
	\index{edge}
	A (finite\footnotemark)
	\footnotetext{The finiteness condition on $F_\tau$ and $V_\tau$ can be dropped, so as to allow for possibly infinite graphs. We require it because we only deal with finite graphs.} 
	\intro{graph}  is a quadruple
	$\tau = (F_{\tau}, V_{\tau}, \partial_{\tau}, j_{\tau})$, where
	\begin{itemize}
		\item $F_{\tau}$ is a finite set, whose elements are called \intro{flags} of
		$\tau$;
		\item $V_{\tau}$ is a finite set, whose elements are called \intro{vertices} of
		$\tau$;
		\item $\partial_{\tau} \colon F_{\tau} \to V_{\tau}$ is a function associating with
		each flag its \emph{endpoint}; 
		\item $j_{\tau} \colon F_{\tau} \to F_{\tau}$ is an involution, \ie $(j_{\tau} \circ j_{\tau}) (f) = f$ for every $f \in F_{\tau}$.
	\end{itemize}
 
	The graph $\tau  = (F_{\tau}, V_{\tau}, \partial_{\tau}, j_{\tau})$ is \intro[empty graph]{empty} if $V_{\tau} = \emptyset$.\footnotemark
	\footnotetext{This implies that $\partial_{\tau}$ is the empty function and $F_{\tau} = \emptyset$ (since $F_\tau$ is the domain of $\partial_\tau$).}
	
	A \intro{subgraph} of a graph $\tau = (F_{\tau}, V_{\tau}, \partial_{\tau}, j_{\tau})$ is a graph $\sigma = (F_{\sigma}, V_{\sigma}, \partial_{\sigma}, j_{\sigma})$ where $F_\sigma \subseteq F_{\tau}$, $V_\sigma \subseteq V_{\tau}$, $\partial_\sigma = \partial_{\tau} \mathord{\upharpoonright}_{F_\sigma}$, and, for all $f \in F_\sigma$, $j_\sigma(f) = j_{\tau}(f)$ if $j_{\tau}(f) \in F_\sigma$, otherwise $j_\sigma(f) = f$. 
\end{defi}

In a graph $\tau  = (F_{\tau}, V_{\tau}, \partial_{\tau}, j_{\tau})$, a \intro{tail} of $\tau$ is a fixed point of the involution $j_\tau$, \ie~a flag $f = j_\tau(f)$.
A set $\{f,f'\}$ of two flags with $j_\tau(f) = f' \neq f$ is an \intro{edge} of
$\tau$ between vertices $\partial_{\tau}(f)$ and $\partial_{\tau}(f')$; 
$f, f'$ are the \intro{halves} of the edge;
if $\partial_{\tau}(f) = \partial_{\tau}(f')$ then the edge is a \intro{loop}. 

Given two graphs $\tau$ and $\tau'$, it is always possible to consider their
disjoint union $\tau \sqcup \tau'$ defined as the disjoint union of the
underlying sets and functions.

\index{corolla}
A one-vertex graph with set of flags $F$ and 
involution the identity function $\id_F$ on $F$ is called a \intro{corolla} 
(the endpoint of each flag is the only vertex); it is usually denoted by $\ast_F$. 

Given a graph $\tau = (F_{\tau}, V_{\tau}, \partial_{\tau}, j_{\tau})$, a vertex
$v \in V_{\tau}$ defines a corolla 
$\tau_v = (F_v, \{v\}, \allowbreak \partial_{\tau}|_{F_v}, \allowbreak\id_{F_v})$ where $F_v = \partial_{\tau}^{-1}(v)$.
A graph $\tau$ can be seen as the disjoint union of the corollas of its vertices, and the involution gluing some flags (the ones that are not tails in $\tau$) in edges.

\begin{defi}[graph morphism and isomorphism]
	\label{def:graph-morphism}
	Let $\tau, \sigma$ be two graphs. A \intro{graph morphism}
	$h \colon \tau \to \sigma$
	from $\tau$ to $\sigma$ is a couple of functions $(h_F \colon F_{\tau} \to F_{\sigma}, h_V \colon V_{\tau} \to
	V_{\sigma})$ such that $h_V \circ \partial_{\tau} = \partial_{\sigma} \circ
	h_F$ and $h_F \circ j_{\tau} = j_{\sigma} \circ h_F$ ($h_F$ and $h_V$ are the \emph{components} of $h$).

A graph morphism is \intro[empty graph morphism]{empty} (\resp\ an \intro[identity graph morphism]{identity}) if 
its components are empty (\resp\ identity) functions.
A \intro{graph isomorphism} is a graph morphism whose components are bijective.
\end{defi}

Intuitively, a graph morphism preserves tails and edges.
The category $\Graph$ \index{$\Graph$} has graphs as objects and graph morphisms as arrows:
indeed, graph morphisms compose in a associative way (by composing their components) and identity graph morphisms are neutral for such a composition. 
$\Graph$ is a monoidal category, with disjoint union as a monoidal~product.

\subsubsection*{Graphs with structure.}

Some structures can be put on top of a \kl{graph}. 
\begin{defi}[structured graph]
	\label{def:labelled-graphs}
	Let $\tau = (F_\tau, V_\tau, \partial_\tau, j_\tau)$ be a graph.
	\begin{itemize}
		\item A \intro{labeled graph} $(\tau,\ell_\tau)$ \emph{with labels in $L$} is a graph
		$\tau$ 
		and a function $\ell_\tau \colon V_{\tau} \to L$.

		\item A \intro{colored graph} $(\tau,\FlagType_\tau)$ \emph{with colors in $C$} is a graph $\tau$ and a function $\FlagType_\tau \colon F_{\tau} \to C$ such that $\FlagType_\tau(f) = \FlagType_\tau(f')$ for the two halves $f,f'$ of any edge of $\tau$.

		\item A \intro{directed graph} $(\tau, \Orient_\tau)$ is a graph $\tau$
		and a function $\Orient_\tau \colon F_{\tau} \to \{\In, \Out\}$ such that $\Orient_\tau(f) \neq \Orient_\tau(f')$ for the two halves $f,f'$ of any edge of $\tau$.
		If $\Orient_\tau(f) = \Out$ and $\Orient_\tau(f') = \In$, $\{f,f'\}$ is said an \emph{edge} of $\tau$ \emph{from} $\partial_{\tau}(f)$ \emph{to} $\partial_{\tau}(f')$;
		$\In$-oriented (\resp $\Out$-oriented) tails of $\tau$ are called \intro{inputs} (resp.~\intro{outputs}) \emph{of $\tau$};
		if $v$ is a vertex of $\tau$, 
		the \emph{inputs} (\resp \emph{outputs}) \emph{of $v$} are the elements of the set
		$\In_{\tau}(v) = \partial_{\tau}^{-1}(v) \cap \mathsf{o}_\tau^{-1}(\In)$
		(resp.~$\Out_{\tau}(v) = \partial_{\tau}^{-1}(v) \cap \mathsf{o}_\tau^{-1}(\Out)$).
		\item An \intro{ordered graph} $(\tau,<_{\tau})$ is a graph $\tau$ together
		with an order $<_{\tau}$ on its flags.
	\end{itemize}
\end{defi}
Different structures on a graph $\tau$ can combine modularly, for instance~$\tau$ can be endowed with a labeling $\VertType_{\tau}$ and an orientation $\Orient_{\tau}$, so as to form a directed labeled graph $(\tau, \Orient_{\tau}, \VertType_{\tau})$; and conversely, a directed labeled graph $(\tau, \Orient_{\tau}, \VertType_{\tau})$ can be seen as a directed graph $(\tau, \Orient_{\tau})$, forgetting the labeling $\VertType_{\tau}$.
Roughly, a graph is labeled (\resp~colored) when labels are associated with its vertices (\resp edges and ``hanging'' edges).
In a directed graph, an input (\resp~output) of a vertex $v$ is a---half or hanging---edge incoming in (\resp~outgoing~from)~$v$.

Graphs can be depicted in diagrammatic form, vertices are represented by small rounded boxes and flags by wires touching their endpoint.
As a graph is just a
disjoint union of corollas glued by the involution, we only need to 
depict corollas (\Cref{fig:corolla5}, on the left) and place the two halves of an edge next to each other (\Cref{fig:corolla5}, on the right).
In directed graphs, inputs of a corolla are depicted above the corolla,
outputs below; arrows also show their orientation. 
The color of a flag $f$ (if any) is written next to $f$.
The label of a vertex $v$ (if any) is written inside $v$.
If ordered, flags of a corolla are depicted increasing from left~to~right.
A dotted gray wire denote an arbitrary number (possibly $0$) of flags (see \Cref{fig:actions-all,fig:actions-daimon}).

\begin{exa}
\label{ex:graph}
	The directed labeled colored ordered \kl{corolla} $\mathbf{5} = (\ast_{\mathbf{5}}, \Orient_\mathbf{5}, \VertType_\mathbf{5}, \FlagType_\mathbf{5}, <_{\mathbf{5}})$ depicted in \Cref{fig:corolla5} (on the left) has $\ast$ as its only vertex and $F_\mathbf{5} = \{0, 1, 2, 3, 4\}$ as its set of flags;
	it is endowed with the order $0 <_{\mathbf{5}} 4$ and $1 <_{\mathbf{5}} 2 <_{\mathbf{5}} 3$, 
the labeling 
$\VertType_\mathbf{5}(\ast) = \maltese$, the orientation $\Orient_\mathbf{5} \colon F_\mathbf{5} \to \{\In, \Out\}$ defined by
		 $\Orient_\mathbf{5}(0) = \Orient_\mathbf{5}(4) = \Out$ and $\Orient_\mathbf{5}(1) = \Orient_\mathbf{5}(2) = \Orient_\mathbf{5}(3) = \In$,
the coloring $\FlagType_\mathbf{5} \colon F_\mathbf{5} \to \{A_0,\ldots,A_4\}$  defined by $\mathsf{c}(i) = A_i$ for all $i \in F_\mathbf{5}$.

	Let $\sigma_{\ax}$ be the directed labeled colored corolla, whose only vertex is labeled by $\ax$, and whose only flags are the outputs $5$ (colored by $A_2$) and $6$ (colored by $A_3$).
	The directed labeled colored ordered two-vertex graph $\rho$ 
	in \Cref{fig:corolla5} (on the right) is  
	obtained by ``grafting'' $\sigma_{\ax}$ into $\mathbf{5}$, \ie~from $\sigma_{\ax}$ and $\mathbf{5}$	by defining the involution $j_{\rho} \colon \{0, \dots, 6\} \to \{0, \dots, 6\}$ as $j_{\rho}(i) = j_{\mathbf{5}}(i)$ for $i \in \{0,1,4\}$, and $j_{\rho}(i) = i+3$ for $i \in \{2,3\}$, and $j_{\rho}(i) = i-3$ for $i \in \{5,6\}$.
\end{exa}

Each enrichment of the graph structure introduced in Definition \ref{def:labelled-graphs}
induces a notion of structure-preserving morphism (and isomorphism) that extends the notion of \kl{graph morphism} (and \kl{graph isomorphism}) seen in Definition \ref{def:graph-morphism}, so as to form an associated category. 
Moreover, different kinds of structure preservation can be combined in a modular way.
For instance, given two directed labeled graphs $(\tau, \Orient_\tau, \VertType_\tau)$ and $(\sigma, \Orient_\sigma, \VertType_\sigma)$ both with labels in $L$, 
(in particular, they can be seen as two directed graphs $(\tau, \Orient_\tau)$ and $(\sigma, \Orient_\sigma)$, forgetting labels), 
a \intro{directed graph morphism} $h \colon (\tau, \Orient_\tau) \to (\sigma, \Orient_\sigma)$ is a \kl{graph morphism} $h =(h_F, h_V) \colon \tau \to \sigma$ such that $\Orient_\sigma \circ h_F = \Orient_\tau$;
this means that $h_F$ maps input (\resp\ output) flags of $\tau$ to input (\resp\ output) flags of $\sigma$.
And a \emph{directed labeled isomorphism} $h \colon (\tau, \Orient_\tau, \VertType_\tau) \to (\sigma, \Orient_\sigma, \VertType_\sigma)$ is a \kl{graph isomorphism} $h =(h_F, h_V) \colon \tau \to \sigma$ such that $\Orient_\sigma \circ h_F = \Orient_\tau$ and $\VertType_\sigma \circ h_V = \VertType_\tau$.

\begin{figure}[!t]
	\centering
	\vspace*{-3.5\baselineskip}
		$\mathbf{5} \ = \ $
		\scalebox{0.8}{
		\pnet{
			\pnsomenet[t]{\small$\maltese$}{1.6cm}{0.3cm}
			\pnoutfrom{t.-145}{$A_0$}
			\pnoutfrom{t.-35}{$A_4$}
			\pninto{t.155}{$A_1$}
			\pninto{t.90}{$A_2$}
			\pninto{t.25}{$A_3$}
		}}
		\qquad\qquad\qquad\qquad
		$\rho \ = \ $
		\scalebox{0.78}{
		\pnet{
			\pnsomenet[t]{\small$\maltese$}{1.6cm}{0.2cm}
			\pnoutfrom{t.-145}[lo]{$A_0$}
			\pnoutfrom{t.-35}[ro]{$A_4$}
			\pninto{t.155}[li]{$A_1$}
			\pninto{t.90}[i]{$A_2$}
			\pninto{t.25}[ri]{$A_3$}
			\pnaxiom{i,ri}[.8]
		}
	}
	\caption{A directed labeled colored ordered corolla $\mathbf{5}$ (left), and a directed labeled colored ordered two-vertex graph $\rho$ (right), see Example~\ref{ex:graph}.}
	\label{fig:corolla5}
\end{figure}

\subsubsection*{Trees and paths.}
\label{subsec:trees}

An \intro{undirected path} on a graph $\tau$ is a finite sequence of flags $\varphi = (f_1, \dots, f_{2n})$ for some $n \in \Nat$ such that, for all
$1 \leqslant i \leqslant n$, $j_\tau(f_{2i-1}) = f_{2i} \neq f_{2i-1}$ and (if $i \neq n$) $\partial_\tau(f_{2i}) = \partial_\tau (f_{2i+1})$ with $f_{2i} \neq f_{2i+1}$. 
We say that $\varphi$ is \emph{between $\partial_\tau(f_1)$ and $\partial_\tau(f_{2n})$} if $n > 0$ (and it is a \intro{cycle} if moreover $\partial_\tau(f_1) =\partial_\tau(f_{2n})$), otherwise it is the \emph{empty} (\emph{undirected}) \emph{path}, which is between any vertex and itself; the \intro[path length]{length} of $\varphi$ is $n$.
Note that $\varphi$ cannot cross any \kl{tail} of $\tau$.
Two vertices in a graph are \emph{connected} if there is an undirected path between them.

Let $\tau$ be a \kl{graph}: $\tau$ is \intro{connected} if any vertices $v, v' \in V_\tau$ are connected;
a \intro{connected component} of $\tau$ is a maximal (with respect to the inclusion of flags and vertices) connected \kl{subgraph} of $\tau$;
$\tau$ is \intro{acyclic} (or a \emph{\kl{forest}}) if 
it has no \kl{cycles};
$\tau$ is a \intro{tree} if it is a connected~forest.
Note that a \kl{forest} can be seen as a non-empty finite sequence of \kl{trees}, each tree is a \kl{connected component};
so, an empty forest can be seen as a non-empty finite sequence of \kl[empty graph]{empty} trees.

A \intro{rooted tree} $\tau$ is a \kl[empty graph]{non-empty} directed tree such that each vertex has exactly one output. 
So, by finiteness, $\tau$ has exactly one output tail $f$: the endpoint of $f$ 
is called the \intro{root} of $\tau$. 

\begin{rem}
	Let $\tau$ and $\tau'$ be two rooted trees, and $h \colon \tau \to \tau'$ be 
	a directed graph morphism.
	As $h_F$ preserves tails and orientation, $h_V$ maps the root of $\tau$ to the root of $\tau'$. 
	Rooted trees and directed graph morphisms over them form a category $\RootedTree$.
\end{rem}

A \emph{directed path} on a directed graph $\tau$ is an undirected \kl{path} $\varphi = (f_1, \dots, f_{2n})$ for some $n \in \Nat$ where $f_{2i-1}$ is output and $f_{2i}$ is input for all $1 \leqslant i \leqslant n$.
We say that $\varphi$ is \emph{from $\partial_\tau(f_1)$ to $\partial_\tau(f_{2n})$} if $n > 0$; otherwise it is the \emph{empty} (directed) \emph{path},
from any vertex~to~itself. 

The set of directed paths on a directed tree $\tau$ is finite. 
As such, we define the \intro{reflexive-transitive closure}, or \emph{free category},
$\tau^{\circlearrowleft}$ of $\tau$ as the directed graph with same vertices and same tails as $\tau$, and with an edge from $v$ to $v'$ for any directed path from $v$ to $v'$ in $\tau$.
The operator $(\cdot)^\circlearrowleft$ lifts to a functor from the category $\RootedTree$ to the category of directed graphs,
so for every morphism $f$ in $\RootedTree$, $f^\circlearrowleft$ is a directed graph morphism between directed graphs.

\subsubsection*{Pullback in the category of graphs}
\label{app:pullback}

The category of graphs has all pullbacks, a fact that we use extensively. We
recall here all the 
definitions and facts involved in that affirmation.

\begin{defi}[pullback]
	Let \(\catC\) be a category. Let \(X\), \(Y\), \(Z\) be  objects of
	\(\catC\), and $f \colon X \to Z$ and $g \colon Y \to Z$ be arrows of $\catC$.
	A \intro{pullback} of 
	$f$ and $g$ is the triple $(P,\oc_X,\oc_Y)$ where $P$ is an object of $\catC$ and $\oc_X \colon P \to X$ and $\oc_Y \colon P \to Y$ are arrows of $\catC$ such that diagram \eqref{eq:pullback}
	commutes and, for any $(Q, h \colon Q \to X, k \colon Q \to Y)$ making the same
	diagram commute, there is a unique arrow $u \colon Q \to P$ factorizing $h$ and $k$, \ie such that diagram \eqref{eq:pullback-universal} commutes.
\end{defi}
\begin{minipage}{0.25\textwidth}
	\begin{align}\label{eq:pullback}
	\begin{tikzcd}
	P \ar[r,"\oc_X"] \ar[d,swap,"\oc_Y"]
	\& 
	X 	\ar[d,"f"]
	\\
	Y  \ar[r,"g"]
	\& 
	Z
	\end{tikzcd}
	\end{align}
\end{minipage}
\quad
\begin{minipage}{0.35\textwidth}
	\begin{align}\label{eq:pullback-universal}
	\begin{tikzcd}
	Q
	\arrow[bend left]{drr}{h}
	\arrow[bend right,swap]{ddr}{k}
	\arrow[dashed]{dr}
	{u} \& \& 
	\\[-15pt]
	\& P \arrow{r}{\oc_X} \arrow{d}[swap]{\oc_Y}
	\& X \arrow{d}{f} 
	\\
	\& Y \arrow{r}{g}
	\& Z
	\end{tikzcd}
	\end{align}
\end{minipage}
\quad
\begin{minipage}{0.3\textwidth}
	\begin{align}\label{eq:pullback-corner}
	\begin{tikzcd}
	X \!\times_Z \! Y \ar[r,"\oc_X"] \ar[d,swap,"\oc_Y"]
	\arrow[dr, phantom, "\scalebox{1.5}{$\lrcorner$}" , very near start, color=black]
	\& 
	X 	\ar[d,"f"]
	\\
	Y  \ar[r,"g"]
	\& 
	Z
	\end{tikzcd}
	\end{align}
\end{minipage}

A pullback $(P, \oc_X, \oc_Y)$ of $f \colon X \to Z$ and $g \colon Y \to Z$ is unique (up to unique isomorphism); $P$ is usually denoted by $X \times_Z Y$ (leaving $f,g$ implicit) and depicted like in diagram \eqref{eq:pullback-corner}.

All pullbacks exist in the category $\Graph$ of graphs. Let us show it explicitly.
Let \(\tau = (F_{\tau}, V_{\tau}, \partial_{\tau}, j_{\tau})\), \(\sigma = (F_{\sigma},
V_{\sigma}, \partial_{\sigma}, j_{\sigma})\) and \(\rho = (F_{\rho}, V_{\rho},
\partial_{\rho}, j_{\rho})\) be graphs, and let $g = (g_F, g_V) \colon \sigma \to \tau$ and $h = (h_F, h_V) \colon \rho \to \tau$ be graph morphisms.
Consider the sets
\begin{align*}
F &= \left\{ (f_\sigma, f_\rho) \in F_{\sigma} \times F_{\rho} \mid g_F(f_\sigma) = h_F(f_\rho)
\right\}
& \
V &= \left\{ (v_\sigma, v_\rho) \in V_{\sigma} \times V_{\rho} \mid g_V(v_\sigma) = h_V(v_\rho)
\right\}.
\end{align*}
They are both equipped with their natural projections
$\pi_{\sigma}^F \colon F \to F_\sigma$, $\pi_{\rho}^F \colon F \to F_\rho$ and $\pi_{\sigma}^V \colon V \to V_\sigma$, $\pi_{\rho}^V \colon V \to V_\rho$.
Let $f\in F$.
\begin{align*}
(g_V \circ \partial_{\sigma}\circ \pi_{\sigma}^F) (f)
&= (\partial_{\tau} \circ g_F \circ \pi_{\sigma}^F) (f)
\quad\text{because \(g\) is a graph morphism}\\ 
&= (\partial_{\tau} \circ h_F \circ \pi_{\rho}^F) (f) \quad\text{by definition of
	\(F\)}\\
&= (h_V \circ \partial_{\rho} \circ \pi_{\rho}^F)(f) \quad\text{because \(h\) is a
	graph morphism.}
\end{align*}
Hence, we can define $\partial \colon F \to V\) by \(\partial (f) =
((\partial_{\sigma}\circ \pi_{\sigma}^F) (f), (\partial_{\rho} \circ
\pi_{\rho}^F)(f))$. In the same way, we define $j \colon F \to F\) by \(j(f) =
((j_{\sigma} \circ \pi_{\sigma}^F)(f), (j_{\rho} \circ \pi_{\rho}^F)(f))$, and
check that it is an involution.

Thus, \(\sigma \times_{\tau} \rho = (F,V,\partial,j)\) is a graph\footnotemark
\footnotetext{Note that if $\sigma$ or $\rho$ is the \kl{empty graph}, then so is $\sigma \times_{\tau} \rho$, and $\pi_\sigma$ and $\pi_\rho$ are \kl{empty graph morphisms}.}
and \(\pi_{\sigma} = (\pi_{\sigma}^F, \pi_{\sigma}^V) \colon \sigma \times_{\tau} \rho \to \sigma\) and
\(\pi_{\rho} = (\pi_{\rho}^F, \pi_{\rho}^V) \colon \sigma \times_{\tau} \rho \to \rho\) are graph morphisms  such that diagram \eqref{eq:pullback-graph} below commutes.
\begin{center}
\begin{minipage}{0.3\textwidth}
	\begin{align}\label{eq:pullback-graph}
		\begin{tikzcd}
		\sigma \!\times_{\tau}\! \rho \ar[r,"\pi_\sigma"] \ar[d,swap,"\pi_\rho"]
		\& 
		\sigma  \ar[d,"g"]
		\\
		\rho 	\ar[r,"h"]
		\& 
		\tau
		\end{tikzcd}
	\end{align}
\end{minipage}
\qquad\qquad
\begin{minipage}{0.3\textwidth}
	\begin{align}\label{eq:pullback-graph-universal}
	\begin{tikzcd}
	\mu \ar[r,"p"] \ar[d,swap,"q"]
	\& 
	\sigma  \ar[d,"g"]
	\\
	\rho 	\ar[r,"h"]
	\& 
	\tau
	\end{tikzcd}
	\end{align}
\end{minipage}
\end{center}

\noindent Consider now any graph $\mu = (F_{\mu},V_{\mu},\partial_{\mu},j_{\mu})$ and graph morphisms $p = (p_F, p_V) \colon \mu \to \sigma$ and $q = (q_F, q_V) \colon \mu \to \rho$ such that diagram \eqref{eq:pullback-graph-universal} above
commutes. For every $f\in F_{\mu}\), let \(r_F(f) = (p_F(f),q_F(f))$, and for every
\(v\in V_{\mu}\), let $r_V(v) = (p_V(v),q_V(v))$. 
It is easy to check that it defines a
graph morphism $r = (r_F, r_V) \colon \mu \to \sigma \times_{\tau} \rho$ that is the unique one to factorize $p$ and $q$.
Therefore, $(\sigma \times_{\tau} \rho, \pi_\sigma, \pi_\rho)$ is the pullback of $g$ and $h$.

Roughly, the pullback $\sigma \times_\tau \rho$ is obtained as a sort of ``lax intersection'' of the ``similar'' vertices of $\sigma$ and $\rho$ (they are ``similar'' when they are sent to the same vertex in $\tau$, by $g_V $ and $h_V$ respectively), and keeping whatever flags that are in the ``intersection''. 
So, $\sigma \times_\tau \rho$ is the maximal graph that is ``compatible'' with both $\sigma$ and $\rho$.

The pullback construction lifts to directed, labeled, colored and ordered graphs.

\section{\MELL proof-structures and quasi-proof-structures}
\label{sect:proof-nets}

	We present here a purely graphical definition of $\MELL$ proof-structures,
following the non-inductive approach of \cite{Wollic}: a \MELL proof-structure $R$ is essentially a labeled directed graph $|R|$ together with some additional information to identify its boxes.
Our definition is completely based on standard notions (recalled in \Cref{sect:graphs}) from graph theory: it is formal (with an eye towards complete computer formalization) but avoids \emph{ad hoc} technicalities to identify boxes. 
Indeed, the inductive and ordered structure of the boxes of $R$ is recovered by means of a tree $\TreeT_{R}$ and a graph morphism  $\BoxFunction_{R}$ from $|R|$ to $\TreeT_{R}$ which allows us to recognize the content and the border of all boxes in $R$.
We use $n$-ary vertices of type $\wn$ collapsing weakening, dereliction and contraction (like in \cite{DanosRegnier95}). 
In this way, we get a \emph{canonical} representation of \MELL proof-structures with respect to operations like associativity and commutativity of contractions, or neutrality of weakening with respect to contraction.

One of the benefits of our purely graphical definition of proof-structure is that notions involving proof-structures such as isomorphism (see below), Taylor expansion (see \Cref{sect:taylor}) and correctness graph (see \cite{Wollic}) can be defined rigorously by means of standard (and not \textit{ad hoc}) mathematical tools from graph theory, making them more ``manageable''.

\subsection{\texorpdfstring{$\MELL$}{MELL} formulas and lists}
\label{subparagraph:Formulas}

Given a countably infinite set of propositional variables $X, Y, Z, \dots$,
\intro{\(\MELL\) formulas} are defined by the following inductive grammar:
\begin{equation*}
  \label{eq:formula}
  A,B \Coloneqq X \mid X^\bot \mid \one \mid \bot \mid A
  \otimes B \mid A \parr B \mid \oc A \mid \wn A
  \qquad \text{(set: } \Formulas_{\MELL}\text{)}
\end{equation*}

\intro{Linear negation} $(\cdot)^\bot$ is defined via De Morgan laws $\one^\bot = \bot$,
$(A \otimes B)^\bot = A^\bot \parr B^\bot$ and $(\oc A)^\bot = \wn A^\bot$, so as to
be involutive, \ie{} $A^{\bot\bot} = A$ for any $\MELL$ formula $A$ (we say that $A$ and $A^\bot$ are \intro{dual}).
Variables and their negations are \intro{atomic} formulas; $\otimes$ and $\parr$ (resp.~$\oc$ and $\wn$) are \emph{multiplicative} (resp.~\intro{exponential}) \intro{connectives}; $\one$ and $\bot$ are  (\emph{multiplicative}) \intro{units}.

We call any finite sequence a \emph{list}.
The \intro{empty list} is denoted by~$\emptylist$.
A list $(x_1, \dots, x_n)$ with $n \!>\! 0$ and $x_i \!= \emptylist$ for all $1 \!\leqslant\! i \!\leqslant\! n$ is denoted by $n \emptylist$ or $\emptylists$.  
	Note that a list $\emptylists$ is not empty.

Let $m,n \in \Nat$ and $\{i_j\}_{0 \leqslant j \leqslant n} \subseteq \Nat$ with $i_j < i_{j+1}$, $i_0 = 0$, $i_n = m$.
Let $A_1, \dots, A_m$ be \kl{\MELL formulas}, let $\Gamma_{\!j} = (A_{i_{j-1}+1}, \dots, A_{i_{j}})$ for all $1 \leqslant j \leqslant n$.
The list $\Gamma = (\Gamma_1, \allowbreak \dots, \allowbreak \Gamma_n)$ of lists of \kl{\MELL formulas} is also denoted by $(A_1, \dots, A_{i_1}; \cdots ; A_{i_{n-1}+1}, \dots, A_m)$, with lists separated by semicolons. 
The \intro{flattening} of $\Gamma$ is the list $(A_1, \dots, A_m)$ of \MELL formulas.

\subsection{Proof-structures}
\label{subsect:proof-structure}

\begin{figure}[!t]
  \vspace*{-3\baselineskip}
  \centering 
  \pnet{
    \pnformulae{
      ~~
      \pnf[Xcut]{$A$}~\pnf[Xncut]{$A^{\bot}$}~~~
      \pnf[Ao]{$A$}~\pnf[Bo]{$B$}~
      \pnf[Ap]{$A$}~\pnf[Bp]{$B$}~
      \pnf[A1]{$A$}~\pnf[l]{$\cdots$}~\pnf[An]{$A$}~
      \pnf[A1e]{$A$}~\pnf[le]{$\cdots$}~\pnf[Ane]{$A$}
      \\
      \pnf[Xax]{$X$}~\pnf[Xnax]{$X^{\bot}$}~~~
      \pnf[1]{$\mathbf{1}$}~\pnf[bot]{$\bot$}~~~~~~~~~~~
      \pnf[i]{$A_1$}~\pnf[d]{$\overset{p > 0}{\ldots}$}~\pnf[j]{$A_p$}
    }
    \pnaxiom[ax]{Xax,Xnax} 
    \pncut[cut]{Xcut,Xncut}
    \pnone{1}
    \pnbot{bot}
    \pntensor{Ao,Bo}[AoB]{$A\otimes B$}
    \pnpar{Ap,Bp}[ApB]{$A \parr B$}
    \pnexp{}{A1,l,An}[]{$\wn A$}
    \pnbag{}{A1e,le,Ane}[]{$\oc A$}
    \pndaimon{i,j}
  }
  
  \vspace*{-\baselineskip}
  \caption{Cells, with their labels and their typed inputs and outputs.
}
  \label{fig:resource-cells}
\end{figure}

We define here proof-structures corresponding to some subsystems and extensions of \LL:
\MELL, \DiLL and $\DiLL_0$.
Full differential linear logic (\DiLL) is an extension of \MELL (with the same
language as \MELL) provided with both promotion rule (\ie, boxes) and
co-structural rules for the $\oc$-modality (the duals of the structural rules handling the
$\wn$-modality): $\DiLL_0$ and \MELL are particular
subsystems of \DiLL, respectively the promotion-free one (\ie, without boxes) and
the one without co-structural rules. 
We reuse the syntax of proof-structures given in \cite{Wollic}, based on the graph notions introduced in \Cref{sect:graphs}.
Note that, unlike \cite{PaganiTasson2009}, we allow the presence of cuts (vertices of type $\cut$), even though we are not interested in cut-elimination and its study is left to future work.

\begin{defi}[module, proof-structure]
\label{def:proof-structure}
  A (\DiLL) \intro{module} $M = (\rvert M \lvert, \VertType, \Orient, \FlagType, \allowbreak <)$ is a \kl[labeled graph]{labeled} ($\VertType$), \kl[directed graph]{directed} ($\Orient$), \kl[colored graph]{colored} ($\FlagType$), \kl[ordered graph]{ordered} ($<$) \kl{graph} $\lvert M \rvert$ without \kl{loops} such that:
	\begin{itemize}
		\item the map $\VertType \colon V_{|M|} \to \{ \textup{\ax}, \textup{\cut}, \one, \bot, \otimes, \parr, \wn, \oc, \maltese\} 
		$ associates with each vertex $v$ its \intro[vertex type]{type} $\VertType(v)$;
		\item the map $\FlagType \colon F_{|M|} \to \Formulas_{\MELL}$ associates with each flag $f$ its \intro[edge type]{type} $\FlagType(f)$;
		\item $<$ is a strict partial order on the \kl{flags} of $|M|$ that is total on the \kl{tails} of $|M|$ and on the inputs of each \kl{vertex} of type $\parr$ or $\otimes$, separately; 

		\item for every vertex $v \in V_{|M|}$,
		\begin{itemize}
			\item if $\VertType(v) = \textup{\cut}$, $v$ has no \kl{output} and 
			two \kl{inputs} $i_1$ and $i_2$, such that $\FlagType(i_1) = \FlagType(i_2)^{\bot}$;
		
			\item if $\VertType(v) = \textup{\ax}$, $v$ has no inputs and 
			two outputs $o_1$ and $o_2$, with $\FlagType(o_1) = \FlagType(o_2)^{\bot}$ \kl{atomic};\footnotemark 
			\footnotetext{\label{note:atomic}We deal with \emph{atomic axioms} only to simplify the presentation in the next sections. 
				The general case with non-atomic axioms (\ie $\FlagType(o_1)$ and $\FlagType(o_2)$ are still \kl{dual} but not necessarily \kl{atomic}) is discussed in \Cref{sect:nonatomic}.}
			
			\item if $\VertType(v) \in \{\one, \bot\}$, $v$ has no
			inputs and only one output $o$, with $\FlagType(o) = \VertType(v)$;
		
			\item if $\VertType(v) \in \{ \otimes, \parr \}$, $v$ has 
			two inputs $i_1 < i_2$ and one output $o$, with $\FlagType(o) = \FlagType(i_1) \, \VertType(v) \, \FlagType(i_2)$;
			
			\item if $\VertType(v) \in \{\wn, \oc\}$, $v$ has  
			$n \geqslant 0$ inputs $i_1, \ldots, i_n$  and one output $o$, such that
			$\FlagType(o) = \VertType(v) \, \FlagType(i_j)$ for all $1\leqslant j \leqslant n$;\footnotemark
			\footnotetext{This implies that $\FlagType(i_j) = \FlagType(i_k)$ for all $1 \leqslant j,k \leqslant n$. There are no conditions on the number $n$ of inputs.}
			$v$ is a \emph{contraction} (\resp \emph{co-contraction}) if $\VertType(v) = \wn$ (\resp $\VertType(v) = \oc$);
			
			\item if $\VertType(v) =  \maltese$, 
			$v$ has no inputs and
			$p > 0$ outputs $o_1,\ldots,o_p$.\footnotemark
			\footnotetext{There are no conditions on the number of outputs $o_1,\ldots,o_p$ (besides $p > 0$) or on their types.} 
		\end{itemize}
		In Figure \ref{fig:resource-cells} we depicted the corollas associated with 
		all types of vertices.
	\end{itemize}

  A (\DiLL) \intro{proof-structure} is a triple $R = (|R|,\TreeT,\BoxFunction)$ where:
  \begin{itemize}
  \item $|R|  = (\rVert R \lVert, \VertType_R, \Orient_R, \FlagType_R, \allowbreak <_R)$ is a \kl{module} with no input tails, called the \intro{structured graph} of $R$, and $\lVert R \rVert  = (F_{\rVert R \lVert}, V_{\rVert R \lVert}, \partial_{\rVert R \lVert}, j_{\rVert R \lVert})$ is the (\emph{undirected}) \intro[undirected graph]{graph} of $R$;
  \item 
    $\TreeT$ is a \kl{rooted tree}
    with no input tails, called the \intro{box-tree} of $R$;\footnotemark
    \footnotetext{Intuitively, $\TreeT$ represents the tree-structure of the nested boxes of $R$, see Remark \ref{rmk:box} below.}
  \item  $\BoxFunction = (\BoxFunction_{F}, \BoxFunction_{V}) \colon |R| \to \TreeT^{\ReflexiveTransitive}$ is a \kl{directed graph morphism},\footnotemark
    \footnotetext{\label{foot:struct}The structured graph $|R|$ of
    	$R$ is more structured than a directed graph such as $\TreeT^{\ReflexiveTransitive}$: $|R|$ is directed but also labeled, colored, ordered. 
    	When we talk of a morphism between two structured graphs where one of the two, say $\sigma$, is less structured than the other, say $\tau$, we mean that $\tau$ must be only considered with the same structure as $\sigma$.
    	Thus, in this case, $\BoxFunction$ is a morphism from $(\lVert R \rVert, \Orient_{R})$---discarding $\VertType_{R}$, $\FlagType_{R}$, $<_{R}$---to $\TreeT^{\ReflexiveTransitive}$.} 
    called the \intro{box-function} of $R$, such that $\BoxFunction_{F}$ induces a bijection from a subset of $\bigcup_{ v \in V_{\lVert R \rVert}, \VertType(v) = \oc}\In_{\lvert R \rvert}(v)$ (the set of inputs of the vertices of type $\oc$ in $\lvert R \rvert$) to the set of input flags in $\TreeT$;\footnotemark
    \footnotetext{\label{foot:box}It means that for any input flag $f'$ in $\TreeT$ there is exactly one input $f$ of some vertex of type $\oc$ in $|R|$ such that $\BoxFunction_F(f) = f'$; 
    	but, for any input $f$ of some vertex of type $\oc$ in $|R|$, $\BoxFunction_{F}(f)$ need not be an input flag in $\TreeT$ (by definition of directed graph morphism, $\BoxFunction_{F}(f)$ is necessarily an input flag in $\TreeT^{\ReflexiveTransitive}$).
    	In particular, if $|R|$ has no vertices of type $\oc$, $\TreeT$ is a root with its output and no inputs; and if $\lVert R \rVert$ is the \kl[empty graph]{empty graph}, $\BoxFunction$ is the \kl{empty graph morphism} from $\lvert R \rvert$ to such a $\TreeT$.	
    	Intuitively, given a vertex $v$ of type $\oc$ in $\lvert R \rvert$,
    	a box is associated with (and only with) each input $f$ of $v$ such that $\BoxFunction_{F}(f)$ is an input flag in $\TreeT$ and not only in $\TreeT^{\ReflexiveTransitive}$: $f$ represents the principal door in the border of such a box, indeed
			by definition of directed graph morphism, $\BoxFunction_V((\partial_{\lVert R \rVert}\circ j_{\lVert R \rVert})(f)) \neq \BoxFunction_V(\partial_{\lVert R \rVert}(f))$, and  for all $f' \in F_{\lVert R \rVert}$, if $f' \neq f$ then $\BoxFunction_{F}(f') \neq \BoxFunction_{F}(f)$.}
    and for any 
    $v \in V_{\lVert R \rVert}$ and any input $f$ of $v$, 
    if $\BoxFunction_V((\partial_{\lVert R \rVert} \circ j_{\lVert R \rVert}) (f)) \neq
    \BoxFunction_V(\partial_{\lVert R \rVert} (f))$ then $\VertType(v) \in \{ \oc, \wn
    \}$.\footnotemark
    \footnotetext{\label{note:border}Roughly, it says that the border of a box is made of inputs
      of vertices of type $\oc$ or $\wn$ in $\lvert R \rvert$.} 
  \end{itemize}
	
  A (\DiLL) \kl{proof-structure} $R = (|R|, \TreeT, \BoxFunction)$ is said to be:
  \begin{enumerate}
  \item \intro[$\MELL$ proof-structure]{$\MELL$} if its structured graph $\lvert R \rvert$ has no vertices of type $\maltese$ and
    \begin{itemize}
    \item all vertices in $|R|$ of type $\oc$ have exactly one input,
    
		\item the bijection induced by $\BoxFunction_F$ is defined on the whole set $\bigcup_{ v \in V_{\lVert R \rVert}, \VertType(v) = \oc}\In_{\lvert R \rvert}(v)$;\footnotemark
    \footnotetext{It means that for any input flag $f'$ in $\TreeT$ there is exactly one vertex of type $\oc$ in $|R|$ whose unique input $f$ is such that $\BoxFunction_F(f) = f'$; 
    	and, if $f$ is the (only) input some vertex of type $\oc$ in $|R|$, $\BoxFunction_{F}(f)$ is an input flag in $\TreeT$. 
    	Intuitively, a box is associated with (and only with) the unique input of each vertex of type $\oc$~in~$\lvert R \rvert$.}
  \end{itemize}
  \item  \intro[$\DiLL_0$ proof-structures]{$\DiLL_0$} (or \emph{resource})
  if $\TreeT$ consists only of the root with its output, and either $|R|$ has no vertices  of type $\maltese$ or $\lvert R \rvert$ is a \intro{daimon}, \ie $\lvert R \rvert$ has only one vertex and it is of type $\maltese$;

	\item \intro[empty proof-structure]{empty} (which is both $\MELL$ and $\DiLL_0$) if its \kl[undirected graph]{graph} $\lVert R \rVert$ is \kl[empty graph]{empty}; it is denoted by $\emptynet$;
	
  \item \intro{cut-free} if its \kl{structured graph} $\lvert R \rvert$ has no vertices of type $\Cut$, otherwise it is \intro{with cuts}. 
  \end{enumerate}
\end{defi}

Our \kl{$\MELL$ proof-structures} correspond to 
usual $\MELL$ proof-structures (as in \cite{CarvalhoTortora}).
Our \kl{$\DiLL_0$ proof-structures} correspond to 
usual $\DiLL_0$ proof-structures (as in \cite{DBLP:journals/tcs/EhrhardR06}) except that we
allow also (proof-structures that are) \emph{daimons}. 
Daimons will be used in the Taylor expansion to deal with the content of a box taken $0$ times (see \Cref{sect:taylor}).
Our interest for $\DiLL$ proof-structures
is just to have a unitary syntax subsuming both \MELL and $\DiLL_0$ without considering
cut-elimination: for this reason, unlike \cite{Pagani09,Tranquilli09,Tranquilli11}, our $\DiLL$ proof-structures are not allowed to contain a
sum of $\DiLL$ proof-structure inside a box.\footnotemark 
\footnotetext{This restriction is without loss of generality, because a box containing a
sum of proof-structures is equivalent (see \cite{Tranquilli09}) to a co-contraction 
of the boxes of each proof-structure in
the sum. Also, the vertex of type $\maltese$ corresponds to the empty sum inside a box: roughly, it is the content of an ``empty box''.}

Given a proof-structure $R = (|R|, \TreeT, \BoxFunction)$, the output \kl{tails}
of $|R|$ are the \intro{conclusions} of $R$. 
The \intro[proof-structure type]{type} of $R$ is the list of the \kl[edge type]{types} of these
conclusions, ordered according to $<_{\lvert R \rvert}$.  

Borrowing the terminology from \cite{DBLP:journals/tcs/EhrhardR06,DBLP:conf/popl/Lafont90}, in proof-structures, we speak of \intro{cells} instead of vertices, and of \emph{$\ell$-cell} for a
cell of \kl[vertex type]{type} $\ell$. 
An \intro{hypothesis} is a cell without inputs.

\begin{rem}[box]
	\label{rmk:box}
	In our syntax, boxes do not have explicit constructors or cells, hence boxes and
	depth of a proof structure $R = (|R|, \TreeT, \BoxFunction)$ are recovered in a non-inductive~way.
	
	Roughly, any non-root vertex $v$ in $\TreeT$ induces a subgraph of $\TreeT^{\ReflexiveTransitive}$ made up of all vertices ``above $v$'' 
	with their inputs and outputs: the preimage of this subgraph through $\BoxFunction$ is the \emph{box} of $v$ in $\lvert R \rvert$. 
	The preimage of the root of $\TreeT$ through $\BoxFunction$ 
	is the part of $|R|$ outside~any~box.
	
	More precisely, with 	every flag $f$ of $|R|$ such that ${\BoxFunction_{F}}(f)$ is an input flag of $\TreeT$\footnotemark
	\footnotetext{\label{foot:box-tree}By the constraints on $\BoxFunction$, this
		condition can be fulfilled only by inputs of $\oc$-cells in $|R|$, and an input of a $\oc$-cell need not fulfill it; 
		in particular, if $R$ is a \kl{\MELL proof-structure}, then this
		condition is fulfilled by all and only the inputs of $\oc$-cells (and such
		an input is unique for any $\oc$-cell) in $|R|$; 
		but if $R$ is a $\DiLL_0$ proof-structure, this condition is not fulfilled by any flag in $|R|$ (since $\TreeT$ has no inputs) and so $\BoxFunction$ is a \kl{directed graph
		morphism} associating the root of $\TreeT$ with any cell of $R$.
		Thus, in a \kl[DiLL0 proof-structure]{$\DiLL_0$ proof-structure}
		$\rho = (|\rho|,\TreeT_{\rho},\BoxFunction_{\rho})$, there are no boxes, $\TreeT_{\rho}$ and
		$\BoxFunction_{\rho}$ do not induce any structure on $|\rho|$: $\rho$ can be
		\mbox{identified with $|\rho|$.}}   
	is associated a \intro{box} $B_{\!f}$, that is the \kl{subgraph}
	of $|R|$ (which is actually a proof-structure) made up of all the cells $v$ (with their inputs and outputs) such that there is a directed path on $\TreeT$ from ${\BoxFunction_{V}}(v)$ to 	${\BoxFunction_{V}}((\partial_{\lVert R \rVert} \circ j_{\lVert R \rVert})(f))$ (note that $f$ and the $\oc$-cell of which $f$ is an input 
	are not in $B_f$).
	A \emph{conclusion} of such a box $B_{\!f}$  
	is any output flag $f'$ 
	in $B_{\!f}$ such that $(\partial_{\lVert R \rVert} \circ j_{\lVert R \rVert})(f')$ is not in $B_{\!f}$.
	The tree-structure of	$\TreeT$ expresses the usual \emph{nesting condition}~of~boxes:
	two boxes in $|R|$ are either disjoint or contained one in the other.
	
	The \intro{depth} of a cell $v$ of $R$ is the \kl[path length]{length} of the directed path in
	$\TreeT$ from $\BoxFunction_{V}(v)$ to the root of $\TreeT$.  
	The \intro[depth of proof-structure]{depth} of $R$ is the maximal \kl{depth} of the cells~of~$R$.
\end{rem}

\begin{figure}[!t]
	\vspace{\beforepn}
	\centering
	\begin{subfigure}[b]{0.58\textwidth}
	\centering
	\scalebox{0.8}{ 
		\pnet{
			\pnformulae{
				~~~\pnf[botax]{$\bot$}~\pnf[1ax]{$\one$}~~~\pnf[Y]{$Y$}~\pnf[Yp]{$Y^\bot$} \\
				~~\pnf[bot]{$\bot$}~~~\pnf[1]{$\one$}
				\\
				~\pnf[X]{$X$}~~~~~~~~~\pnf[1']{$\one$} 
				\\ \\
				\pnf[Xp]{$X^\bot$}
			}
			\pnone[1axcell]{1ax}
			\pnbot[botaxcell]{botax}
			\pnaxiom[axprop]{Xp,X}
			\pnaxiom[axprop']{Y,Yp}
			\pnboxc{1axcell,botaxcell,botax,1ax}{purple}
			\pnbag{}{1ax}[!1ax]{$\oc \one$}[1.2]
			\pnbot[botcell]{bot}[1][-.02]
			\pnone[1cell]{1}
			\pnboxc{1,1cell,1axcell,botaxcell,botax,1ax,!1ax}{blue}
			\pnbag{}{1}[!1]{$\oc \one$}[1.7]
			\pnexp{}{bot,botax}[?bot]{$\wn \bot$}[1.8]
			\pnexp{}{!1ax,!1}[?!1]{$\wn\oc \one$}[1]
			\pntensor{X,?bot}[tens]{$X \otimes \wn \bot$}[0.9]
			\pnexp{}{Y}[?Y]{$\wn Y$}
			\pnpar{?Y,Yp}[par]{$\wn Y \parr Y^\bot$}
			\pnboxc{axprop',Y,Yp,?Y,par}{red}
			\pnbag{}{par}[!par]{$\oc (\wn Y \parr Y^\bot)$}[1.3]
			\pnone[1cell']{1'}
			\pnboxc{1',1cell'}{orange}
			\pnbag{}{1'}[!1']{$\oc \one$}[1.2]
		}
		}
		\caption{The \kl{structured graph} $\lvert S \rvert$.}
		\end{subfigure}
		\begin{subfigure}[b]{0.4\textwidth}
		\centering
		\scalebox{0.8}{
		\raisebox{-1cm}{
			\begin{forest}
				for tree={grow'=north, edge=<-, l-=2em}
				[$ $, name=root
					[${\,\bullet\,}$, name = root1
						[$ $, name = box11o,
							[$\textcolor{blue}{\bullet}$, name = box11, edge = blue
								[$ $, name = box111o, edge = blue
									[$\textcolor{purple}{\,\bullet\,}$, name = box111, edge = purple]
								]
							]
						]
						[$ $, name = box21o
							[$\textcolor{red}{\bullet}$, name = box21, edge = red]
						]
						[$ $, name = box22o
							[$\textcolor{orange}{\bullet}$, name = box22, edge = orange]
						]
					]
				]
				\coordinate (box111') at ([shift={(0:0.8)}]box111);
				\draw[<-,dotted,purple,line width=1pt] (box111) to[out=north east,
				in=north west] (box111');
				\draw[<-,dotted,purple,line width=1pt] (box111') to[out=south west,
				in=south east] (box111);
				\coordinate (box111root1) at ([shift={(-110:2.5)}]box111);
				\draw[->,dotted,purple,line width=1pt] (box111) to[out=south west,
				in=north west] (box111root1);
				\draw[->,dotted,line width=1pt] (box111root1) to[out=south east,
				in=west] (root1);        
				\coordinate (box11') at ([shift={(180:0.8)}]box11);
				\draw[<-,dotted,blue,line width=1pt] (box11) to[out=north west,
				in=north east] (box11');
				\draw[<-,dotted,blue,line width=1pt] (box11') to[out=south east,
				in=south west] (box11);
				\coordinate (box21') at ([shift={(90:0.8)}]box21);
				\draw[->,dotted,red,line width=1pt] (box21) to[out=north west,
				in=south west] (box21');
				\draw[->,dotted,red,line width=1pt] (box21') to[out=south east,
				in=north east] (box21);
				\coordinate (box22') at ([shift={(0:0.8)}]box22);
				\draw[<-,dotted,orange,line width=1pt] (box22) to[out=north east,
				in=north west] (box22');
				\draw[<-,dotted,orange,line width=1pt] (box22') to[out=south west,
				in=south east] (box22);
				\coordinate (root1') at ([shift={(-30:0.8)}]root1);
				\draw[->,dotted,line width=1pt] (root1) to[out=south,
				in=west] (root1');
				\draw[->,dotted,line width=1pt] (root1') to[out=north,
				in=east] (root1);
			\end{forest}
		}
	}
	\caption{The \kl{box-tree} \texorpdfstring{$\TreeT_{S}$}{} (without dotted
		lines) and the reflexive-transitive closure
		\texorpdfstring{$\TreeT_{S}^{\ReflexiveTransitive}$}{} of
		\texorpdfstring{$\TreeT_{S}$}{} (with also dotted lines).}
	\end{subfigure}
	\caption{A \MELL proof-structure $S = (\lvert S \rvert, \TreeT_{S}, \BoxFunction_{S})$.}
	\label{fig:pointed-proof-net}
\end{figure}

\begin{exa}
	In \Cref{fig:pointed-proof-net} a \MELL proof-structure $S = (\lvert S \rvert, \TreeT_{S}, \BoxFunction_{S})$ is depicted. 
	The box-function $\BoxFunction_{S}$ is kept implicit by
	means of colors: the colored areas in $|S|$ represent boxes (the preimages of non-root vertices of $\TreeT_{S}$ through $\BoxFunction_{S}$), and the same color is used on $\TreeT_{S}$ to show where each box is mapped by $\BoxFunction_{S}$.
\end{exa}

The proof-structures we have defined depend on the carrier sets for vertices and flags. 
It is natural to abstract away from them, and identify two proof-structures that differ only in the name of their vertices and flags, through a handy notion of isomorphism (\ie~structure-preserving bijection) inherited from the one of \kl{graph isomorphism} defined in \Cref{sect:graphs}.

\begin{definition}[isomorphism]
Let $R = (|R|, \TreeT, \BoxFunction)$ and $R' = (|R'|, \TreeT', \BoxFunction')$ be proof-structures.
	An \emph{isomorphism} from $R$ to $R'$ is a couple $(\varphi_{\lvert\cdot\rvert}, \varphi_\mathsf{tree})$ where
	 $\varphi_{\lvert\cdot\rvert} \colon |R| \to |R'|$ is a \kl{structured graph} isomorphism,
	 and $\varphi_\mathsf{tree} \colon \TreeT \to \TreeT'$ is a \kl[directed graph morphism]{directed graph isomorphism}, 
		such that the diagram below commutes.
		\begin{equation*}
		\begin{tikzcd}[row sep=3.0ex]
		\lvert R \rvert \arrow[r, "\BoxFunction"] \arrow[d, "\varphi_{\lvert\cdot\rvert}"]
		\& \mathcal{A}^\circlearrowleft \arrow[d, "\varphi_\mathsf{tree}^\circlearrowleft"] 
		\\
		\lvert R' \rvert \arrow[r, "\BoxFunction'"]
		\& \mathcal{A}'^\circlearrowleft
		\end{tikzcd}
		\end{equation*}
\end{definition}

Note that if $R$ is a \MELL or
$\DiLL_0$ proof-structure isomorphic to a proof-structure $R'$, then $R'$ is respectively \MELL or $\DiLL_0$.
Often we implicitly consider proof-structures up to isomorphism (as we do, for instance, in every figure representing proof-structures).

\subsection{Quasi-proof-structures}

We need to consider more general structures in order to accommodate our
rewrite rules, as discussed in \Cref{sect:outline}, p.~\pageref{subsect:intro-quasi-proof-structure} (the rewrite rules are defined in \Cref{sect:rules,sec:naturality}). 
This is why we extend all the definitions to \emph{tuples} of proof-structures.

\begin{defi}[quasi-proof-structure]
	\label{def:quasi-proof-structure}
	A (\DiLL) \intro{quasi-proof-structure} is a tuple
	$R = (R_1,\dots,R_n)$
	of proof-structures, for some $n > 0$;
	for all $ 1 \leqslant i \leqslant n$, $R_i$ is a \intro{component} of $R$.
	If $R_i = \emptynet$ 
		for all $ 1 \leqslant i \leqslant n$, then $R$ is an \intro{empty quasi-proof-structure}, noted $n \emptynet$ or $\emptynets$.\footnotemark
	\footnotetext{The $0$-ary tuple of proof-structures is \emph{not} a quasi-proof-structure. 
	Indeed, every quasi-proof-structure has at least one (possibly \kl[empty proof-structure]{empty}) component. Note that the components of a quasi-proof-structure are independent of each other, in particular they may differ in the number of their conclusions or in their type.}
	
	For convenience, given a proof-structure $R_i = (\lvert R_i \rvert, \TreeT_{i}, \BoxFunction_{i})$ for all $1 \leqslant i \leqslant n$, we denote the quasi-proof-structure $R = (R_1, \dots, R_n)$ as $R = (\lvert R \rvert, \ForestT_{R}, \BoxFunction_{R})$, where $|R| = (\lvert R_1 \rvert, \dots, \lvert R_n \rvert)$ is the \intro[structured graph of quasi-proof-structure]{structured graph} of $R$, 
	$\ForestT_{R} = (\TreeT_{1}, \dots, \TreeT_{n})$ is the \intro{box-forest} of $R$,  $\BoxFunction_{R} = (\BoxFunction_{1}, \dots, \BoxFunction_{n})$ is the \intro[box-function of quasi-proof-structure]{box-function} of $R$;
	the \emph{graph} of $R$ is $\lVert R \rVert = (\lVert R_1 \rVert, \dots, \lVert R_n\rVert)$.
	
	A \intro[conclusion of quasi-proof-structure]{conclusion} of $R$ is any \kl{conclusion} of any component of $R$. 
	
	A quasi-proof-structure $R$ is \intro[MELL quasi-proof-structure]{\MELL} (\resp \intro[DiLL0 quasi-proof-structure]{$\DiLL_0$}) if all components of $R$ are \MELL (\resp $\DiLL_0$) proof-structures.
\end{defi}

The short notation $R = (\lvert R \rvert, \ForestT_{R}, \BoxFunction_{R})$ for quasi-proof-structures makes sense because the structured graph $|R| = (|R_1|, \dots, |R_n|)$ can be seen as the disjoint union  of its components $|R_1|, \dots, |R_n|$, and similarly for $\ForestT_{R}$ and $\BoxFunction_{R}$. In particular, the box-forest $\ForestT_{R}$ is the disjoint union of the box-trees $\TreeT_{1}, \dots, \TreeT_{n}$ of $R_1, \dots, R_n$.
The box-function $\BoxFunction_{R}$ locates not only the boxes on $\lvert R \rvert$, but also the different components of $R$ on $\lvert R \rvert$ (see also Remark~\ref{rmk:components}).

The only delicate point is the definition of the order $<_{|R|}$ for the \kl[conclusion of quasi-proof-structure]{conclusions} of the quasi-proof-structure $R = (R_1, \dots, R_n)$, given the order $<_{|R_i|}$ for each component $|R_i|$.
Given the conclusions $f,f'$ of $R$, we set $f <_{|R|} f'$ if either $f$ is a conclusion of $R_i$ and $f'$ is a conclusion of $R_{i'}$ with $i < i'$, or $f$ and $f'$ are conclusions of the same \mbox{component $R_j$ and $f <_{|R_j|} f'$.}

We often identify the conclusions of $R$ and their order $<_{|R|}$ with a finite initial segment of $\Nat \setminus \{0\}$ and its natural order.
Note that $\emptynets$ has no conclusions, akin to $\emptynet$, as $\lVert\emptynet\rVert$ is \kl[empty graph]{empty}.

The \intro[quasi-proof-structure type]{type} of a quasi-proof-structure $R = (R_1, \dots, R_n)$ is a list $\Gamma = (\Gamma_1, \dots, \Gamma_n)$ of lists of \MELL formulas such that $\Gamma_i$ is the \kl[proof-structure type]{type} of $R_i$ for all $1 \leqslant i \leqslant n$.
When $n=1$, we often identify $R$ and $R_1$ (a quasi-proof-structure with only one component and a proof-structure), or $\Gamma$ and $\Gamma_1$ (the type of a quasi-proof-structure with only one component and the type of a proof-structure).
So, the type of a quasi-proof-structure $R$ \mbox{determines if $R$ is a proof-structure.}

The \kl[quasi-proof-structure type]{type} of an \kl{empty quasi-proof-structure} $\emptynets$ is a non-empty list of empty lists of \MELL formulas, where each occurrence of the empty list of \MELL formulas is the type of a component $\emptynet$ of $\emptynets$. 
	More precisely, $\emptylist$ is the type of $\emptynet$, and $\kl[empty quasi-proof-structure]{n\emptylist}$ is the type of $n \emptynet$ for every $n > 0$.

\begin{exa}
	In \Cref{fig:pointed-quasi-proof-net} a \(\MELL\)
	\kl{quasi-proof-structure} $R$ is depicted. The colored areas represent the preimages of
	boxes, and 
	each dashed box represents a component of $R$ (in other examples, the absence of a dashed line means that there is only one component).
\end{exa}

\begin{figure}[!t]
	\vspace{\beforepn}
	\centering
	\begin{subfigure}[b]{0.58\textwidth}
		\centering
		\scalebox{0.8}{
		\pnet{
			\pnformulae{
				~~~\pnf[botax]{$\bot$}~\pnf[1ax]{$\one$}~~~\pnf[Y]{$Y$}~\pnf[Yp]{$Y^\bot$} \\
				~~\pnf[bot]{$\bot$}~~~\pnf[1]{$\one$} \\
				~\pnf[X]{$X$}~~~~~~~~~\pnf[1']{$\one$} \\ \\ \\
				\pnf[Xp]{$X^\bot$}
			}
			\pnone[1axcell]{1ax}
			\pnbot[botaxcell]{botax}
			\pnaxiom[axprop]{Xp,X}
			\pnaxiom[axprop']{Y,Yp}
			\pnboxc{1axcell,botaxcell,botax,1ax}{purple}
			\pnbag{}{1ax}[!1ax]{$\oc \one$}[1.2]
			\pnbot[botcell]{bot}[1][-.02]
			\pnone[1cell]{1}
			\pnboxc{1,1cell,1axcell,botaxcell,botax,1ax,!1ax}{blue}
			\pnbag{}{1}[!1]{$\oc \one$}[1.7]
			\pnexp{}{bot,botax}[?bot]{$\wn \bot$}[1.7]
			\pnexp{}{!1ax,!1}[?!1]{$\wn\oc \one$}[.9]
			\pntensor{X,?bot}[tens]{$X \otimes \wn \bot$}[.9]
			\pnexp{}{Y}[?Y]{$\wn Y$}
			\pnpar{?Y,Yp}[par]{$\wn Y \parr Y^\bot$}
			\pnboxc{axprop',Y,Yp,?Y,par}{red}
			\pnbag{}{par}[!par]{$\,\,\oc (\wn Y \parr Y^\bot)$}[1.2]
			\pnone[1cell']{1'}
			\pnboxc{1',1cell'}{orange}
			\pnbag{}{1'}[!1']{$\oc \one$}[1.3]
			\pnsemiboxc{axprop,Xp,1,1axcell,botaxcell}{green}
			\pnsemiboxc{axprop',!1',Y}{gray}
		}
		}
		\caption{The structured graph $\lvert R \rvert$.}
	\end{subfigure}%
	\begin{subfigure}[b]{0.4\textwidth}
		\centering
		\scalebox{0.8}{
			\begin{forest}
				for tree={grow'=north, edge=<-, l-=2em}
				[$ $, name=root
					[$\textcolor{green}{\bullet}$, name = root3, edge = green
						[$ $, name = box11o, edge = green
							[$\textcolor{blue}{\bullet}$, name = box11, edge = blue
								[$ $, name = box111o, edge = blue
									[$\textcolor{purple}{\,\bullet\,}$, name = box111, edge = purple]
								]
							]
						]
					]
				]
				\coordinate (box111') at ([shift={(0:0.8)}]box111);
				\draw[<-,dotted,purple,line width=1pt] (box111) to[out=north east,
				in=north west] (box111');
				\draw[<-,dotted,purple,line width=1pt] (box111') to[out=south west,
				in=south east] (box111);
				\coordinate (box111root3) at ([shift={(-110:2.7)}]box111);
				\draw[->,dotted,purple,line width=1pt] (box111) to[out=south west,
				in=north west] (box111root3);
				\draw[->,dotted,green,line width=1pt] (box111root3) to[out=south east,
				in=west] (root3);        
				\coordinate (box11') at ([shift={(180:0.8)}]box11);
				\draw[<-,dotted,blue,line width=1pt] (box11) to[out=north west,
				in=north east] (box11');
				\draw[<-,dotted,blue,line width=1pt] (box11') to[out=south east,
				in=south west] (box11);
				\coordinate (root3') at ([shift={(-30:0.8)}]root3);
				\draw[->,dotted,green,line width=1pt] (root1) to[out=south,
				in=west] (root3');
				\draw[->,dotted,green,line width=1pt] (root3') to[out=north,
				in=east] (root3);
			\end{forest}
			\begin{forest}
				for tree={grow'=north, edge=<-, l-=2em}
				[$ $, name=root
					[$\textcolor{gray}{\bullet}$, name = root1, edge = gray
						[$ $, name = box21o, edge = gray
							[$\textcolor{red}{\bullet}$, name = box21, edge = red]
						]
						[$ $, name = box22o, edge = gray
							[$\textcolor{orange}{\bullet}$, name = box22, edge = orange]
						]
					]
				]
				\coordinate (box11') at ([shift={(180:0.8)}]box11);
				\coordinate (box21') at ([shift={(90:0.8)}]box21);
				\draw[->,dotted,red,line width=1pt] (box21) to[out=north west,
				in=south west] (box21');
				\draw[->,dotted,red,line width=1pt] (box21') to[out=south east,
				in=north east] (box21);
				\coordinate (box22') at ([shift={(0:0.8)}]box22);
				\draw[<-,dotted,orange,line width=1pt] (box22) to[out=north east,
				in=north west] (box22');
				\draw[<-,dotted,orange,line width=1pt] (box22') to[out=south west,
				in=south east] (box22);
				\coordinate (root1') at ([shift={(-30:0.8)}]root1);
				\draw[->,dotted,gray,line width=1pt] (root1) to[out=south,
				in=west] (root1');
				\draw[->,dotted,gray,line width=1pt] (root1') to[out=north,
				in=east] (root1);
			\end{forest}
		}
		\caption{The \kl{box-forest} \texorpdfstring{$\ForestT_{R}$}{} (without dotted
			lines) and the reflexive-transitive closure
			\texorpdfstring{$\ForestT_{R}^{\ReflexiveTransitive}$}{} of
			\texorpdfstring{$\ForestT_{R}$}{} (with also dotted lines).}
	\end{subfigure}
	\caption{A \(\MELL\) quasi-proof-structure $R = (\lvert R \rvert, \ForestT_{R}, \BoxFunction_{R})$.}
	\label{fig:pointed-quasi-proof-net}
\end{figure}

\begin{remark}[components]
	\label{rmk:components}
	A difference arises between a $\DiLL_0$ proof-structure and a $\DiLL_0$ quasi-proof-structure. 
	Both have no boxes.
	In the former, its box-tree and box-function do not induce any structure on its \kl{structured graph}, so we can identify a $\DiLL_0$ proof-structure with its \kl{structured graph}. 
	In the latter, its box-forest and box-function do 
	induce a structure on its structured graph: indeed, its box-forest is made up only of roots with their output and its box-function separates the components of its structured graph, by mapping the vertices to the roots. 
	So, we cannot identify a $\DiLL_0$ quasi-proof-structure with~its~\kl{structured graph}.
\end{remark}


\section{The Taylor expansion}
\label{sect:taylor}

The \emph{Taylor expansion} $\Taylor{R}$ \cite{Ehrhard:2008} of a \MELL (or more generally \DiLL) proof-structure $R$ is a possibly infinite set of $\DiLL_0$ proof-structures: 
roughly, each element of $\Taylor{R}$ is obtained from $R$
by approximating---\ie~replacing---each box $B$ in $R$ with $n_B$ copies of its content (for some $n_B \in \Nat$), recursively on the \kl{depth} of $R$. 
Note that $n_B$ depends not only on $B$ but also on which ``copy'' 
of all boxes containing $B$ we are considering. 
Usually, the Taylor expansion $\Taylor{R}$ of a \MELL proof-structure $R$ is defined globally and inductively \cite{MazzaPagani07,PaganiTasson2009,ChouquetVaux18}: with $R$ is directly associated the whole set $\Taylor{R}$ by induction on the depth of~$R$.
A drawback of this
approach is that, for each element of $\Taylor{R}$, the way the copies of the content of a box are merged is defined ``by hand'', which is syntactically heavy. 

Following \cite{Wollic}, we adopt an alternative non-inductive approach, which 
strongly refines \cite{FSCD2016}: the Taylor expansion is defined \emph{pointwise}
(see Example \ref{ex:proto-taylor} and \Cref{fig:taylor-expansion}).
Indeed, proof-structures have a tree structure made explicit by their \kl{box-function}.
The definition of the Taylor expansion uses this tree structure: 
first, we define how to ``\emph{expand}'' a tree via the notion of \kl{thick subtree} \cite{Boudes:2009} (Definition \ref{def:thick}; roughly, it states the
number of copies of each box to be taken, recursively), 
we then take all
these expansions of the box-tree of a proof-structure and we \emph{pull} them \emph{back} to the underlying graphs (Definition \ref{def:proto}), finally
we \emph{forget} the tree structures associated with them (Definition \ref{def:taylor}).
Thus, pullbacks give an abstract and elegant way to define the merging of copies of the content of a box in an
element of $\Taylor{R}$.
The use of pullbacks is made possible because all the ingredients in our
definition of a proof-structure live in the category of directed graphs (as defined in \Cref{sect:graphs}).

An advantage of our approach is that it smoothly generalizes to \kl{quasi-proof-structures}.

\begin{defi}[thick subtree and subforest]
  \label{def:thick}
  Let $\sigma$ be a \kl{rooted tree}. 
  A \intro{thick subtree} of $\sigma$ is a pair
  $(\tau, h)$ of a rooted tree $\tau$ and a directed graph morphism
  $h = (h_F, h_V) \colon \tau \to \sigma$.
  
  Let $\sigma = (\sigma_1, \dots, \sigma_n)$ be a \kl{forest} of \kl{rooted trees}, with $n > 0$.
  A \intro{thick subforest} of $\sigma$ is a tuple $((\tau_1, h_1), \dots, (\tau_n, h_n))$ where  $(\tau_i, h_i)$ is a thick subtree of $\sigma_i$ for all $1 \leqslant i \leqslant n$.
  It is denoted by $(\tau, h)$, where $\tau = (\tau_1, \dots, \tau_n)$ and $h = (h_1, \dots, h_n)$.
\end{defi}

\begin{exa}
  \label{ex:proto-taylor}
  The following is a graphical presentation of a \kl{thick subforest} $(\tau,h)$
  of the \kl{box-forest} $\ForestT_R$ of the \kl{quasi-proof-structure} in
  \Cref{fig:pointed-quasi-proof-net}, where the directed graph morphism
  $h = (h_F, h_V) \colon \tau \to \ForestT_R$ is depicted chromatically (same color means same
  image via $h$).  
  \vspace{-0.3\baselineskip}
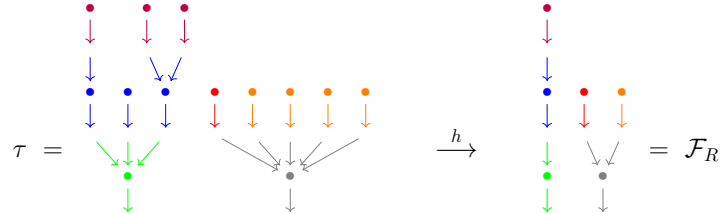
  \begin{align*}
    \scalebox{0.7}{
    \raisebox{1.3cm}{\Large$\tau \ = \ $}
		\begin{forest}
			for tree={grow'=north, edge=<-, l=0.2pt}
			[$ $ 
				[$\textcolor{green}{\bullet}$, edge = green,
					[$ $, edge = green
						[$\textcolor{blue}{\bullet}$, edge = blue
							[$ $, edge = blue
								[$\textcolor{purple}{\bullet}$, edge = purple]
							]
						]
					]
					[$ $, edge = green,
						[$\textcolor{blue}{\bullet}$, edge = blue]
					]
					[$ $, edge = green,
						[$\textcolor{blue}{\bullet}$, edge = blue 
							[$ $, edge = blue
								[$\textcolor{purple}{\bullet}$, edge = purple]
							]
							[$ $, edge = blue
								[$\textcolor{purple}{\bullet}$, edge = purple]
							] 
						]
					]
				]
			]
		\end{forest}
		\begin{forest}
			for tree={grow'=north, edge=<-, l=0.2pt}
			[$ $, edge = gray
				[$\textcolor{gray}{\bullet}$, edge = gray
					[$ $, edge = gray
						[$\textcolor{red}{\bullet}$, edge = red]
					]
					[$ $, edge = gray
						[$\textcolor{orange}{\bullet}$, edge = orange]
					]
					[$ $, edge = gray
						[$\textcolor{orange}{\bullet}$, edge = orange]
					]
					[$ $, edge = gray
						[$\textcolor{orange}{\bullet}$, edge = orange]
					] 
					[$ $, edge = gray
						[$\textcolor{orange}{\bullet}$, edge = orange]
					]
				]
			]
		\end{forest}
		\raisebox{1.3cm}{\Large$\qquad\overset{h}{\longrightarrow}\qquad$}
		\begin{forest}
		for tree={grow'=north, edge=<-, l=0.2pt}
		[$ $, edge = green
			[$\textcolor{green}{\bullet}$, edge = green
				[$ $, edge = green
					[$\textcolor{blue}{\bullet}$, edge = blue
						[$ $, edge = blue
							[$\textcolor{purple}{\bullet}$, edge = purple]
						]
					]
				]
			]
		]
		\end{forest}
		\
		\begin{forest}
		for tree={grow'=north, edge=<-, l=0.2pt}
		[$ $, edge = gray
			[$\textcolor{gray}{\bullet}$, edge = gray
				[$ $, edge = gray
					[$\textcolor{red}{\bullet}$, edge = red]
				]
				[$ $, edge = gray
					[$\textcolor{orange}{\bullet}$, edge = orange]
				]
			]
		]
		\end{forest}
    \raisebox{1.3cm}{\Large$ \ =  \ \ForestT_R$}
    }
  \end{align*}
  Intuitively, it means that $\tau$ is obtained from $\ForestT_R$ by taking $3$
  copies of the blue box, $1$ copy of the red box and $4$ copies of the orange
  box; in the first (resp.~second; third) copy of the blue box, $1$ copy
  (resp.~$0$ copies; $2$ copies) of the purple box has been taken.
\end{exa}

\begin{remark}[roots]
	\label{rmk:roots}
	If $(\tau,h)$ is a \kl{thick subforest} of a forest $\sigma$ of rooted trees, $h$ sets up a bijection between the roots of $\tau$ and $\sigma$, by definition of \kl{directed graph morphism}; 
	so, we can identify the roots of $\tau$ with the ones of $\sigma$.
	In particular, if $\sigma$ is made of roots with their output and no inputs, the only thick subforest of $\sigma$ is $\sigma$ with its \kl{identity graph morphism}.
\end{remark}

The crucial point is to pull back the expansion of a forest to quasi-proof-structures.

\begin{defi}[proto-Taylor expansion]
  \label{def:proto}
  Let $R = (|R|, \ForestT_R, \BoxFunction_{R})$ be a \kl{quasi-proof-structure}.
  The \intro{proto-Taylor expansion} of $R$ is the set $\ProtoTaylor{R}$ of
  thick subforests of $\ForestT_{R}$.
  
  Let $t = (\tau_t, h_t) \in \ProtoTaylor{R}$.
  The \intro[tree expansion]{$t$-expansion} of $R$ is the \kl{pullback} $(R_t, p_t, p_R)$ below,
  computed in the category of
  directed graphs and directed graph morphisms.\footnotemark
  \footnotetext{So, $\lvert R \rvert$ is considered as directed graph (see \Cref{foot:struct}), forgetting that it is colored, labeled, ordered.}
  \begin{align*}
    \begin{tikzcd}
      R_t \ar[r,"p_t"] \ar[d,swap,"p_R"]
      \arrow[dr, phantom, "\scalebox{1.5}{$\lrcorner$}" , very near start, color=black]
      \& \tau_t^{\circlearrowleft}
      \ar[d,"h_t^{\circlearrowleft}"]\\
      \lvert R \rvert  \ar[r,"\BoxFunction_{R}"]\& \ForestT_{R}^{\circlearrowleft}
    \end{tikzcd}
  \end{align*}
\end{defi}

Given a \kl{quasi-proof-structure} $R$ and $t = (\tau_t, h_t) \in \ProtoTaylor{R}$,
the directed graph $R_t$ inherits the types of its vertices and flags by pre-composition of $\VertType_{\lvert R \rvert}$ and $\FlagType_{\lvert R \rvert}$ with the \kl{graph morphism} $p_R \colon R_t \to |R|$.
The order on the flags of $R_t$ is induced by the one on~$\lvert R \rvert$~via~$p_R$.

Let $[\tau_t]$ be the forest made up of the roots of $\tau_t$ and
$\iota \colon \tau_t \to [\tau_t]$ be the graph morphism sending each vertex of
$\tau_t$ to the root below it; $\iota^\circlearrowleft$ induces by
post-composition a morphism
$\overline{h_t} = \iota^\circlearrowleft \circ p_t \colon R_t \to
[\tau_t]^{\circlearrowleft}$. The triple $(R_t, [\tau_t], \overline{h_t})$ is a
$\DiLL_0$ \kl{quasi-proof-structure}, and it is a $\DiLL_0$ \kl{proof-structure} if $R$ is a
proof-structure.
We can now define the \emph{Taylor expansion} $\Taylor{R}$ of a quasi-proof-structure $R$ (an example of an element of a Taylor expansion~is~in~\Cref{fig:taylor-expansion}).

\begin{defi}[Taylor expansion]
  \label{def:taylor}
  Let $R$ be a quasi-proof-structure. The \intro{Taylor expansion} of $R$ is the
  set of \polyadic{} quasi-proof-structures
  \mbox{$\Taylor{R} = \{ (R_t, [\tau_t], \overline{h_t}) \mid t = (\tau_t, h_t) \in
  \ProtoTaylor{R}\}$.}
\end{defi}

An element $(R_t, [\tau_t], \overline{h_t})$ of the Taylor expansion of a quasi-proof-structure $R$ 
has much less
structure than the pullback $(R_t,p_t,p_R)$: the latter 
indeed is a $\DiLL_0$ quasi-proof-structure $R_t$ coming with its projections
$|R| \stackrel{p_R}{\longleftarrow} R_t \stackrel{p_t}{\longrightarrow}
\tau_t^\circlearrowleft$, which establish a precise correspondence between cells and flags of $R_t$ and cells and flags of $R$: 
a cell in $R_t$ is labeled (via 
the projections) by both the cell of $|R|$ and the branch of the box-forest of $R$ it
arose from.  But $(R_t, [\tau_t], \overline{h_t})$ where $R_t$ is without its
projections $p_t$ and $p_R$ loses the correspondence with
$R $. 

\begin{figure}[!t]
	\vspace{\beforepn}
	\centering
	\scalebox{0.75}{ 
		\pnet{
			\pnformulae{
				~~~\pnf[botax]{$\bot$}~\pnf[1ax]{$\one$}~\pnf[botaxbis]{$\bot$}~\pnf[1axbis]{$\one$}~\pnf[botaxtris]{$\bot$}~\pnf[1axtris]{$\one$}~~~~~\pnf[Y]{$Y$}~\pnf[Yp]{$Y^\bot$} \\
				~~\pnf[bot]{$\bot$}~~~~~~~\pnf[1]{$\one$}~\pnf[1bis]{$\one$}~\pnf[1tris]{$\one$} \\
				~\pnf[X]{$X$}~~~~~~~~~~~~~~\pnf[1']{$\one$}~\pnf[1'']{$\one$}~\pnf[1''']{$\one$}~\pnf[1'''']{$\one$} \\ 
				\pnf[Xp]{$X^\bot$}~~~~~~\pnf[1cow]{$\oc \one$}
			}
			\pnone[1axcell]{1ax}
			\pnbot[botaxcell]{botax}
			\pnaxiom[axprop]{Xp,X}
			\pnone[1axcellbis]{1axbis}
			\pnbot[botaxcellbis]{botaxbis}
			\pnone[1axcelltris]{1axtris}
			\pnbot[botaxcelltris]{botaxtris}
			\pnaxiom[axprop']{Y,Yp}
			\pnbag{}{1ax}[!1ax]{$\oc 1$}
			\pnbag{}{1axbis,1axtris}[!1axtris]{$\oc \one$}[1][.3]
			\pninitial{$\oc$}[cow]{1cow}
			\pnbot[botcell]{bot}
			\pnone[1cell]{1}
			\pnone[1cellbis]{1bis}			
			\pnone[1celltris]{1tris}
			\pnbag{}{1,1bis,1tris}[!1]{$\oc \one$}
			\pnexp{}{bot,botax,botaxbis,botaxtris}[?bot]{$\wn \bot$}[1.5][-1]
			\pnexp{}{!1ax,!1axtris,!1,1cow}[?!1]{$\wn\oc 1$}[1]
			\pntensor{X,?bot}[tens]{$X \otimes \wn \bot$}[.9]
			\pnexp{}{Y}[?Y]{$\wn Y$}
			\pnpar{?Y,Yp}[par]{$\wn Y \parr Y^\bot$}
			\pnbag{}{par}[!par]{$\,\,\oc (\wn Y \parr Y^\bot)$}
			\pnone[1cell']{1'}
			\pnone[1cell'']{1''}
			\pnone[1cell''']{1'''}
			\pnone[1cell'''']{1''''}
			\pnbag{}{1',1'',1''',1''''}[!1']{$\oc 1$}[1.5]
			\pnsemiboxc{1axcell,botaxcell,1axcellbis,botaxcellbis,1axcelltris,axprop,Xp,1,1bis,1tris,ax,tens}{green}
			\pnsemiboxc{axprop',!1',1'''',Y}{gray}
		}
		\qquad
		\begin{forest}
			for tree={grow'=north, edge=<-, l=0.2pt}
			[$ $, edge = green
				[$\textcolor{green}{\bullet}$, name = box12,  edge = green]
			]
			\coordinate (box12') at ([shift={(90:0.8)}]box12);
			\draw[->,dotted,green,line width=1pt] (box12) to[out=north west,
			in=south west] (box12');
			\draw[->,dotted,green,line width=1pt] (box12') to[out=south east,
			in=north east] (box12);
		\end{forest}
		\ 
		\begin{forest}
			for tree={grow'=north, edge=<-, l=0.2pt}
			[$ $, edge = gray
				[$\textcolor{gray}{\bullet}$, name = box21, edge = gray]
			]
			\coordinate (box21') at ([shift={(90:0.8)}]box21);
			\draw[->,dotted,gray,line width=1pt] (box21) to[out=north west,
			in=south west] (box21');
			\draw[->,dotted,gray,line width=1pt] (box21') to[out=south east,
			in=north east] (box21);
		\end{forest}
	}
	\caption{An element $\rho$ of the Taylor expansion of the $\MELL$ quasi-\-proof-structure $R$ in \Cref{fig:pointed-quasi-proof-net}, obtained from the element of $\ProtoTaylor{R}$ in Example~\ref{ex:proto-taylor}.}
	\label{fig:taylor-expansion}
\end{figure}

\begin{rem}[conclusions]
	\label{rmk:conclusions}
		By definition and Remark \ref{rmk:roots}, the Taylor expansion preserves \kl[conclusion of quasi-proof-structure]{conclusions} and \kl[quasi-proof-structure type]{type}: each element of the Taylor expansion of a quasi-proof-structure $R$ has the \emph{same conclusions} and the \emph{same type} as $R$.
	More precisely, there is a
  bijection $\varphi$ from the conclusions of a quasi-proof-structure $R$ to the
  ones in each element $\rho$ of $\Taylor{R}$ such that $i$ and $\varphi(i)$
  have the same type and the same root (\ie
  $\BoxFunction_{R_F}(i) = \BoxFunction_{\rho_F}(\varphi(i))$ up to the bijection of Remark \ref{rmk:roots}).
  So, the \kl[proof-structure type]{types} of $R$ and $\rho$ are the same (as a list of lists).
\end{rem}

\begin{example}[Taylor of the empty]
	\label{ex:taylor-empty}
	By Definitions \ref{def:proof-structure} and \ref{def:quasi-proof-structure}, the box-tree of $\emptynet$ is just a root with its output and no inputs (we call it a \emph{root-tree}), and the box-forest of $\kl[empty quasi-proof-structure]{n\emptynet}$ is a list of $n$ root-trees, for all $n > 0$; the box-function of $\emptynet$ is the \kl{empty graph morphism} and the box-function of $n\emptynet$ is a list of $n$ such morphisms. 
	By Remark \ref{rmk:roots}, $\ProtoTaylor{\emptynet}$ (\resp\ $\ProtoTaylor{n\emptynet}$ for all $n \!>\! 0$) is $\{(\tau, h)\}$ where $\tau$ is a root-tree (\resp\ a forest of $n$ root-trees) and $h$ is the \kl{identity graph morphism} on $\tau$.
	By Definition \ref{def:taylor}, $\Taylor{\emptynet} = \{\emptynet\}$ and $\Taylor{\kl[empty quasi-proof-structure]{n \emptynet}} = \{n \emptynet\}$~for~all~$n \!>\! 0$.
\end{example}

\subsection{The filled Taylor expansion}
\label{subsec:fat}

As discussed in \Cref{subsec:intro-fattened}
(p.~\pageref{subsec:intro-fattened}), the rewriting rules we will introduce in \Cref{sect:rules} need to ``represent'' the emptiness introduced by the Taylor expansion (taking $0$
copies of a box) so as to preserve the conclusions.  
Thus, an element of the \emph{filled Taylor expansion} $\FatTaylor{R}$ of a quasi-proof-structure $R$ 
(an example is in \Cref{fig:emptyings}) is obtained from an element of $\Taylor{R}$
where some \kl{components} can be erased and replaced by 
\kl{daimons} with the same conclusions (hence $\Taylor{R} \subseteq \FatTaylor{R}$).
The filled Taylor expansion (and not the plain one) will play a crucial role in \Cref{sec:naturality} to define a natural transformation.

\begin{defi}[emptying, filled Taylor expansion]
	\label{def:FilledTaylor}
	The \intro[emptying of proof-structure]{emptying} of a $\DiLL_0$ proof-structure~
	$\rho$ with at least one conclusion is a \kl{daimon} whose \kl{conclusions} and \kl[proof-structure type]{type} are the same as $\rho$. 
	
  An \intro{emptying} of a \kl[DiLL0 quasi-proof-structure]{$\DiLL_0$ \emph{quasi}-proof-structure} $\rho$ is a $\DiLL_0$ quasi-proof-structure obtained from $\rho$ by replacing 
	some (possibly none) \kl{components} of $\rho$ having at least one conclusion with their \kl[emptying of proof-structure]{emptying}.
		
	The \intro{filled Taylor expansion} $\FatTaylor{R}$ of a
	quasi-proof-structure $R$ is the set of all the emptyings of all the elements of
	its Taylor expansion $\Taylor{R}$.
	
	If $\rho = (R_t, [\tau_t], \overline{h_t}) \in \Taylor{R}$, and $r_1, \dots, r_n$ are some roots of  $[\tau_t]$, the \intro[emptying on]{emptying of $\rho$ on $r_1, \dots, r_n$}, denoted by $\rho_{r_1 \dots r_n}$, is the element of $\FatTaylor{R}$ obtained by substituting a daimon 
	for each component of $\rho$ with at least one conclusion corresponding to the root $r_i$ (\ie, if such a conclusion is an output of a vertex $v$, then $\overline{h_t}_V(v) = r_i$), for all $1 \leqslant i \leqslant n$.
\end{defi}

\begin{figure}[!t]
  \vspace{\beforepn}
  \centering
  \scalebox{0.75}{
    \pnet{
      \pnformulae{
        ~~~~~~\pnf[yg]{}~~\pnf[1']{$1$}~\pnf[1'']{$1$}\\
        \pnf[Xp]{$X^\bot$}~~\pnf[tens]{$X\otimes \wn \bot$}~~\pnf[?!1]{$\wn\oc 1$}~~~\pnf[YY]{$\oc (\wn Y \parr Y^\bot)$}
      }
      \pndaimon[dai]{Xp,tens,?!1}
      \pncown[!par]{YY}
      \pnone[1cell']{1'}
      \pnone[1cell'']{1''}
      \pnbag{}{1',1''}[!1']{$\oc 1$}
      \pnsemiboxc{dai,Xp,tens,?!1}{green}
      \pnsemiboxc{YY,1cell',1cell'',!par,yg}{gray}
    }
		\qquad
		\quad
		\raisebox{-.5cm}{
		\begin{forest}
			for tree={grow'=north, edge=<-, l=0.2pt}
			[$ $, edge = green
				[$\textcolor{green}{\bullet}$, name = box12, edge = green]
			]
			\coordinate (box12') at ([shift={(90:0.8)}]box12);
			\draw[->,dotted,green,line width=1pt] (box12) to[out=north west,
			in=south west] (box12');
			\draw[->,dotted,green,line width=1pt] (box12') to[out=south east,
			in=north east] (box12);
		\end{forest}
		\ 
		\begin{forest}
			for tree={grow'=north, edge=<-, l=0.2pt}
			[$ $, edge = gray
				[$\textcolor{gray}{\bullet}$, edge = gray]
			]
			\coordinate (box21') at ([shift={(90:0.8)}]box21);
			\draw[->,dotted,gray,line width=1pt] (box21) to[out=north west,
			in=south west] (box21');
			\draw[->,dotted,gray,line width=1pt] (box21') to[out=south east,
			in=north east] (box21);
		\end{forest}
		\quad
		}
	}
  \caption{An element of the filled Taylor expansion of the $\MELL$ quasi-proof-structure $R$ in \Cref{fig:pointed-quasi-proof-net}, obtained as an emptying of $\rho$ in~\Cref{fig:taylor-expansion}.}
  \label{fig:emptyings}
\end{figure}

\begin{remark}[conclusions of emptying]
	\label{rmk:conclusions-emptying}
	By construction, the \kl{filled Taylor expansion} preserves conclusions and type, as the Taylor expansion does (Remark \ref{rmk:conclusions}):  conclusions and type of a quasi-proof-structure $R$ and of \emph{any} emptying of \emph{any} element of $\Taylor{R}$ are the same.
\end{remark}

\begin{example}[filled Taylor of the empty]
	\label{ex:filled-taylor-empty}
	From Example \ref{ex:taylor-empty} and Definition \ref{def:FilledTaylor}, it follows that $\Taylor{\emptynet} = \{\emptynet\} = \FatTaylor{\emptynet}$ and for every $n > 0$, $\Taylor{n\emptynet} = \{\kl[empty quasi-proof-structure]{n \emptynet}\} = \FatTaylor{\kl[empty quasi-proof-structure]{n\emptynet}}$.
\end{example} 

\begin{remark}[connection]
	\label{rmk:connection}
	The Taylor expansion does not ``create connection''; the filled Taylor expansion does, but only in a component that has been filled by a daimon. 
	More precisely, let $i$ and $j$ be conclusions of a quasi-proof-structure $R = (|R|, \ForestT_R, \BoxFunction_R)$, let $\rho = \allowbreak(|\rho|, \ForestT_\rho, \BoxFunction_{\rho}) \!\in \Taylor{R}$ 
	and let $\rho_{r_1 \dots r_n} \in \FatTaylor{R}$ be the \kl[emptying on]{emptying of $\rho$ on} the roots $r_1, \dots, r_n$. 
	If $i$ and $j$ are not connected in the (undirected) \kl[undirected graph]{graph} $\lVert R \rVert$ of $R$, then $i$ and $j$ are not connected in the (undirected) graph $\lVert \rho \rVert$ of $\rho$.
	And if $i$ and $j$ are connected in the (undirected) graph of $\lVert \rho_{r_1 \dots r_n} \rVert$ of $\rho_{r_1 \dots r_n}$, then $i$ and $j$ are output tails of the same $\maltese$-cell in $\rho_{r_1 \dots r_n}$.

\end{remark}


\section{Means of destruction: unwinding \texorpdfstring{$\MELL$}{MELL} quasi-proof-structures}
\label{sect:rules}

\begin{figure}[!t]
  \centering
  \small
  \arraycolsep=1.4pt
  $\begin{array}{rcl}
     (\Gamma_1; \cdots ; \Gamma_k, \FlagType(i), \FlagType(i\!+\!1), \Gamma_{k}'; \cdots; \Gamma_n)
     &\xrightarrow{\Ex_{i}}
     &(\Gamma_1; \cdots ; \Gamma_k, \FlagType(i\!+\!1),  \FlagType(i), \Gamma_{k}'; \cdots; \Gamma_n)
     \\
     (\Gamma_1; \cdots ; \Gamma_{k}, \FlagType(i), \FlagType(i\!+\!1), \Gamma_{k}'; \cdots; \Gamma_n)
     &\xrightarrow{\Mix_{i}}
     &(\Gamma_1; \cdots ; \Gamma_{k}, \FlagType(i); \FlagType(i\!+\!1), \Gamma_{k}'; \cdots; \Gamma_n)
     \\
     (\Gamma_1; \cdots ; \Gamma_{k\!-\!1}; \FlagType(i\!-\!1), \FlagType(i); \Gamma_{k\!+\!1}; \cdots; \Gamma_n)
     &\xrightarrow{\Ax_{i}}
     &(\Gamma_1; \cdots ; \Gamma_{k\!-\!1}; \emptylist; \Gamma_{k\!+\!1} ; \cdots; \Gamma_n) \text{ if } \FlagType(i\!-\!1) = \FlagType(i)^\bot \text{ \kl{atomic}}
     \\
     (\Gamma_1; \cdots ; \Gamma_{k}; \cdots; \Gamma_n)
     &\xrightarrow{\Cut^i}
     &(\Gamma_1; \cdots ; \Gamma_k, \FlagType(i), \FlagType(i\!+\!1); \cdots; \Gamma_n)
       \text{ if } \FlagType(i) = \FlagType(i\!+\!1)^\bot 
       \\
     (\Gamma_1; \cdots ; \Gamma_{k\!-\!1}; \FlagType(i); \Gamma_{k\!+\!1}; \cdots; \Gamma_n)
     &\xrightarrow{\One_i}
     &(\Gamma_1; \cdots ; \Gamma_{k\!-\!1}; \emptylist; \Gamma_{k\!+\!1} ; \cdots; \Gamma_n) \text{ if } \FlagType(i) = \One 
     \\
     (\Gamma_1; \cdots; \Gamma_{k\!-\!1}; \FlagType(i); \Gamma_{k\!+\!1}; \cdots; \Gamma_n)
     &\xrightarrow{\bot_i}
     &(\Gamma_1; \cdots ; \Gamma_{k\!-\!1}; \emptylist; \Gamma_{k\!+\!1}; \cdots; \Gamma_n)  \text{ if } \FlagType(i) = \bot 
     \\
     (\Gamma_1; \cdots ; \Gamma_k, \FlagType(i); \cdots; \Gamma_n)
     &\xrightarrow{\otimes_i}
     &(\Gamma_1; \cdots ; \Gamma_{k}, A, B; \cdots; \Gamma_n) \text{ if } \FlagType(i) = A \otimes B 
     \\
     (\Gamma_1; \cdots ; \Gamma_{k}, \FlagType(i); \cdots; \Gamma_n)
     &\xrightarrow{\parr_i}
     &(\Gamma_1; \cdots ; \Gamma_{k}, A, B; \cdots; \Gamma_n) \text{ if } \FlagType(i) = A \parr B 
     \\
     (\Gamma_1; \cdots ; \Gamma_{k}, \FlagType(i); \cdots; \Gamma_n)
     &\xrightarrow{\Contr_i}
     &(\Gamma_1; \cdots ; \Gamma_{k}, \wn A, \wn A; \cdots; \Gamma_n) \text{ if } \FlagType(i) = \wn A 
     \\
     (\Gamma_1; \cdots ; \Gamma_{k}, \FlagType(i); \cdots; \Gamma_n)
     &\xrightarrow{\Der_i}
     &(\Gamma_1; \cdots ; \Gamma_{k}, A; \cdots; \Gamma_n)
       \text{ if } \FlagType(i) = \wn A
     \\
     (\Gamma_1; \cdots ; \Gamma_{k\!-\!1}; \FlagType(i); \Gamma_{k\!+\!1}; \cdots; \Gamma_n)
     &\xrightarrow{\Weak_i}
     &(\Gamma_1; \cdots ; \Gamma_{k\!-\!1}; \emptylist; \Gamma_{k\!+\!1} ; \cdots; \Gamma_n) \text{ if } \FlagType(i) = \wn A
     \\
		(\Gamma_{\!1}; \cdots ; \Gamma_{\!k\!-\!1}; \wn A_1, \overset{p \geqslant 0}{\dots}, \wn A_p, \FlagType(i); \Gamma_{\!k\!+\!1}\cdots; \Gamma_{\!n})
		&\xrightarrow{\BoxR_{i}}
		&(\Gamma_{\!1}; \cdots ; \Gamma_{\!k\!-\!1}; \wn A_1, \overset{p \geqslant 0}{\dots}, \wn A_p, A; \Gamma_{\!k\!+\!1}; \cdots; \Gamma_{\!n}) \text{ if } \FlagType(i) = \oc A
   \end{array}$
   \caption{The generators of $\Scheduling$.
     In the source $\Gamma$ of each arrow, $\FlagType(i)$ 
     is the $i^\text{th}$ formula in the \kl{flattening} of $\Gamma$ (the 
     arrow's name keeps track~of~$i$).
   	The $\Gamma_k$'s are (possibly empty) lists of \MELL formulas.}
   \label{fig:elementary-schedulings}
 \end{figure}

Our aim is to deconstruct ($\MELL$ or \polyadic) proof-structures 
from their conclusions. 
To do that, we introduce the category of schedulings. 
The arrows of this category are sequences of deconstructing
rules, acting on lists of lists of \MELL formulas. 
These arrows act through functors on (\MELL or $\DiLL_0$) quasi-proof-structures, exhibiting their sequential structure.

\begin{defi}[the category $\Scheduling$]
  \label{def:unwinding-paths}
  Let $\Scheduling$ be the category of \kl{schedulings} whose
  \begin{itemize}
  \item objects are lists $\Gamma = (\Gamma_1; \cdots; \Gamma_n)$ of lists of
    $\MELL$ formulas, with $n \geqslant 0$;
    
  \item arrows are freely generated by composition of the \intro{elementary
      schedulings} in \Cref{fig:elementary-schedulings}; 
	identities are the empty sequence of elementary schedulings, called \intro{empty scheduling}.
  \end{itemize}
  A \intro{scheduling} is any arrow $\sched \colon \Gamma \to \Gamma'$ in $\Scheduling$.  
  We write the composition of schedulings
  by juxtaposition in the diagrammatic
  order; so, if $\sched \colon \Gamma \to \Gamma'$ and
  $\sched' \colon \Gamma' \to \Gamma''$, then~$\sched\sched' \colon \Gamma\to \Gamma''$.
  If $\sched = a_1 \cdots a_k$ where $k \geqslant 0$ and the $a_i$'s are elementary schedulings, the \intro{length} of $\sched$ is $k$.
\end{defi}

Reading \Cref{fig:elementary-schedulings} left-to-right, $\Mix_{i}$ is the only elementary scheduling that changes (and increases) the number of lists in a list $\Gamma$ of lists of \MELL formulas. 
	The only elementary schedulings decreasing the number of \MELL formulas in the \kl{flattening} of $\Gamma$ are $\Ax_{i}, \One_i, \bot_{i}, \Weak_i$.

\begin{exa}
  \label{ex:barbara-path}
  $\parr_1 \, \parr_2 \, \parr_3 \, \otimes_1 \, \otimes_3 \, \Ex_1 \, \Ex_2 \, \Mix_{2} 
  \, \Ax_2 \, \Ex_2 \, \Mix_{2} \, \Ax_2 \, \Ax_2$ is a scheduling 
	from $\big((X \otimes Y^\bot) \parr ((Y \otimes Z^\bot) \parr (X^\bot \parr	Z))\big)$ 
	to $3\emptylist$, the list of 3 \kl{empty lists} of \MELL formulas.
\end{exa}



The category $\Scheduling$ acts on \kl{$\MELL$ quasi-proof-structures}, exhibiting a
sequential structure in their construction, by means of the rewrite rules in \Cref{fig:actions-all} (Definition \ref{def:unwinding}). 
To ease their reading, we make explicit the name of the active conclusions (to be intended as positive integers), we only draw the components affected by the rewrite rule and omit the ones left unchanged; 
\eg{}, if the affected component consists only of an $\Ax$-cell whose outputs are the conclusions $i\!-\!1$ and $i$, we write
\parbox[b][1.5\baselineskip]{1.05cm}{
	\(
	\scalebox{\smallproofnets}{
		\pnet{
			\pnformulae{
				\pnf[i]{$\,\,\,i\!-\!1$}~\pnf[j]{$i$}
			}
			\pnaxiom[ax]{i,j} [0.8]
			\pnsemibox{ax}
		}
	}
	\)
} ignoring the other components.
Given a list $\Gamma$ of lists of $\MELL$ formulas, $\MELLFunctor(\Gamma)$ is the set of \kl{$\MELL$ quasi-proof-structures} of type $\Gamma$.  

\begin{figure}[t]
	\vspace{\beforepn}
  \centering
  \begin{subfigure}[c]{0.4\textwidth}
    \begin{align*}
      \pnet{
      \pnsomenet[u]{}{2cm}{0.7cm}[at (-1,0.75)]
      \pnoutfromswitch{u.-153}[A]{$\Gamma_k\,$}
      \pnoutfrom{u.-135}[B]{$i$}
      \pnoutfrom{u.-65}[A']{$i\!+\!1$}
      \pnoutfromswitch{u.-27}[B']{$\,\,\Gamma_{k}'$}
      \pnsemibox{u,A,A',B,B'}
      }
      \pnrewrite[\Ex_{i}]
      \pnet{
      \pnsomenet[u]{}{2cm}{0.7cm}[at (-1,0.75)]
      \pnoutfromswitch{u.-153}[A]{$\Gamma_{k}\,$}
      \pnoutfrom{u.-110}[B]{$i\!+\!1$}
      \pnoutfrom{u.-45}[A']{$i$}
      \pnoutfromswitch{u.-27}[B']{$\,\,\Gamma_{k}'$}
      \pnsemibox{u,A,B,A',B'}
      }
    \end{align*}
    \vspace{-\baselineskip}
    \caption{Exchange}
    \label{fig:exc}
  \end{subfigure}\hfill%
  \begin{subfigure}[c]{0.4\textwidth}
    \begin{align*}
      \pnet{
      \pnsomenet[u]{}{1cm}{0.7cm}[at (-1,0.75)]
      \pnoutfromswitch{u.-125}[A]{$\Gamma_k\,$}
      \pnoutfrom{u.-55}[B]{$i$}
      \pnsomenet[u']{}{1cm}{0.7cm}[at (0.3,0.75)]
      \pnoutfrom{u'.-127}[A']{$i\!+\!1$}
      \pnoutfromswitch{u'.-52}[B']{$\,\,\Gamma_k'$}
      \pnsemibox{u,u',A,A',B,B'}
      }
      \pnrewrite[\Mix_{i}]
      \pnet{
      \pnsomenet[u]{}{1cm}{0.7cm}[at (-1,0.75)]
      \pnoutfromswitch{u.-125}[A]{$\Gamma_{k}$}
      \pnoutfrom{u.-55}[B]{$i$}
      \pnsomenet[u']{}{1cm}{0.7cm}[at (0.8,0.75)]
      \pnoutfrom{u'.-127}[A']{$i\!+\!1$}
      \pnoutfromswitch{u'.-52}[B']{$\,\,\Gamma_k'$}
      \pnsemibox{u,A,B}
      \pnsemibox{u',A',B'}
      }
    \end{align*}
    \vspace{-\baselineskip}
    \caption{Mix}
    \label{fig:mix}
  \end{subfigure}
  \begin{subfigure}[c]{0.3\textwidth}
  	\vspace*{-\baselineskip}
    \begin{align*}
      \pnet{
				\pnformulae{
					~~\pnf[i]{$i\!-\!1$}~\pnf[j]{$i$}
				}
	      \pnaxiom[ax]{i,j}
	      \pnsemibox{ax,i,j}
      } 
    	\pnrewrite[\Ax_{i}]
      \pnet{
				\pnformulae{
					\pnf[empty]{}
				}  
    		\pnsemibox{empty}   
      }
    \end{align*}
    \vspace{-\baselineskip}
    \caption{Hypothesis (\(\Ax, \One, \bot, \Weak\))}
    \label{fig:hypothesis}
  \end{subfigure}\hfill%
  \begin{subfigure}[c]{0.4\textwidth}
    \begin{align*}
      \pnet{
      \pnsomenet[u]{}{1.75cm}{0.7cm}[at (-2,0)]
      \pnoutfromswitch{u.-150}[k]{$\Gamma_{\!k}$}
			\pnoutfrom{u.-95}[A]{$A$}
			\pnoutfrom{u.-30}[B]{$\,A^\bot$}
      \pncut[cut]{A,B}
      \pnsemibox{u,cut}
      }
      \pnrewrite[\Cut^i]
      \pnet{
      \pnsomenet[u]{}{1.75cm}{0.7cm}[at (-2,0)]
      \pnoutfromswitch{u.-150}[k]{$\Gamma_{\!k}$}
      \pnoutfrom{u.-95}[B]{$i$}
      \pnoutfrom{u.-30}[A]{$\,i\!+\!1$}
      \pnsemibox{u,A,B}
      }
    \end{align*}
    \vspace{-3\baselineskip}
    \caption{Cut}
    \label{fig:cut}
  \end{subfigure}
  \begin{subfigure}[b]{0.4\textwidth}
    \begin{align*}
      \pnet{
      \pnsomenet[u]{}{1.75cm}{0.7cm}
      \pnoutfromswitch{u.-150}[k]{$\Gamma_{\!k}$}
			\pnoutfrom{u.-95}[A]{$A$}
			\pnoutfrom{u.-30}[B]{$B$}
      \pntensor{A,B}[i]{$i$}
      \pnsemibox{u,k,A,B,i}
      }
      \pnrewrite[\otimes_i]
      \pnet{
      \pnsomenet[u]{}{1.75cm}{0.7cm}
      \pnoutfromswitch{u.-150}[k]{$\Gamma_{\!k}$}
      \pnoutfrom{u.-95}[A]{$i$}
      \pnoutfrom{u.-30}[B]{$\,i\!+\!1$}
      \pnsemibox{u,k,A,B}
      }
    \end{align*}
    \vspace{-\baselineskip}
    \caption{Multiplicative (\(\otimes, \parr\))}
    \label{fig:tensor}
  \end{subfigure}\hfill%
  \begin{subfigure}[b]{0.4\textwidth}
    \begin{align*}
			\pnet{
				\pnsomenet[u]{}{1.75cm}{0.7cm}
				\pnoutfromswitch{u.-150}[k]{$\Gamma_{k}$}
				\pnoutfromswitch{u.-95}[C]{$A$}
				\pnoutfromswitch{u.-30}[D]{$A$}
				\pnexp{}{C*, D*}[i]{$i$}
				\pnsemibox{u,i}
			}
			\pnrewrite[\Contr_i]
			\pnet{
				\pnsomenet[u]{}{1.75cm}{0.7cm}
				\pnoutfromswitch{u.-150}[k]{$\Gamma_{k}$}
				\pnoutfromswitch{u.-95}[C]{$A$}
				\pnexp{}{C*}[i]{$i$}
				\pnoutfromswitch{u.-30}[D]{$A$}
				\pnexp{}{D*}[i1]{$i\!+\!1$}
				\pnsemibox{u,i,i1}
			}
		\end{align*}

    \vspace{-\baselineskip}
    \caption{Contraction}
    \label{fig:contraction}
  \end{subfigure}
  \begin{subfigure}[b]{0.4\textwidth}
    \begin{align*}
      \pnet{
      \pnsomenet[u]{}{1.25cm}{0.7cm}
      \pnoutfromswitch{u.-130}[k]{$\Gamma_k$}
      \pnoutfrom{u.-50}[A]{$A$}
      \pnexp{}{A}[i]{$i$}
      \pnsemibox{u,k,A,i}
      }
      \pnrewrite[\Der_i]
      \pnet{
      \pnsomenet[u]{}{1.25cm}{0.7cm}
      \pnoutfromswitch{u.-130}[k]{$\Gamma_{k}$}
      \pnoutfrom{u.-50}[A]{$i$}
      \pnsemibox{u,k,A}
      }
    \end{align*}
    \vspace{-\baselineskip}
    \caption{Dereliction}
    \label{fig:dereliction}
  \end{subfigure}\hfill%
  \begin{subfigure}[b]{0.4\textwidth}
    \begin{align*}
    \pnet{
    	\pnsomenet[u]{}{2cm}{0.7cm}
    	\pnbox{u}
    	\pnoutfrom{u.-25}[pi1c]{$A$}
    	\pnbag{}{pi1c}{$i$}
    	\pnoutfromswitch{u.-65}[ap]{$-$}
    	\pnexp{}{ap*}[?ap]{\small $\,\wn A_{p}$}
    	\pnoutfromswitch{u.-155}[a1]{$-$}
    	\pnexp{}{a1*}[?a1]{\small $\quad\wn A_1 \overset{p \geqslant 0}{\cdots}$}    
    	\pnsemibox{pi1c,ap,?ap,u}
    }
    \pnrewrite[\BoxR_{i}]
    \pnet{
    	\pnsomenet[u]{}{2cm}{0.7cm}
    	\pnoutfrom{u.-25}[pi1c]{$i$}[2.5cm][]
    	\pnoutfromswitch{u.-65}[ap]{$-$}
    	\pnexp{}{ap*}[?ap]{\small$\,\wn A_{p}$}
    	\pnoutfromswitch{u.-155}[a1]{$-$}
    	\pnexp{}{a1*}[?a1]{\small $\quad\wn A_1 \overset{p \geqslant 0}{\cdots}$}    
    	\pnsemibox{u,pi1c,ap,?ap}
    }
    \end{align*}
    \vspace{-\baselineskip}
    \caption{Box}
    \label{fig:box}
  \end{subfigure}
  \caption{\mbox{Action of elementary schedulings on \MELL quasi-proof-structures.}}
  \label{fig:actions-all}
\end{figure}

\begin{defi}[action of schedulings on $\MELL$ quasi-proof-structures]
  \label{def:unwinding}
  An \kl{elementary scheduling} {$a \colon \Gamma \to \Gamma'$} defines a
  relation $\mathop{\pnrewrite[a]} \subseteq \MELLFunctor(\Gamma) \times \MELLFunctor(\Gamma')$, called the \intro[action of elementary scheduling]{action} of $a$, as the smallest relation containing all the cases in
  \Cref{fig:actions-all}, with the following remarks:
  \begin{description}
  \item[exchange] $\pnrewrite[\Ex_i]$ swaps the order of two consecutive conclusions in the same component.
  
	\item[mix] given a \kl{component} $R$ with at least two conclusions $i$ and $i\!+\!1$ such that no conclusion $j \leqslant i$ is \kl{connected} in $\lVert R \rVert$ to any conclusion $j' \geqslant i\!+\!1$, $\!\pnrewrite[\!\Mix_i\!]$ splits $R$ into two components, one with the conclusions $j \leqslant i$, the other with the conclusions $j' \geqslant i\!+\!1$.\footnotemark
	\footnotetext{See \Cref{note:split-tree} for a more precise account. 
			Read in reverse, $\Mix_{i}$ is analogous to the mix rule in the sequent calculus for $\MELL$: it merges two components into one component putting together their conclusions.}
  
  \item[hypothesis] if  $a\in \{\Ax_{i}, \One_i, \bot_i, \Weak_i \}$, 
    $\pnrewrite[a]$ deletes an \kl{hypothesis} (\ie a \kl{cell} without inputs) that is the only
    cell in its \kl{component}, yielding an \kl[empty proof-structure]{empty} component.
    We have drawn the case $\Ax$ in
    \Cref{fig:hypothesis}, the other cases ($\One$, $\bot$, $\wn$) are similar, except for the number of conclusions.
    
  \item[cut] read in reverse, $\pnrewrite[\Cut^i]$ relates a quasi-proof-structure with two conclusions $i$
    and $i+1$ with the quasi-proof-structure where these two
    conclusions are cut by a $\cut$-\kl{cell} of \kl{depth} $0$. 
    This rule, from left to right, is nondeterministic (as there are many possible cuts).
  
  \item[multiplicative] if $a \in \{\otimes_i, \parr_i\}$, 
  $\pnrewrite[a]$ deletes a cell labeled by a multiplicative connective, $\otimes$ or $\parr$. 
  The new conclusions $i$ and $i+1$ inherit the order from the inputs of the erased cell.
  We have only drawn the case $\otimes$ in \Cref{fig:tensor}, the case $\parr$ is
    similar.  
    
  \item[contraction] $\pnrewrite[\Contr_i]$ splits a $\wn$-\kl{cell} with $h+k$ inputs into two
    $\wn$-cells (so, it duplicates a $\wn$-cell) with $h$ and $k$ inputs, respectively (\(h,k \geqslant 0\)). 
    The rule, from left to right, \mbox{is nondeterministic}.
    
  \item[dereliction] $\pnrewrite[\Der_i]$ only applies if the $\wn$-\kl{cell} (with $1$ input) does not
    shift a level in the \kl{box-forest}, \ie it is not just outside a \kl{box}
    (otherwise $\mathop{\pnrewrite[\Der_i]}$ would not yield a \mbox{$\MELL$ quasi-proof-structure)}.
    
  \item[box] $\pnrewrite[\BoxR_i]$ only applies if a box (and its border) is alone in its \kl{component}. 
  In particular, each $\wn$-cell represented in \Cref{fig:box} has at least one input (this is what ``$-$'' stands for).
  \end{description}
  
	The rewrite relation is extended by composition of relations to define, for every \kl{scheduling} $\sched \colon \Gamma \to \Gamma'$, a relation $\mathop{\pnrewrite[\sched]} \subseteq \MELLFunctor(\Gamma) \times \MELLFunctor(\Gamma')$, called the \intro[action of scheduling]{action} of $\sched$.\footnotemark
	\footnotetext{That is,
		given the schedulings $\sched \colon \Gamma \to \Gamma'$ and $\sched' \colon \Gamma' \to \Gamma''$ and the quasi-proof-structures $R \in \MELLFunctor(\Gamma)$ and $R'' \in \MELLFunctor(\Gamma'')$, $R \pnrewrite[\sched\sched'] R''$ if and only if there is $R' \in \MELLFunctor(\Gamma')$ such that $R \pnrewrite[\sched] R'$ and $R' \pnrewrite[\sched'] R''$ (often denoted by $R \pnrewrite[\sched] \, \pnrewrite[\sched'] R''$).
	The action of the \kl{empty scheduling} is the identity relation.}
\end{defi}
 
Except $\Mix_i$, $\Ex_i$ and $\Contr_i$, the \kl[action of elementary scheduling]{action} of an \kl{elementary scheduling} is a rewrite rule on \MELL quasi-proof-structures that destroys either a $\cut$-cell of \kl{depth} $0$ or a cell whose output is a conclusion. 
The types of the conclusions are updated according to \Cref{fig:elementary-schedulings}. 
The \kl[action of scheduling]{action of a scheduling} is just 
the composition of a number (possibly none) of these rewrite rules.

The action of schedulings cannot decrease the number of components of a \MELL quasi-proof-structure, and the action of $\Mix_{i}$ is the only one that increases this number.

Among the \kl[action of elementary scheduling]{actions} of \kl{elementary schedulings} in \Cref{fig:actions-all}, the cases $\Mix$ and  $\BoxR$ 
are the only ones that change 
the \kl{box-forest} of a quasi-proof-structure: $\Mix$ splits a tree in two distinct trees by splitting its root, while $\BoxR$ merges a root of a tree with a non-root vertex just above it.
For instance, the action $\Mix_i$ rewrites the \MELL proof-structure $S$ in \Cref{fig:pointed-proof-net} to the \MELL quasi-proof-structure $R$ in \Cref{fig:pointed-quasi-proof-net}, when $i$ is the conclusion of $S$ of type $\wn\oc \one$.

The way the \kl[action of elementary scheduling]{action of an elementary scheduling} is defined (Definition \ref{def:unwinding}) automatically rules out the possibility that a \MELL quasi-proof-structure rewrites to a \kl{module} that is not a \MELL quasi-proof-structure.
For instance, the elementary scheduling $\Der_1$ cannot \kl[action of elementary scheduling]{act} on the \MELL proof-structure $Q$ below, because it would yield a module $Q'$ that is not a \MELL quasi-proof-structure, as in $Q'$ it would be impossible to define a \kl{box-function}  that fulfills the constraints of Definition \ref{def:proof-structure} for \MELL (see \Cref{note:border});
$Q'$ is just a $\DiLL_0$ proof-structure.

\begin{center}
	\vspace{\beforepn}
	$Q \ = $
	\scalebox{0.8}{
		\pnet{
			\pnformulae{
				\pnf[i]{$X$}~\pnf[j]{$ X^\bot$}
			}
			\pnexp{}{i}[i']{$\wn X$}[1.2]
			\pnbag{}{j}[j']{$\oc X^\bot$}[1.2]
			\pnbox{i,j,ax}
			\pnaxiom[ax]{i,j}
		} 
	}
	$ \ \ \not\pnrewrite[\Der_{1}] \ $
	\scalebox{0.8}{
		\pnet{
			\pnformulae{
				\pnf[i]{$X$}~\pnf[j]{$ X^\bot$}
			}
			\pnbag{}{j}[j']{$\oc X^\bot$}
			\pnaxiom[ax]{i,j}
		}
	}
	$= \ Q'$
\end{center}

\kl[action of scheduling]{Actions of schedulings} can been seen as arrows in the category $\Rel$ of sets and relations.

\begin{defi}[functor $\MELLFunctor$]
  \label{def:FunctorMELL}
  We define a functor \intro[mell functor]{$\MELLFunctor \colon \Scheduling \to \Rel$} by:
    \begin{itemize}
    \item \emph{on objects:} for every list $\Gamma$ of lists of \MELL formulas, $\MELLFunctor(\Gamma)$ is the set of $\MELL$
      quasi-proof-structures of type $\Gamma$;
    \item \emph{on arrows:} for any $\sched \colon \Gamma \to \Gamma'$,
      $\MELLFunctor(\sched)$ is $\pnrewrite[\sched] \colon \MELLFunctor(\Gamma) \to \MELLFunctor(\Gamma')$ (Definition~\ref{def:unwinding}).
    \end{itemize}
\end{defi}

Lemmas \ref{prop:unwinding-co-functional}, \ref{lem:cell-or-mix}, \ref{lemma:contraction} and Proposition~\ref{prop:unwinding-to-empty} express some properties of our rewrite rules.

\begin{lem}[co-functionality]
  \label{prop:unwinding-co-functional}
  Let $\sched \colon \Gamma \to \Gamma'$ be a \kl{scheduling}.
  The action $\pnrewrite[\sched]$ of $\sched$ is a co-function 
  from $\MELLFunctor(\Gamma)$ to $\MELLFunctor(\Gamma')$,
  that is, 
  a function $ \op{\pnrewrite[\sched]} \colon \MELLFunctor(\Gamma') \to \MELLFunctor(\Gamma)$.
\end{lem}

Let $R, R'$ be \MELL quasi-proof-structures, $a$ be an \kl{elementary scheduling}. 
We say that:
\begin{itemize}
	\item $a$ \intro{applies} to $R$ if $R\pnrewrite[\!\Ex_{i_1}\!\dots\Ex_{i_n}a\!]R'$
	for some elementary schedulings $\Ex_{i_1},\dots,\Ex_{i_n}$ ($n \!\geqslant\! 0$);

	\item $R \pnrewrite[a] R'$ is \intro[nullary splitting for elementary scheduling]{nullary $\Contr$-splitting} if $a = \Contr_{i}$ and its action from $R$ to $R'$ splits the $\wn$-\kl{cell} of $R$ with output $i$ and $h \geqslant 0$ inputs into 
	two $\wn$-cells, one with $h$ inputs and one with $0$ inputs. 
\end{itemize}

\begin{lem}[applicability of elementary scheduling]
	\label{lem:cell-or-mix}
	Let $R$ be a non-empty $\MELL$ quasi-proof-structure.
	Then an \kl{elementary scheduling}
	$a\in \{\Mix_i, \Ax_i, \One_i, \bot_i, \otimes_i, \allowbreak \parr_i, \Contr_i, \Der_i, \Weak_i,
	\Cut^i,	\allowbreak \BoxR_{i}\}$ \kl{applies} to $R$, for some conclusion $i$ of $R$ (or input $i$ of a $\cut$-cell of \kl{depth} $0$ in $R$).
	
	No elementary scheduling applies to an \kl[empty quasi-proof-structure]{empty} quasi-proof-structure.
\end{lem}

To deconstruct a \MELL quasi-proof-structure $R$, nullary $\Contr$-splitting actions are problematic because	they increase any reasonable size for $R$. 
Luckily, the lemma below says roughly that, in \MELL, nullary $\Contr$-splitting actions are superfluous as means of destruction. 
	Remark \ref{rmk:nullary-splitting} will show that they are needed to achieve the naturality result (\Cref{thm:projection-natural}).
	
\begin{lemma}[splitting]
	\label{lemma:contraction}
	Let $R$ be a \MELL quasi-proof-structure with at least a conclusion $i$.
	If $a = \Contr_{i}$  and $R \pnrewrite[a] R'$ is \kl[nullary splitting for elementary scheduling]{nullary $\Contr$-splitting}, then for some \MELL quasi-proof-structure $S$:
	\begin{enumerate}
		\item	either $i$ is the output of a $\wn$-cell of $R$ with $k \!\geqslant\! 2$ inputs and $R \!\pnrewrite[\!a]\! S$ is not nullary $\Contr\!$-splitting;
		
		\item\label{p:cell-or-mix-der}
		or $i$ is the output of a $\wn$-cell of $R$ with $0/1$ inputs and in the same component as $i$ there are conclusions $i_0, \dots, i_n$ with $n \!\geqslant\! 0$ (possibly $i \!=\! i_j$ or $i_j \!=\! i_{j'}$ for some $0 \leqslant j,j' \!\leqslant n$; $i_0$ may also be an input of a $\cut$-cell of \kl{depth} $0$) such that $R \pnrewrite[\!\Ex_{i_1} \!\dots \Ex_{i_n}\!] S' \pnrewrite[\!a'\!] S$ with $a' \in \allowbreak \{\Mix_{i_0}, \allowbreak\BoxR_{i_0}, \allowbreak\otimes_{i_0}, \parr_{i_0}, \Cut^{i_0}, \Der_{i_0}, \Weak_{i_0}, \Contr_{i_0}\}$ and if $a' = \Contr_{i_0}$ then $S' \pnrewrite[\!a'\!] S$ is not nullary~$\Contr$-splitting.

	\end{enumerate}
\end{lemma}

Lemmas \ref{prop:unwinding-co-functional}, \ref{lem:cell-or-mix} and \ref{lemma:contraction} are proven by simple inspection of the rewrite rules in \Cref{fig:actions-all}.
The proof of Lemma \ref{prop:unwinding-co-functional} has a subtle case: 
given $R' \!\in \MELLFunctor(\Gamma')$, if $\sched = \Contr_i$ there is a (unique) $R \in \MELLFunctor(\Gamma)$ such that $R\pnrewrite[\sched] R'$ because any conclusion of $R'$ of type $\wn A$ is the output of a $\wn$-cell (thanks to atomic axioms, \Cref{note:atomic}; 
for non-atomic axioms,~see~\Cref{sect:nonatomic}).

Given a scheduling $\sched$, we say $R \pnrewrite[\sched] R'$ \intro[nullary splitting for scheduling]{has nullary $\Contr$-splittings} if there is an elementary scheduling $a$ such that $R \pnrewrite[\sched] R' = R \pnrewrite[\sched_{1}] R_1 \pnrewrite[a] R_2 \pnrewrite[\sched_2] R'$ and $R_1 \pnrewrite[a] R_2$ is nullary $\Contr$-splitting.

Every \MELL quasi-proof-structure can be unwound from its conclusions to an \kl[empty quasi-proof-structure]{empty} one by the action of some scheduling, and nullary $\Contr$-splittings are not needed for that.

\begin{prop}[normalization]
	\label{prop:unwinding-to-empty}
	Let $R$ be a $\MELL$ quasi-proof-structure of \kl[quasi-proof-structure type]{type} $\Gamma$.
	Then, there exists a \kl{scheduling} $\sched \colon \Gamma \to \emptylists$ such that $R \pnrewrite[\sched] \emptynets$, without nullary $\Contr$-splittings. 
\end{prop}

\begin{proof}
	Let the \emph{size} of a \MELL quasi-proof-structure $R$ be the triple $(p,q,r)$ where:
	\begin{itemize}
		\item $p$ is the (finite) multiset of the number of inputs of each $\wn$-cell in $R$;
		\item $q$ is the number of cells in $R$;
		\item $r$ is the (finite) multiset of the number of \kl{conclusions} of each \kl{component} of $R$.
	\end{itemize}
	Multisets are well-ordered as usual, triples are well-ordered lexicographically.
	
	By Lemmas \ref{lem:cell-or-mix} and \ref{lemma:contraction}, it suffices to show that such a size is left unchanged by $\pnrewrite[\!\Ex_{i}]$ and decreases with any action of $\Contr_i$ that is not nullary $\Contr$-splitting ($p$ decreases) and with any action in \Cref{fig:actions-all} other than the ones of $\Ex_{i}$ and $\Contr_{i}$ (the action of $\Mix_{i}$ leaves $p$ and $q$ unchanged and decreases $r$; the remaining actions leave $p$ unchanged and decrease $q$).
\end{proof}

Proposition \ref{prop:unwinding-to-empty} and the last statement in Lemma \ref{lem:cell-or-mix} imply that actions of elementary schedulings define a weakly normalizing rewrite system on \MELL quasi-proof-structures: 
the normal forms are exactly the empty quasi-proof-structures $n \emptynet$, for all $n > 0$. 
This rewrite system is not strongly normalizing because of the actions of $\Ex_i$ and, more significantly, $\Contr_i$ when are nullary $\Contr$-splitting. 
An example of an indefinitely extendable rewriting is~below.

\vspace{\beforepn}
\begin{center}
	\scalebox{0.8}{
		\pnet{
		\pnformulae{\pnf[X]{$\wn \one$}}
		\pnwn[cown]{X}
		}
	}
		\ \ $\pnrewrite[\Contr_1]$ \ \
	\scalebox{0.8}{
		\pnet{
			\pnformulae{\pnf[X]{$\wn \one$}~\pnf[X']{$\wn \one$}}
			\pnwn[cown]{X}
			\pnwn{X'}
		}
	}
		\ \ $\pnrewrite[\Contr_1]$ \ \
	\scalebox{0.8}{
		\pnet{
			\pnformulae{\pnf[X]{$\wn \one$}~\pnf[X']{$\wn \one$}~\pnf[X'']{$\wn \one$}}
			\pnwn[cown]{X}
			\pnwn{X'}
			\pnwn{X''}
		}
	}
		\ \ $\pnrewrite[\Contr_1]$ \ \
		\dots
\end{center}

The action of a scheduling $\sched$ such that $R \pnrewrite[\sched] \emptynets$, read 
backward, can be seen as a 
sequence of elementary steps to build the $\MELL$ quasi-proof-structure $R$ from an empty one $\emptynets$.
An example of a rewriting from some $R$ to $\emptynets$ is in \Cref{fig:longex2} (\Cref{sec:glueable}).
Intuitively, a \kl{scheduling} $\sched \colon (\Gamma) \to \emptylists\,$ ``encodes'' a way to build a \MELL proof-structure of type $\Gamma$ from scratch.


\section{Naturality of unwinding \texorpdfstring{\polyadic}{DiLL0} quasi-proof-structures}
\label{sec:naturality}

We show here how the actions of schedulings act on $\DiLL_0$ quasi-proof-structures, mimicking the behavior of the actions of schedulings on \MELL quasi-proof-structures seen in \Cref{sect:rules}.
\kl{Schedulings} are the same, the novelty is that $\DiLL_0$ quasi-proof-structures have no boxes and might have daimons. 
The bridge between the $\MELL$ and $\DiLL_0$ frameworks is given by the \kl{filled Taylor expansion}, which actually defines a \emph{natural transformation} (\Cref{thm:projection-natural}).

For $\Gamma$ a list of lists of $\MELL$ formulas, $\intro[qDiLL0]{\PolyPN}(\Gamma)$ is the
set of $\DiLL_0$ quasi-proof-structures of type $\Gamma$ (we assume they have the \emph{same} conclusions).  
For any set $X$, its powerset is 
$\PowerSet{X}$.

\begin{figure}[t]
  \newcommand{\figureNineRightColumnWidth}{0.5\textwidth}
	\vspace*{-3\baselineskip}
	\centering
	\begin{subfigure}[c]{0.4\textwidth}
		\scalebox{0.88}{
		\pnet{
			\pnformulae{
				\pnf[A]{$\Gamma_k$}~\pnf[B]{$i$}~\pnf[A']{$i\!+\!1$}~\pnf[B']{$\Gamma_{k}'$}
			}
			\pndaimon[dai]{A*,A',B,B'*}
			\pnsemibox{dai,A,A',B,B'}
		}
		$\pnrewrite[\Ex_{i}]$
		$\Bigg\{
		\pnet{
			\pnformulae{
				\pnf[A]{$\Gamma_k$}~\pnf[A']{$i\!+\!1$}~\pnf[B]{$i$}~\pnf[B']{$\Gamma_{k}'$}
			}
			\pndaimon[dai]{A*,A',B,B'*}
			\pnsemibox{dai,A,A',B,B'}
		}
		\Bigg\}$
		}
		\vspace{-.5\baselineskip}
		\caption{Daimoned exchange}
		\label{fig:exc-daimon}
	\end{subfigure}\hfill%
	\begin{subfigure}[c]{\figureNineRightColumnWidth}
		\scalebox{0.88}{
		\pnet{
			\pnformulae{
				\pnf[A]{$\Gamma_k$}~\pnf[B]{$i$}~\pnf[A']{$i\!+\!1$}~\pnf[B']{$\Gamma_{k}'$}
			}
			\pndaimon[dai]{A*,A',B,B'*}
			\pnsemibox{dai,A,A',B,B'}
		}
		$\pnrewrite[\Mix_{i}]$
		$\Bigg\{
		\pnet{
			\pnformulae{
				\pnf[A]{$\Gamma_k$}~\pnf[B]{$i$}~~\pnf[A']{$i\!+\!1$}~\pnf[B']{$\Gamma_{k}'$}
			}
			\pndaimon[dai]{A*,B}
			\pndaimon[dai']{A',B'*}
			\pnsemibox{dai,A,B}
			\pnsemibox{dai',A',B'}
		}
		\Bigg\}$
		}
		\vspace{-.5\baselineskip}
		\caption{Daimoned mix}
		\label{fig:mix-daimon}
	\end{subfigure}

	\begin{subfigure}[c]{0.4\textwidth}
		\vspace*{-\baselineskip}
		\scalebox{0.88}{
		\pnet{
			\pnformulae{
				\pnf[i]{$i\!-\!1$}~\pnf[j]{$i$}
			}
			\pndaimon[ax]{i,j}
			\pnsemibox{ax,i,j}
		} 
		$\pnrewrite[\Ax_{i}]$
		$\bigg\{
			\pnet{
				\pnformulae{
					\pnf[empty]{}
				}   
				\pnsemibox{empty}   
			}
		\bigg\}$
		}
		\vspace{-.5\baselineskip}
		\caption{Daimoned hypothesis (\(\Ax, \One, \bot, \Weak\))}
		\label{fig:hypothesis-daimon}
	\end{subfigure}\hfill%
	\begin{subfigure}[c]{\figureNineRightColumnWidth}
		\vspace*{-2\baselineskip}
		\scalebox{0.88}{
		\pnet{
			\pnformulae{
				\pnf[G]{$\Gamma_k$}
			}
			\pndaimon[dai]{G*}[1][.5]
			\pnsemibox{dai,G}
		} 
		$\pnrewrite[\Cut^i]$
		$\Bigg\{
		\pnet{
			\pnformulae{
				\pnf[G]{$\Gamma_k$}
				~\pnf[i]{$i$}~\pnf[j]{$i\!+\!1$}
			}
			\pndaimon[dai]{G*,i,j}
			\pnsemibox{dai,G,i,j}
		} 
		\Bigg\}$
		}
		\vspace{-.5\baselineskip}
		\caption{Daimoned cut ($\Gamma_{\!k} \neq \emptylist$)}
		\label{fig:cut-daimon}
	\end{subfigure}

	\vspace{-\baselineskip}
	\begin{subfigure}[b]{0.5\textwidth}
		\scalebox{0.88}{
			\pnet{
				\pnformulae{
					\pnf[G]{$\Gamma_k$}~\pnf[i]{$i$}
				}
				\pndaimon[dai]{G*,i}[1][.05]
				\pnsemibox{dai,G,i}
			} 
		$\pnrewrite[\otimes_i]$
		$\Bigg\{
			\pnet{
				\pnformulae{
					\pnf[G]{$\Gamma_k$}~\pnf[i]{$i$}~\pnf[j]{$i\!+\!1$}
				}
				\pndaimon[dai]{G*,i,j}
				\pnsemibox{dai,G,i,j}
			} 
		\Bigg\}$
		}
		\vspace{-.5\baselineskip}
		\caption{Daimoned multiplicative (\(\otimes, \parr\))}
		\label{fig:tensor-daimon}
	\end{subfigure}\hfill%
	\begin{subfigure}[b]{\figureNineRightColumnWidth}
		\scalebox{0.88}{
		\pnet{
			\pnformulae{
				\pnf[G]{$\Gamma_k$}~\pnf[i]{$i$}
			}
			\pndaimon[dai]{G*,i}[1][.05]
			\pnsemibox{dai,G,i}
		} 
		$\pnrewrite[\Contr_i]$
		$\Bigg\{
		\pnet{
			\pnformulae{
				\pnf[G]{$\Gamma_k$}~\pnf[i]{$i$}~\pnf[j]{$i\!+\!1$}
			}
			\pndaimon[dai]{G*,i,j}
			\pnsemibox{dai,G,i,j}
		} 
		\Bigg\}$
		}
		\vspace{-.5\baselineskip}
		\caption{Daimoned contraction}
		\label{fig:contraction-daimon}
	\end{subfigure}

	\vspace{-\baselineskip}
	\begin{subfigure}[b]{0.4\textwidth}
		\scalebox{0.88}{
		\pnet{
			\pnformulae{
				\pnf[G]{$\Gamma_k$}~\pnf[i]{$i$}
			}
			\pndaimon[dai]{G*,i}[1][.0]
			\pnsemibox{dai,G,i}
		} 
		$\pnrewrite[\Der_i]$
		$\Bigg\{
		\pnet{
			\pnformulae{
				\pnf[G]{$\Gamma_k$}~\pnf[i]{$i$}
			}
			\pndaimon[dai]{G*,i}[1][.0]
			\pnsemibox{dai,G,i}
		} 
		\Bigg\}$
		}
		\vspace{-.5\baselineskip}
		\caption{Daimoned dereliction}
		\label{fig:dereliction-daimon}
	\end{subfigure}\hfill%
	\begin{subfigure}[b]{\figureNineRightColumnWidth}
		\scalebox{0.88}{
		\pnet{
			\pnformulae{
				\pnf[a1]{$\quad\; \wn \! A_1\;\overset{p\geqslant 0}{\cdots}$}~~
				\pnf[ap]{$\wn \! A_p$}~   
				\pnf[i]{$i$}
			}
			\pndaimon[dai]{a1,ap,i}[1][.0]
			\pnsemibox{dai,a1,ap,i}
		} 
		$\pnrewrite[\BoxR_{i}]$
		$\Bigg\{
		\pnet{
			\pnformulae{
				\pnf[a1]{$\quad\; \wn \! A_1\;\overset{p\geqslant 0}{\cdots}$}~~
				\pnf[ap]{$\wn \! A_p$}~   
				\pnf[i]{$i$}
			}
			\pndaimon[dai]{a1,ap,i}[1][.0]
			\pnsemibox{dai,a1,ap,i}
		} 
		\Bigg\}$
		}
		\vspace{-.5\baselineskip}
		\caption{Daimoned box}
		\label{fig:box-daimon}
	\end{subfigure}

  \begin{subfigure}[b]{0.2\textwidth}
  	\scalebox{0.88}{
		\pnet{
			\pnformulae{
				\pnf[a1]{$\quad\; \wn \! A_1\;\overset{p\geqslant 0}{\cdots}$}~~
				\pnf[ap]{\!$\wn \! A_p$}~      
				\pnf[i]{$i$}
			}
			\pnwn[?a1]{a1}
			\pnwn[?ap]{ap}
			\pncown[!i]{i}
			\pnsemibox{a1,ap,i,?a1,?ap,!i}
		}
		$\pnrewrite[\BoxR_i]$
		$\Bigg\{
		\pnet{
			\pnformulae{
				\pnf[a1]{$\quad\; \wn \! A_1\;\overset{p\geqslant 0}{\cdots}$}~~
				\pnf[ap]{\!$\wn \! A_p$}~   
				\pnf[i]{$i$}
			}
			\pndaimon[dai]{a1,ap,i}[1][.0]
			\pnsemibox{dai,a1,ap,i}
		} 
		\Bigg\}$
		}
		\vspace{-.5\baselineskip}
		\caption{Empty box}
		\label{fig:empty-box}   
	\end{subfigure}\hfill%
	\begin{subfigure}[b]{\figureNineRightColumnWidth}
		\scalebox{0.85}{
			\pnet{
				\pnformulae{
					\pnf{$\dots$}}
				\pnsomenet[rhom]{$\rho_m$}{2cm}{0.5cm}[at (0,1.5)]
				\pnsomenet[rho1]{$\rho_1$}{1.4cm}{0.5cm}[at (0,0)]
				\pnoutfrom{rho1.-26}[rho1c]{$A$}
				\pnoutfrom{rhom.-18}[rhomc]{$A$}
				\pnbag[!]{}{rho1c,rhomc}{$i$}[1][0.2]
				\pnoutfromswitch{rho1.-82}[rho1ap]{$A_p$}
				\pnoutfromswitch{rhom.-40}[rhomap]{$A_p$}
				\pnoutfromswitch{rho1.-154}[rho1a1]{$A_1$}
				\pnoutfromswitch{rhom.-162}[rhoma1]{$A_1$}
				\pnexp[?a1]{}{rho1a1*,rhoma1*}{$\quad \wn \! A_{1} \,\overset{p \geqslant 0}{\cdots}$}[1][-0.2]
				\pnexp[?ap]{}{rho1ap*,rhomap*}{$\,\wn \! A_{p}$}		
				\pnsemibox{rho1,rhom,!,?a1}
			}
			$\pnrewrite[\BoxR_i]$
			$\left\{
			\pnet{
				\pnformulae{
					\pnf{$\dots$}}
				\pnsomenet[rho1]{$\rho_j$}{2cm}{0.5cm}[at (0,0)]
				\pnoutfrom{rho1.-20}[pi1c]{$i$}[2]
				\pnoutfromswitch{rho1.-45}[rho1ap]{$A_p$}
				\pnoutfromswitch{rho1.-160}[rho1a1]{$A_1$}
				\pnexp[?a1]{}{rho1a1*}{$\quad \wn \! A_{1} \,\overset{p\geqslant 0}{\cdots}$}
				\pnexp[?ap]{}{rho1ap*}{$\,\wn \! A_{p}$}
				\pnsemibox{rho1,?a1,?ap}
			}\right\}_{1 \leqslant j \leqslant m}$
		}
		\vspace{-.5\baselineskip}
		\caption{Non-empty box ($m > 0$)}
		\label{fig:polybox}
	\end{subfigure}

Lists $\Gamma_{\!k}$ and $\Gamma_{\!k'}$ may be empty, except in \Cref{fig:cut-daimon} since a $\maltese$-cell has at least  an output.
	\caption{\mbox{Action of elementary schedulings on $\DiLL_0$ quasi-proof-structures.}}
	\label{fig:actions-daimon}
\end{figure}

\begin{defi}[action of schedulings on \texorpdfstring{\polyadic}{DiLL0}
  quasi-proof-structures]
  \label{def:paths-on-polyadic}
  An \kl{elementary scheduling} $a \colon \Gamma \to \Gamma'$ defines a relation
  $\pnrewrite[a] \ \subseteq \kl{\PolyPN}(\Gamma) \times \PowerSet{\kl{\PolyPN}(\Gamma')}$, called the \intro[action of elementary scheduling on dill]{action} of $a$, as the smallest relation containing 
  all the rules in \Cref{fig:actions-all}---except \Cref{fig:box}, and with the same remarks\footnotemark---and in \Cref{fig:actions-daimon}.
 	\footnotetext{More precisely, if $\pnrewrite[a] \ \subseteq \kl{\PolyPN}(\Gamma) \times \PowerSet{\kl{\PolyPN}(\Gamma')}$, in the rewrite rule $\pnrewrite[a]$ in \Cref{fig:actions-all} the quasi-proof-structure on the right-hand-side must be replaced by the singleton set containing it.}
  We extend it to a relation $ \pnrewrite[a] \ \subseteq \PowerSet{\PolyPN(\Gamma)} \times\PowerSet{\PolyPN(\Gamma')}$ by the multiplication of the monad powerset $X \mapsto \PowerSet{X}$.\footnotemark
  \footnotetext{\label{foot:lifting}
  	More explicitly, given $\Pi \subseteq \PolyPN(\Gamma)$ and $\Pi' \subseteq \PolyPN(\Gamma')$, $\Pi \pnrewrite[a] \Pi'$ means that 
  	$\Pi' = \bigcup\{\rho' \subseteq \PolyPN(\Gamma') \mid \exists \, \rho \in \Pi : \rho \pnrewrite[a] \rho'\}$. 
  	Thus, $\rho \pnrewrite[a] \rho' \subseteq \Pi'$ for every $\rho \in \Pi$.
  	And, in particular, $\emptyset \pnrewrite[a] \emptyset$.
		Roughly, the same rewrite rule applies to the \emph{same} conclusion of \emph{each} element of a set of $\DiLL_0$ quasi-proof-structures. 
		Note that $\emptyset \neq \Pi \pnrewrite[a] \Pi'$ implies $\Pi' \neq \emptyset$, since there is no elementary scheduling $a$ such that $\rho \pnrewrite[a] \emptyset$.}
  The rewrite relation on
  $\PowerSet{\PolyPN(\Gamma)} \times\PowerSet{\PolyPN(\Gamma')}$ is extended by composition of relations to define, for any \kl{scheduling} $\sched \colon \Gamma \to \Gamma'$, a relation $\pnrewrite[\sched] \,\subseteq \PowerSet{\PolyPN(\Gamma)} \times \PowerSet{\PolyPN(\Gamma')}$, \mbox{the
  \intro[action of scheduling on dill]{action} of $\sched$.}
\end{defi}

Roughly, any rewrite rule in \Cref{fig:actions-daimon}---except
\Cref{fig:empty-box,fig:polybox}---mimics the behavior of the corresponding rule in
\Cref{fig:actions-all} using a $\maltese$-cell.
Since every elementary scheduling has at least two versions of its action, one in \Cref{fig:actions-all} and one (or two for $\BoxR_{i}$) in \Cref{fig:actions-daimon}, when $\Pi \pnrewrite[a] \Pi'$, some elements of $\Pi$ may rewrite according to one version of the action of $a$, others elements according to another version.
	In any case, all the elements of $\Pi'$ have the same conclusions, as the elements of $\Pi$ do.
	Types are updated according to \Cref{fig:elementary-schedulings}.
  
The \kl[action of elementary scheduling on dill]{action} in \Cref{fig:empty-box} acts on a component made of one \emph{co-weakening} (a $\oc$-cell without inputs) and some \emph{weakenings} ($\wn$-cells without inputs), it erases them and creates a $\maltese$-cell, whose outputs are all the conclusions of the component. 
If we just erased the co-weakening with output of type $\oc A$, we would lose the information about the type $A$; 
this is one of the reasons why a $\maltese$-cell is needed.
Intuitively, a co-weakening represents a box taken $0$ times, so there is no information about the content of the box, apart from its type.
The \kl{daimon} is a ``universal'' $\DiLL_0$ proof-structure representing this information plainly~missing.

The \kl[action of elementary scheduling on dill]{action} in \Cref{fig:polybox} separates the copies of a candidate for a box. It
 requires that, 
 in the \kl[undirected graph]{graph} of the component on the left of $\pnrewrite[\!\BoxR_i]$, $\rho_{j}$ is not connected to $\rho_{j'}$ for $j \neq j'$, except for the $\oc$-cell and the $\wn$-cells in the
conclusions.  Read in reverse, the rule associates with a non-empty finite set
of $\DiLL_0$ quasi-proof-structures $\{\rho_1, \dots, \rho_m\}$ the merging of
$\rho_1, \dots, \rho_m$, that is the $\DiLL_0$ quasi-proof-structure depicted on
the left of $\pnrewrite[\!\BoxR_i]$.
Intuitively, $\rho_1, \dots, \rho_m$ represent $m\!>\!0$ possible copies of a box, 
whose ``compatibility'' must be analyzed ``in parallel'': this is why the \kl[action of elementary scheduling on dill]{action} rewrites a single $\DiLL_0$ quasi-proof-structure to a (finite) \emph{set} of them.

Note that a $\maltese$-cell can be created only by the rewrite rule in \Cref{fig:empty-box}, and erased only by the \kl[action of elementary scheduling on dill]{action of elementary scheduling} called hypothesis in \Cref{fig:hypothesis-daimon}: $\Ax_{i}, \One_i, \bot_i, \Weak_i$\,.

\begin{example}
	\label{ex:box}
	Consider the $\DiLL_0$ proof-structures below ($X$ is a propositional variable).
	We have that $\{\rho_0,\rho_1,\rho_2\} \pnrewrite[\BoxR_{2}] \{\rho_0',\rho_1',\rho_2'\}$ because $\rho_0 \pnrewrite[\BoxR_{2}] \{\rho_0'\}$ (according to \Cref{fig:empty-box}), $\rho_1 \pnrewrite[\BoxR_{2}] \{\rho_1'\}$ (according to \Cref{fig:polybox}) and $\rho_2 \pnrewrite[\BoxR_{2}] \{\rho_1',\rho_2'\}$ (according to \Cref{fig:polybox}).
	\begin{center}
		
		\vspace*{-2.5\baselineskip}
		$\rho_0 =\!\! $
		\scalebox{0.8}{
			\pnet{
				\pnformulae{\\
					\pnf[x]{$\wn X$}~\pnf[bbx']{$\oc\oc X^{\!\perp}$}
				}
				\pnwn{x}
				\pncown{bbx'}
			}
		}
		\quad
		$\rho_1 =\!\! $
		\scalebox{0.8}{
			\pnet{
				\pnformulae{
					\\
					\pnf[x]{$\wn X$}~\pnf[bx']{$\oc X^{\!\perp}$}
				}
				\pnwn{x}
				\pncown{bx'}
				\pnbag{}{bx'}{$\oc\oc X^{\!\perp}$}
			}
		}
		\quad
		$\rho_2 =\!\! $
		\scalebox{0.8}{
			\pnet{
				\pnformulae{\pnf[x]{$X$}~\pnf[x']{$X^{\!\perp}$}
					\\
					~~\pnf[bx'']{$\oc X^{\!\perp}$}
				}
				\pnaxiom{x,x'}
				\pnexp{}{x}{$\wn X$}
				\pnbag{}{x'}[bx']{$\oc X^{\!\perp}$}
				\pncown{bx''}
				\pnbag{}{bx',bx''}{$\oc\oc X^{\!\perp}$}
			}
		}
		\quad 
		$\rho_0' =\!\! $
		\scalebox{0.8}{
			\pnet{
				\pnformulae{\\
					\pnf[x]{$\wn X$}~\pnf[bx']{$\oc X^{\!\perp}$}
				}
				\pndaimon{x,bx'}
			}
		}
		\quad
		$\rho_1' =\!\! $
		\scalebox{0.8}{
			\pnet{
				\pnformulae{\\
					\pnf[x]{$\wn X$}~\pnf[bx']{$\oc X^{\!\perp}$}
				}
				\pnwn{x}
				\pncown{bx'}
			}
		}
		\quad
		$\rho_2' =\!\! $
		\scalebox{0.8}{
			\pnet{
				\pnformulae{\pnf[x]{$X$}~\pnf[x']{$X^{\!\perp}$}
				}
				\pnaxiom{x,x'}
				\pnexp{}{x}{$\wn X$}
				\pnbag{}{x'}[bx']{$\oc X^{\!\perp}$}
			}
		}
	\end{center}
\end{example}

\begin{example}
	\label{ex:church}
	Consider the $\DiLL_0$ proof-structures $\pi$ and $\pi'$ below.
	\begin{center}
		\vspace{\beforepn}
		$\pi = $
		\scalebox{0.8}{
			\pnet{
				\pnformulae{~~\pnf[ba]{$\oc X$}~\pnf[na]{$X^{\bot}$}\\\\\\
					\\
					\pnf[na']{$\wn X^{\!\bot}$}~~~~~\pnf[a]{$X$}
				}
				\pnwn{na'}
				\pncown{ba}
				\pnaxiom{na,a}
				\pntensor{ba,na}[bana]{$\oc X \otimes X^{\bot}$}
				\pnexp{}{bana}{$\wn (\oc X \otimes X^{\bot})$}[1.1]
			}
			\qquad\qquad\qquad
		}
		$\pi' = $
		\scalebox{0.8}{
			\pnet{
				\pnformulae{~~~~~~\pnf[ba]{$\oc X$}~\pnf[na]{$X^{\bot}$}\\\\\\
					\\
					\pnf[na']{$\wn X^{\!\bot}\,$}~~~\pnf[wntens]{$\wn (\oc X \!\otimes\! X^{\!\bot})$}~~~~~~\pnf[a]{$X$}
				}
				\pnwn{na'}
				\pnwn{wntens}
				\pncown{ba}
				\pnaxiom{na,a}
				\pntensor{ba,na}[bana]{$\oc X \otimes X^{\bot}$}
				\pnexp{}{bana}{$\wn (\oc X \!\otimes\! X^{\!\bot})$}[1.1]
			}
		}
	\end{center}

		For every $n > 0$, the singleton $\{\pi\}$ rewrites to \(\{(n\!+\!1)\emptynet\}\) by the \kl[action of scheduling]{action} of the scheduling $\nu_n^1$ where, for all $k > 0$, $\nu_n^k \colon (\wn X^\bot, \wn (\oc X \otimes X^\bot), X; \underbrace{\emptylist; \cdots; \emptylist}_{k\!-\!1 \text{ times}}) \to (k\!+\!n)\emptylist$ is defined by:
		\begin{align*}
		\nu_1^k &= \Der_2 \, \otimes_2 \, \Mix_{2} \, \Ax_{4} \, \BoxR_2 \, \Der_1 \, \Ax_{2}
		&
		\nu_{n+1}^k &= \Contr_2 \, \Der_3 \, \otimes_3 \, \Mix_{3} \, \Ax_{5} \, \BoxR_3 \, \nu_{n}^{k+1} \,.
		\end{align*}
			
	\vspace{\beforepn}
	Indeed, we have the following rewritings, where
	$\pi_\maltese = \scalebox{0.75}{
		\pnet{
			\pnformulae{
				\pnf[na']{$\wn X^{\!\bot}$}~~\pnf[wntens]{$\,\,\wn (\oc X \!\otimes\! X^{\!\bot}\!)$}~~\pnf[ba]{$ X\!\!$}
			}
			\pndaimon[d]{na',wntens,ba}
		}
	}$
		
	
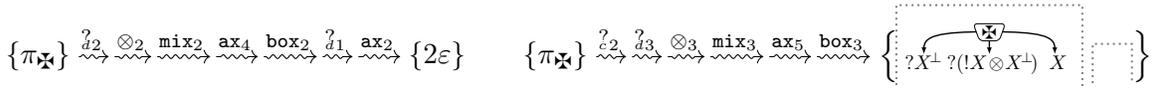
\vspace{-2.5\baselineskip}
	\begin{align*}
	\{\pi_\maltese\}
	\pnrewrite[\!\Der_2] \, \pnrewrite[\!\otimes_2] \, \pnrewrite[\!\Mix_{2}] \, \pnrewrite[\!\Ax_{4}] \,
	\pnrewrite[\!\BoxR_{2}] \,
	\pnrewrite[\!\Der_1] \, \pnrewrite[\!\Ax_{2}]
	\{2\emptynet\}
	\quad\quad
	\{\pi_\maltese\}
	\pnrewrite[\!\Contr_2] \,
	\pnrewrite[\!\Der_3] \, \pnrewrite[\!\otimes_3] \, \pnrewrite[\!\Mix_{3}] \, \pnrewrite[\!\Ax_{5}] \,
	\pnrewrite[\!\BoxR_3] 
	\Big\{\!\!
	\scalebox{0.75}{
		\pnet{
			\pnformulae{
				\pnf[na']{$\wn X^{\!\bot}$}~~\pnf[wntens]{$\,\,\wn (\oc X \!\otimes\! X^{\!\bot}\!)$}~~\pnf[ba]{$ X\!\!$}~~\pnfsmall[empty]{}
			}
			\pndaimon[d]{na',wntens,ba}
			\pnsemibox{empty}   
			\pnsemibox{d,na',wntens,ba}   
	}}
	\Big\} 
	\end{align*}
	\noindent and therefore, for every $n > 1$, the scheduling $\nu_n^1$ acts on $\{\pi\}$ in the way below
	
	\vspace{-2.5\baselineskip}
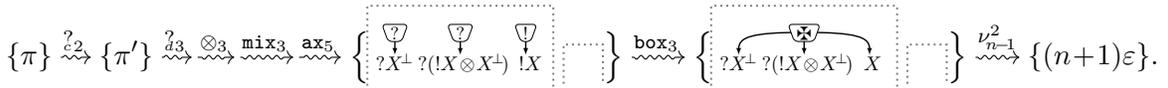
	\begin{align*}
	\{\pi\} 
	\pnrewrite[\!\Contr_2] 
	\{\pi'\} 
	\pnrewrite[\!\Der_3] \, \pnrewrite[\!\otimes_3] \, \pnrewrite[\!\Mix_{3}] \, \pnrewrite[\!\Ax_{5}] 
	\Big\{\!\!
		\scalebox{0.75}{
		\vspace*{\beforepn}
		\pnet{
		\pnformulae{
			\pnf[na']{$\wn X^{\!\bot}$}~~\pnf[wntens]{$\,\,\wn (\oc X \!\otimes\! X^{\!\bot}\!)$}~~\pnf[ba]{$\oc X\!\!\!$}~~\pnfsmall[empty]{}
		}
		\pnwn{na'}
		\pnwn{wntens}
		\pncown{ba}
		\pnsemibox{empty}   
		\pnsemibox{d,na',wntens,ba}  
	}}
	\Big\} 
	\pnrewrite[\!\BoxR_3] 
	\Big\{\!\!
	\scalebox{0.75}{
	\pnet{
		\pnformulae{
			\pnf[na']{$\wn X^{\!\bot}$}~~\pnf[wntens]{$\,\,\wn (\oc X \!\otimes\! X^{\!\bot}\!)$}~~\pnf[ba]{$ X\!\!$}~~\pnfsmall[empty]{}
		}
		\pndaimon[d]{na',wntens,ba}
		\pnsemibox{empty}   
		\pnsemibox{d,na',wntens,ba}   
	}}
	\Big\} 
	\pnrewrite[\!\nu_{\!n\!-\!1}^2\!] 
	\{(n\!+\!1)\emptynet\}.
	\end{align*}
\end{example}

\begin{figure}[!t]
	\vspace{\beforepn}
	\centering
			$\Pi = \Bigg\{\!\!\!\!\!$
			\scalebox{0.75}{
			\!\pnet{
				\pnformulae{
					\pnf[botax]{$X^\bot$}~\pnf[Xp]{$X^\bot$}~\pnf[1ax]{$X$}~\pnf[X]{$X$}
				}
				\pnaxiom[ax]{botax,1ax}[1][-0.3]
				\pnaxiom[axprop]{Xp,X}[1][0.3]
				\pnbag{}{1ax,X}[!1ax]{$\oc X$}
				\pnexp{}{Xp,botax}[?bot]{$\wn X^\bot$}
			}
			}
			$\!\!\!\!\!\!\Bigg\}$
			\quad
			$\Pi' = \Bigg\{\!\!\!\!\!$
			\scalebox{0.75}{
			\pnet{
				\pnformulae{
					\pnf[botax]{$X^\bot$}~\pnf[Xp]{$X^\bot$}~\pnf[1ax]{$X$}~\pnf[X]{$X$}
				}
				\pnaxiom[ax]{botax,1ax}[1][-0.3]
				\pnaxiom[axprop]{Xp,X}[1][0.3]
				\pnbag{}{1ax,X}[!1ax]{$\oc X$}
			}
			}
			$\!\!\!\!\!\!\Bigg\}$	
			\quad \ \
			$\Pi 
			\pnrewrite[\Contr_1]\,\pnrewrite[\Der_1]\,\pnrewrite[\Der_2]
			\Pi'$
			\quad
			$\Pi 
			\pnrewrite[\BoxR_2]$
			$\Bigg\{\!\!\!\!$
			\scalebox{0.75}{
			\!\pnet{
				\pnformulae{
					\pnf[botax]{$X^\bot$}~\pnf[1ax]{$X$}
				}
				\pnaxiom[ax]{botax,1ax}[1]
				\pnexp{}{botax}[?bot]{$\wn X^{\!\bot}$}
			}\!
			}
			$\!\!\!\!\!\!\Bigg\}
			\pnrewrite[\Der_{1}] \, \pnrewrite[\Ax_{2}]
			\{\emptynet\}
			$	
	\caption{Examples of actions of schedulings for $\DiLL_0$.}
	\label{fig:rewrite-diff}
\end{figure}

The rewriting on $\DiLL_0$ quasi-proof-structures behaves differently than on $\MELL$ quasi-proof-structures. 
In $\DiLL_0$,	Lemma \ref{lem:cell-or-mix} and Proposition \ref{prop:unwinding-to-empty} are false.
Indeed, in \Cref{fig:rewrite-diff} no elementary scheduling applies to the singleton $\Pi'$ of a non-empty $\DiLL_0$ proof-structure.
So, $\Pi \pnrewrite[\Contr_1] \, \pnrewrite[\Der_1] \, \pnrewrite[\Der_2] \Pi'$ and gets stuck, but 
the action of another scheduling rewrites $\Pi$ to $\{\emptynet\}$ 
(\Cref{fig:rewrite-diff} on the right). 
Thus, some rewritings from $\Pi$ end in $\{\emptynets\}$, others get stuck before.

\begin{defi}[functor \texorpdfstring{$\PPolyPN$}{PPolyPN}]
	\label{def:FunctorPPoly}
	We define a functor $\PPolyPN \colon \Scheduling \to \Rel$ by:
	\begin{itemize}
		\item\emph{on objects:} for every list $\Gamma$ of lists of \MELL formulas,
		$\PPolyPN(\Gamma) = \PowerSet{\PolyPN(\Gamma)}$, the powerset of the set $\PolyPN(\Gamma)$ of
		$\DiLL_0$ quasi-proof-structures of type $\Gamma$;
		\item\emph{on arrows:} for any $\sched \colon \Gamma \to \Gamma'$,
		$\PPolyPN(\sched)$ is $\pnrewrite[\sched] \colon \PPolyPN(\Gamma) \to \PPolyPN(\Gamma')$ (Def.~\ref{def:paths-on-polyadic}).
	\end{itemize}
\end{defi}

We can now compare the functors $\MELLFunctor$ and $\PPolyPN$ from $\Scheduling$ to $\Rel$.

\begin{thm}[naturality]
  \label{thm:projection-natural}
  The filled Taylor expansion defines a natural transformation
  \begin{align*}
    \NatTransf \colon \PPolyPN \Rightarrow \MELLFunctor \colon \Scheduling \to
    \Rel
  \end{align*}
  by: for $\Gamma$ a list of lists of \MELL 
  formulas, $(\Pi, R) \in \NatTransf_{\Gamma}$ iff $\Pi \subseteq \FatTaylor{R}$ and the type of $R$~is~$\Gamma$.  
\end{thm}
In other words, diagram \eqref{eq:naturality} below commutes for every scheduling
$\sched \colon \Gamma \to \Gamma'$.

\noindent
\!\!\!\!
\begin{minipage}{0.5\textwidth}
\begin{align}
\label{eq:naturality}
	\!\!\!\!\!
  \begin{tikzcd}
    \!\!\!\!\PPolyPN(\Gamma) \ar[rr,"\PPolyPN(\sched)"]
    \ar[d,swap,"\NatTransf_{\Gamma}"] \& \&
    \PPolyPN(\Gamma')\!\!\!\!\!\!\!\!\! \ar[d,"\NatTransf_{\Gamma'}"]
    \\
    \!\!\!\!\MELLFunctor(\Gamma) \ar[rr,"\MELLFunctor(\sched)"]\& \& \MELLFunctor(\Gamma')\!\!\!\!\!\!\!\!\!\!
  \end{tikzcd}
\end{align}
\end{minipage}
\quad
\begin{minipage}{0.23\textwidth}
	\small
	\begin{align}
	\label{eq:naturality-square}
	\begin{tikzcd}
	\Pi \ar[r,squiggly,"\sched"]
	\ar[d,swap,gray,"\NatTransf_{\Gamma}"] \& 
	\Pi'\!\!\! \ar[d,swap,"\NatTransf_{\Gamma'}"]
	\\
	\textcolor{gray}{R} \ar[r,squiggly,"\sched",gray]\& R'\!\!\!
	\end{tikzcd}
	\end{align}
\end{minipage}
\quad
\begin{minipage}{0.23\textwidth}
	\small
	\begin{align}
	\label{eq:naturality-square-bis}
	\begin{tikzcd}
	\Pi \ar[r,squiggly,gray,"\sched"]
	\ar[d,swap,"\NatTransf_{\Gamma}"] \& 
	\textcolor{gray}{\Pi'}\!\!\! \ar[d,gray,swap,"\NatTransf_{\Gamma'}"]
	\\
	R \ar[r,squiggly,"\sched"]\& R'\!\!\!
	\end{tikzcd}
	\end{align}
\end{minipage}

\noindent 
That is, given $\Pi \pnrewrite[\sched] \Pi'$ with
$\Pi'\subseteq \FatTaylor{R'}$, we can \emph{simulate backward} the rewriting to $R$
(here the co-functionality of actions on $\MELLFunctor$ expressed by
Lemma \ref{prop:unwinding-co-functional} comes in handy) so that $R \pnrewrite[\sched] R'$
and $\Pi\subseteq \FatTaylor{R}$ (square \eqref{eq:naturality-square}); 
and conversely, given $R \pnrewrite[\sched] R'$,
we can \emph{simulate forward} the rewriting for any $\Pi \subseteq \FatTaylor{R}$, so that
$\Pi \pnrewrite[\sched] \Pi'$ for some $\Pi'\subseteq \FatTaylor{R'}$ (square \eqref{eq:naturality-square-bis}).
In both squares, the \emph{same} kind of action is performed on $\DiLL_0$ and on \MELL.

\begin{proof}
  $\NatTransf$ is a family of arrows of $\Rel$ indexed by the objects in $\Scheduling$. 
	We have to show 
	that squares \eqref{eq:naturality-square}--\eqref{eq:naturality-square-bis} commute for any \kl{scheduling}
  $\sched \colon \Gamma\to \Gamma'$, and actually, it is enough to show 
  commutation for \kl{elementary schedulings}.  
  Let $a \colon \Gamma\to \Gamma'$ be an elementary scheduling.

  \begin{enumerate}
  \item \emph{Square \eqref{eq:naturality-square-bis}:} Let us prove that
    $\MELLFunctor(a) \circ \NatTransf_{\Gamma} \subseteq \NatTransf_{\Gamma'} \circ
    \PPolyPN(a)$.

    Let $(\Pi,R') \in \MELLFunctor(a) \circ \NatTransf_{\Gamma}$. Let
    $R = (\lvert R \rvert, \ForestT, \BoxFunction)\in \MELLFunctor(\Gamma)$ be a witness of composition, that is, an element such
    that $(\Pi,R) \in \NatTransf_{\Gamma}$ and $R\pnrewrite[a] R'$.
    If $\Pi = \emptyset$ then we are done because $\Pi \pnrewrite[a] \emptyset$ (see \Cref{foot:lifting}) and $\emptyset \subseteq \FatTaylor{R'}$. 
    So, we suppose that $\Pi \neq \emptyset$.

    Let $(t,h)$ be a \kl{thick subforest} of $\ForestT$, let
    $r_{1},\ldots,r_{n}$ be some roots in $\ForestT$ (and in $t$, by Remark \ref{rmk:roots}), and let $\rho_{r_{1}\ldots r_{n}} \in \Pi$ be the element of the \kl{filled Taylor
    expansion} $\FatTaylor{R}$ of $R$ associated with $(t,h)$ and $r_{1},\ldots,r_{n}$ (\ie~$\rho_{r_{1}\ldots r_{n}}$ is the \kl{emptying on} $r_{1},\ldots,r_{n}$ of the element of $\Taylor{R}$ obtained via the thick subforest $(t,h)$ of $\ForestT$, see Definition \ref{def:FilledTaylor}).
  	By Remark \ref{rmk:conclusions-emptying}, we can identify the conclusions of $R$ and of $\rho_{r_{1}\ldots r_{n}}$.
  	The case $a = \Ex_{i}$ is trivial since it just swaps the order of two consecutive conclusions. 
  	Other cases for $a$:
    \begin{itemize}
    \item If $a = \Mix_{i}$, 
    	let $r$ be the root in $\ForestT$ corresponding to the conclusion $i$ (\ie $\BoxFunction_{F}(i)$ is the output of $r$) and let $p,\ldots, p+k$ be the conclusions of the component $S$ of $R$ containing $i$. 
    	In the \kl{undirected graph} $\lVert S \rVert$, none of the conclusions $p, \dots, i$ is  \kl{connected} to any conclusion $i+1, \dots, p+k$.
      By Definitions~\ref{def:taylor} and~\ref{def:FilledTaylor},
      as $\rho_{r_{1}\ldots r_{n}} \in \Pi \subseteq \FatTaylor{R}$, 
			$r$ is a root in the box-forest of $\rho_{r_{1}\ldots r_{n}}$, and $p,\ldots, p+k$ are  the conclusions of $\rho_{r_{1}\ldots r_{n}}$ corresponding to $r$.
			There are two cases, according to Remark~\ref{rmk:connection}:
      \begin{itemize}
      \item the \kl{connected components} of $i$ and $i\!+\!1$ are disjoint in the undirected graph
        $\lVert \rho_{r_{1}\ldots r_{n}} \rVert$;
      \item $i$ and $i\!+\!1$ belong to the same connected component of  $\lVert \rho_{r_1 \dots r_n} \rVert$, 
			and their component in $\rho_{r_{1}\ldots r_{n}}$ is a daimon with conclusions $p, \dots,  i, i\!+\!1, \dots, p\!+\!k$.
      \end{itemize}
      In both cases the rule $\Mix_{i}$ is also applicable to
      $\rho_{r_{1}\ldots r_{n}}$, yielding a $\DiLL_0$ quasi-proof-structure
      $\rho'$. 
			The box-forest $\ForestT'$ of $R'$ is obtained from the box-forest $\ForestT$ of $R$ by root-splitting\footnotemark the tree with root $r$ in two trees with roots $r_1'$, $r_2'$.
			Let $t'$ be a forest obtained from $t$ in a similar way.
			\footnotetext{\label{note:split-tree}It means that $\ForestT = \ForestT'$ except that in $\ForestT'$ the root $r$ with its output is replaced by the roots $r_1'$ and $r_2'$ with their output, and that the inputs of $r$ are split into the inputs of $r_1'$ and the inputs of $r_2'$.} 
			Let $h' \colon t'\to \ForestT'$ be the directed graph morphism such that $h_V'(r_1') = r_1'$, $h_V'(r_2') = r_2'$ and $h_V'(v) = h_V(v)$ for all 
			$v \in V_{\ForestT} \cap V_{\ForestT'}$.
      We verify that $\rho'$ is the element of 
      $\FatTaylor{R'}$ associated with the thick subforest $(t',h')$ of $\ForestT'$.
       
			\item If $a \in \{\otimes_i, \parr_i, \Der_i,\Contr_i, \Weak_i, \One_i, \bot_i, \Ax_{i}\}$, let $k$ be such that the rule $a$ acts on the conclusions $i-k,\dots,i$ of a component of $R$, and let
			$\ell$ be the type of the cell in $R$ whose output is the conclusion
			$i$ (or whose outputs are the conclusions $i-1$ and $i$, if $\ell = \Ax$). 
			In $\rho_{r_{1}\ldots r_{n}}$ 
			either there is a $\ell$-cell with $\ell \neq \Ax$ whose output is the conclusion $i$, or there is a $\Ax$-cell whose outputs are the conclusions $i-1$ and $i$, or the component containing the conclusion $i$ is a daimon with conclusions $i-k, \dots, i$.
			Clearly $a$ is applicable to $\rho_{r_{1}\ldots r_{n}}$, which yields a $\DiLL_0$
			quasi-proof-structure $\rho'$.
			
			The box-forest $\ForestT'$ of $R'$ is $\ForestT$, the one of $R$. 
			We verify that $\rho'$ is the element of the filled Taylor expansion $\FatTaylor{R'}$ of $R'$ associated with the thick subforest $(t,h)$ of $\ForestT'$.

    \item If $a = \Cut^i$, let $c$ be the $\Cut$-cell of \kl{depth} $0$ in $R$ to which $a$
      is applied. As in the case above, the $\Cut$-cell $c$ has either one image in
      $\rho_{r_{1}\ldots r_{n}}$ or is represented by a $\maltese$-cell (with at least one conclusion). In both
      cases, $\Cut^i$ is applicable to $\rho_{r_{1}\ldots r_{n}}$, yielding
      $\rho'$.

      The box-forest $\ForestT'$ of $R'$ is $\ForestT$, the one of $R$. 
      We verify that $\rho'$ is the element of the filled Taylor expansion $\FatTaylor{R'}$ of $R'$ associated with the thick subforest $(t,h)$ of $\ForestT'$.
      
    \item If $a = \BoxR_{i}$, consider the component of $R$ whose conclusions are $i\!-\!k,\dots,i$ with $k \geqslant 0$.
				In the component $\pi$ of $\rho_{r_1 \dots r_n}$ with conclusions $i\!-\!k, \dots, i$, there are three cases:
      \begin{enumerate}[ref=\alph*]
			\item\label{p:box-daimoned} $\pi$ is a daimon whose conclusions are $i\!-\!k,\dots,i$; 
				so, the rule in \Cref{fig:box-daimon} applies;
			\item\label{p:box-empty} $\pi$ consists of a $\oc$-cell with output tail $i$ and no inputs; and of $k$ $\wn$-cells without inputs;
			therefore, the rewrite rule in \Cref{fig:empty-box} applies;
			\item\label{p:box-nonempty} $\pi$ has a $\oc$-cell $c$ with output tail $i$ and at least an input, $k$ $\wn$-cells $c_1, \dots, c_k$ immediately above each of the other $k$ conclusions, and other cells that can be split in $m$ pairwise disjoint sub-graphs of $\lVert \pi' \rVert$, where $m \geqslant 1$ and $\pi'$ is $\pi$ without $c, c_1, \dots, c_k$ and their output;
			therefore, the rewrite rule in \Cref{fig:polybox} applies.
      \end{enumerate}
      In any case, $\BoxR_i$ applies to $\rho_{r_1 \dots r_n}$, yielding a family
      	$\rho'_1,\dots,\rho'_{m}$ of $\DiLL_0$ proof-structures, where $m = 1$ in cases \eqref{p:box-daimoned} and \eqref{p:box-empty}, while $m \geqslant 1$ in case \eqref{p:box-nonempty}.
      
      The box-forest $\ForestT'$ of $R'$ is obtained from the box-forest
      $\ForestT$ of $R$ by erasing the root $b$ corresponding to the conclusions $i\!-\!k,\dots,i$ (\ie, $\BoxFunction_{F}(i\!-\!k) = \dots = \BoxFunction_{F}(i)$ is the output of $b$):
      the new root $b'$ of this tree of $\ForestT'$ is the unique vertex immediately above $b$ in $\ForestT$.
    	We have $h^{-1}(b')=\{b'_{1},\ldots,b'_{m}\}$, and $m$ trees
      $t'_{1},\ldots,t'_{m}$, where $b'_{j}$ is the root of
      $t'_{j}$. The directed graph morphisms $h'_{j} \colon t'_{j}\to \ForestT'$ are defined
      accordingly.
      We verify that $\rho'_j$ is the element of $\FatTaylor{R'}$ associated~with the thick subforest $(t'_j,h'_{j})$ of $\ForestT'$.
    \end{itemize}
  
  \item \emph{Square \eqref{eq:naturality-square}:} Let us prove that
    $\NatTransf_{\Gamma'} \circ \PPolyPN(a) \subseteq \MELLFunctor(a) \circ
    \NatTransf_{\Gamma}$.
    
    Let $(\Pi,R') \in \NatTransf_{\Gamma'} \circ \PPolyPN(a)$ with $R' = (|R'|, \ForestT', \BoxFunction')$. Let $\Pi'$ be a
    witness of composition, \ie a set in $\PPolyPN(\Gamma')$ such that 
    $\Pi \pnrewrite[a] \Pi'$ and $(\Pi',R') \in \NatTransf_{\Gamma'}$.

    We want to exhibit a $\MELL$ quasi-proof-structure $R$ such that
    $R \pnrewrite[a] R'$ and $\Pi$ is a part of the filled Taylor expansion $\FatTaylor{R}$ of
    $R$. By co-functionality of $\MELLFunctor(a)$
    (Lemma \ref{prop:unwinding-co-functional}), we have a candidate for such an $R$:
    the preimage of $R'$ by this co-functional relation $\pnrewrite[a]$. 
    In other terms, if $a \colon \Gamma \to \Gamma'$ and
    $R'\in\MELLFunctor(\Gamma')$, then there exists a unique $R = (|R|, \ForestT, \BoxFunction) \in\MELLFunctor(\Gamma)$ such that
    $R\pnrewrite[a] R'$ (Lemma \ref{prop:unwinding-co-functional}). 
    We only have to check that $\Pi\subseteq \FatTaylor{R}$.
    If $\Pi = \emptyset$ then trivially $\Pi \subseteq \FatTaylor{R}$. 
    So, we suppose that $\Pi \neq \emptyset$.

		Let $\rho \in \Pi$ and $\{\rho'_1, \dots, \rho'_m\} \subseteq \Pi'$ such that $\rho \pnrewrite[a] \{\rho'_1, \dots, \rho'_m\}$. 
		In all cases except $a = \BoxR_i$, we have $m = 1$ \ie $\{\rho'_1, \dots, \rho'_m\} = \{\rho'_1\}$.
		For $a \neq \BoxR_i$, let $(t',h')$ be a \kl{thick subforest} of $\ForestT'$, let $r'_{1},\ldots,r'_{k}$ be some \kl{roots} in $\ForestT'$ (and in $t'$, by Remark \ref{rmk:roots}), and let
		$\rho'_1 = \rho'_{r'_{1}\ldots r'_{k}} \in \Pi'$ be the element of the \kl{filled Taylor expansion} $\FatTaylor{R'}$ of $R'$ associated with $(t',h')$ and $r'_{1},\ldots,r'_{k}$ (\ie $\rho'_{r'_{1}\ldots r'_{k}}$ is the \kl{emptying on} $r'_{1},\ldots,r'_{k}$ of the element of $\Taylor{R'}$ obtained via the thick subforest $(t',h')$ of $\ForestT'$, see Definition \ref{def:FilledTaylor}).
		By Remark \ref{rmk:conclusions-emptying}, we can identify the conclusions of $R'$ and of $\rho'_{r'_{1}\ldots r'_{k}}$.
		The case $a = \Ex_{i}$ is trivial since it just swaps the order of two consecutive conclusions. 
		Other cases for $a$:
		\begin{itemize}
			\item If $a \in \{\Cut^i, \otimes_i, \parr_i, \Der_i, \Contr_i, \Weak_i, \One_i, \bot_i, \Ax_{i} \}$ then $\ForestT = \ForestT'$; let $t = t'\!$, $h = h'\!$.
			We verify that $\rho$ is the element of $\FatTaylor{R}$ associated with $r'_{1},\ldots,r'_{k}$ and the thick subforest~$(t,h)$~of~$\ForestT\!$.
			
			\item If $a = \Mix_{i}$, let $s'_1, s'_2$ be the roots in $\ForestT'$ (and in $t'$, by Remark~\ref{rmk:roots}) corresponding to the conclusions $i, i\!+\!1$ in $R'$, respectively (\ie $\BoxFunction_{F}'(i)$ is the output of $s_1'$, $\BoxFunction_{F}'(i\!+\!1)$ is the output of $s_2'$).
			Let $s$ be the root of $\ForestT$ corresponding to $i, i\!+\!1$ (\ie $\BoxFunction_{F}(i) = \BoxFunction_{F}(i\!+\!1)$ is the output of $s$). 
			We define the thick subforest $(t,h)$ of~$\ForestT$: $t$ is the forest obtained from $t'$ by merging its roots $s'_1$ and $s'_2$ into the root $s$, and $h \colon t \to \ForestT$ is such that $h_V(s) = s$ and $h_V(v) = h'_V(v)$ for all $v \in V_t \cap V_{t'}$. 
			We verify that $\rho$ is the element of $\FatTaylor{R}$ associated with $r'_{1},\ldots,r'_{k}$ and the thick subforest $(t,h)$~of~$\ForestT\!$.
			
			\item If $a = \BoxR_{i}$, we describe the rule for non-empty box 
			(\Cref{fig:polybox}); 
			the easier rules for daimoned and empty boxes (\Cref{fig:box-daimon,fig:empty-box}) are left to the reader.
			Let $(t'_1, h'_{1}), \dots, \allowbreak (t'_m, h'_{m})$ 
			be \kl{thick subforests} of $\ForestT'$, let $r'_{1},\ldots,r'_{k}$ be some \kl{roots} in $\ForestT'$ (and in $t'_1, \dots, t'_m$, by Remark \ref{rmk:roots}).
			For all $1 \leqslant j \leqslant m$, let 
			$\rho'_j$ be the element of $\FatTaylor{R'}$ associated with $(t'_j, h'_{j})$ and $r'_{1},\ldots,r'_{k}$ (\ie $\rho'_j$ is the \kl{emptying on} $r'_{1},\ldots,r'_{k}$ of the element of $\Taylor{R'}$ obtained via the thick subforest $(t'_j,h'_j)$ of $\ForestT'$, see Definition \ref{def:FilledTaylor}). 
			The forests $t_j'$'s differ from each other only in the rooted tree corresponding to the conclusion $i$ of the $\rho_{j}'$'s.
			Let $t$ be the forest made of all the rooted trees on which the forests $t_1', \dots, t'_m$ do not differ, and of the union of the rooted trees on which the forests differ, joined by a new root;
			define $h$ from $h'_1, \dots, h'_m$ accordingly. 
			We verify that $\rho$ is the element of $\FatTaylor{R}$ associated with $r'_{1},\ldots,r'_{k}$ and the thick subforest $(t,h)$ of $\ForestT$.
			\qedhere
		\end{itemize}               
            
  \end{enumerate}
\end{proof}

The rewriting system, the \kl{filled Taylor expansion} (which generalizes the usual Taylor expansion) and the notion of \kl{quasi-proof-structure} (which generalizes usual proof-structures) are designed to obtain the naturality result (Theorem~\ref{thm:projection-natural}), as shown in the example below.

\begin{example}[the need for quasi]
	Extending proof-structures to \kl{quasi-proof-structures}, with dashed boxes at the root level to separate their \kl{components}, is crucial to achieve naturality (Theorem~\ref{thm:projection-natural}) when dealing with a box in a \MELL \kl{proof-structure} \(R\) copied zero times in an element of its Taylor expansion. Take for instance $\rho_1$, $\rho_2$ and $R$ below.
      
  \begin{center}
    $\rho_1$
    \quad\qquad\quad
    $\rho_2$
    \quad\qquad\quad
    $R$
    \quad\qquad\quad
    $R'$
    \quad\qquad\qquad\quad
    $R_0$
    \quad\qquad\qquad
    $R_0'$
    \qquad
    
    \vspace{-2\baselineskip}
    \scalebox{0.8}{
    \pnet{
	    \pnformulae{
	    	~\pnf[in]{$\mathbf{1}$}\\[0.3cm]
	    	\pnf[out]{$\mathbf{1}$}
	    }
	    \pnone{in}
	    \pnone{out}
	    \pnbag{}{in}{$\oc \mathbf{1}$}
    }}
  	\qquad
  	\scalebox{0.8}{
    \pnet{
	    \pnformulae{
	    	\pnf[out]{$\mathbf{1}$}~\pnf[in]{$\oc \mathbf{1}$}
	    }
	    \pnone{out}
	    \pncown{in}
    }}
  	\qquad
  	\scalebox{0.8}{
    \pnet{
	    \pnformulae{
		    ~\pnf[in]{$\mathbf{1}$}\\[0.3cm]
		    \pnf[out]{$\mathbf{1}$}
	    }
	    \pnone[in1]{in}
	    \pnone{out}
	    \pnbag{}{in}{$\oc \mathbf{1}$}[1.1]
	    \pnbox{in1,in}
    }}
  	\qquad
  	\scalebox{0.8}{
  	\pnet{
  		\pnformulae{
  			\pnf[out]{$\mathbf{1}$}~\pnf[in]{$\mathbf{1}$}
  		}
  		\pnone[in1]{in}
  		\pnone{out}
  	}}
  	\quad\qquad
  	\scalebox{0.8}{
  	\pnet{
  		\pnformulae{
  			~~\pnf[in]{$\mathbf{1}$} \\[0.3cm]
  			\pnf[out]{$\mathbf{1}$}
  		}
  		\pnone[in1]{in}
  		\pnone[out1]{out}
  		\pnbag{}{in}[b]{$\oc \mathbf{1}$}[1.1]
  		\pnbox{in1,in}
  		\pnsemibox{in,in1,b}
  		\pnsemibox{out,out1}
  	}}
  	\qquad
  	\scalebox{0.8}{
  	\pnet{
  		\pnformulae{
  			\pnf[out]{$\mathbf{1}$}~~\pnf[in]{$\mathbf{1}$}
  		}
  		\pnone[in1]{in}
  		\pnone[out1]{out}
  		\pnsemibox{in,in1}
  		\pnsemibox{out,out1}
  	}}
  \end{center}
  We have $\{\rho_1,\rho_2\} \subseteq \Taylor{R}$.
  If the action of a scheduling $a$ were a binary relation $\pnrewrite[a]$ on \MELL proof-structures and not on \MELL quasi-proof-structure---ignoring the dashed boxes at the root level in \Cref{fig:actions-all} to separate their \kl{components}---we would have $R \pnrewrite[\BoxR_{2}] R'$. 
  To make square \eqref{eq:naturality-square-bis} commute (as required by naturality), we observe the following facts: 
  \begin{enumerate}
  	\item the action of $\BoxR_{2}$ on $\rho_2$ cannot simply erase the $\oc$-cell, otherwise we would lose the match with the type of $R'$;
  	so, the action of $\BoxR_{2}$ on $\rho_2$ must create a $\maltese$-cell of type $\One$;
  	\item thus, when $\{\rho_1,\rho_2\} \pnrewrite[\BoxR_{2}] \Pi'$, the relation between $\Pi'$ and $R'$ is not $\Pi' \subseteq \Taylor{R'}$, because of the $\maltese$-cell;
  	this is why we introduced the filled Taylor expansion, so that $\Pi' \subseteq \FatTaylor{R'}$;
  	\item the two $\One$-cells in $R'$ must be handled in two different ways;
  	indeed, the second one comes from the content of the box in $R$, and can be erased by the Taylor expansion of $R$ by taking 0 copies of the box (as done in $\rho_2$), while the first $\One$-cell cannot; 
  	to mark this difference, and to make explicit the scope of the $\maltese$-cell introduced by the action of $\BoxR_{2}$ on $\rho_2$, we split the two conclusions of $R$ into two distinct \kl{components} via $R \pnrewrite[\Mix_{1}] R_0$. 
  \end{enumerate}  
	
	According to our rewrite rules, $R \pnrewrite[\Mix_{1}] R_0 \pnrewrite[\BoxR_{2}] R_0'$ and this rewriting commutes with the \kl{filled Taylor expansion} on quasi-proof-structures, as required by square \eqref{eq:naturality-square-bis} for naturality.
\end{example}


\section{\Gluability of \texorpdfstring{\polyadic}{DiLL} quasi-proof-structures and inhabitation}
\label{sec:glueable}

The properties of actions of schedulings allow us to prove two of our main original results: 
\begin{enumerate}
	\item a characterization of the sets of $\DiLL_0$ proof-structures that are in the Taylor expansion of some $\MELL$ proof-structure (\emph{inverse Taylor expansion problem}, \Cref{thm:characterization});
	
	\item a characterization of the lists of \MELL formulas inhabited by some \kl{cut-free} \MELL proof-structure (\emph{type inhabitation problem} for cut-free \mbox{$\MELL$ proof-structures, \Cref{thm:inhabitation}}).
\end{enumerate}

Both characterizations are \emph{constructive}: when we claim the existence of an object (a \MELL proof-structure with a certain property), we can actually construct a witness of it, through our rewrite rules. 
Thus, our characterizations can be seen as \emph{nondeterministic} procedures, because the rewritings to build witnesses are not unique.
Our procedures also \emph{decide} the problems above in special cases of interest (\Cref{thm:decisionFinite}, Proposition \ref{fact:decide-without-contractions}).

\subsection{\Gluability}
\label{sec:glueability}
We introduce the \emph{\gluability criterion} to solve the inverse Taylor expansion problem.
We actually solve it in two different domains. 
We look for a solution:
\begin{enumerate}
	\item in the set of \MELL proof-structures (\Cref{thm:characterization}.\ref{p:characterization-with-cut});
	\item in the set of \kl{cut-free} \MELL proof-structures (\Cref{thm:characterization}.\ref{p:characterization-cut-free}).
\end{enumerate}

We say that a scheduling is \intro[cut-free scheduling]{cut-free} if does not contain the \kl{elementary scheduling} $\Cut^i$.

\begin{lemma}[cut-free scheduling]
	\label{lemma:cut-free-scheduling}
	Let $R$ be a \MELL quasi-proof-structure and $\sched$ be a \kl{scheduling} such that $R \pnrewrite[\sched] \emptynets$.
	We have that $R$ is \kl{cut-free} if and only if $\sched$ is \kl[cut-free scheduling]{cut-free}.
\end{lemma}

\begin{proof}
	Let $R$ be \MELL quasi-proof-structure.
	If $R$ is \kl{cut-free} and $R \pnrewrite[a] R'$, then $a \neq \Cut^i$ because the action $R \pnrewrite[\Cut^i] R'$ requires the presence of a $\cut$-cell in $R$; 
	moreover, $R'$ is cut-free because no rewrite rule in \Cref{fig:actions-all}---when read left to right---introduces a $\cut$-cell. 

	Conversely, if $R$ is \kl{with cuts} and $R \pnrewrite[a] R'$ where $R'$  is cut-free, then $a = \Cut^i$, because among the rewrite rules in \Cref{fig:actions-all} only $\pnrewrite[\Cut_i]$ erases a $\cut$-cell. 
	
	A straightforward induction on the length of the scheduling $\sched$ allows us to conclude.
\end{proof}

\begin{definition}[\gluability, cut-free \gluability]
	Let $\Pi$ be a set of $\DiLL_0$ quasi-proof-structures of type $\Gamma$. 
	If $\Pi = \emptyset$, then $\Pi$ is \emph{\gluable} and \emph{cut-free \gluable}.
	If $\Pi \neq \emptyset$, $\Pi$ is \intro[glueable]{\gluable} (\resp~\intro[cut-free glueable]{cut-free \gluable}) if
	$ \Pi \pnrewrite[\sched] \{\emptynets\}$ for some \kl{scheduling} (\resp~\kl{cut-free scheduling}) $\sched \colon \Gamma \to \emptylists$.
\end{definition}

\begin{theorem}[\gluability criterion]
  \label{thm:characterization}
	Let \(\Pi\) be a 
	set of $\DiLL_0$ proof-structures without $\maltese$-cells.
	\begin{enumerate}
		\item\label{p:characterization-with-cut} \(\Pi\) is \gluable if and only if $\Pi \subseteq \Taylor{R}$ for some \(\MELL\) proof-structure $R$.
		\item\label{p:characterization-cut-free} \(\Pi\) is cut-free \gluable if and only if $\Pi \subseteq \Taylor{R}$ for some cut-free \(\MELL\) proof-structure~$R$.
	\end{enumerate}
\end{theorem}

\begin{proof} If $\Pi = \emptyset$, Points \ref{p:characterization-with-cut}--\ref{p:characterization-cut-free} trivially hold. 
	So, we can now assume $\Pi \neq \emptyset$.
	\begin{enumerate}
		\item 
		If $\Pi \subseteq \Taylor{R}$ for some $\MELL$ proof-structure $R$,
	then by normalization (Proposition \ref{prop:unwinding-to-empty})
	$R \pnrewrite[\sched] \emptynets$ for some scheduling $\sched \colon (\Gamma) \to \emptylists$, where $\Gamma$ is the type of $R$ (and of each element of $\Pi$, by Remark \ref{rmk:conclusions}).
	Therefore, $\Pi \pnrewrite[\sched] \{\emptynets\}$ by naturality
	(\Cref{thm:projection-natural}), since $\emptyset \neq \Pi \pnrewrite[\sched] \Pi'$ implies $\Pi' \neq \emptyset$ (\Cref{foot:lifting}), and	$\FatTaylor{\emptynets} = \{\emptynets\}$ (Example~\ref{ex:filled-taylor-empty}). 
	
  Conversely, if $\Pi \pnrewrite[\sched] \{\emptynets\}$ for some scheduling $\sched \colon (\Gamma) \to \emptylists$ (where $\Gamma$ is the type of each element of $\Pi$), then by naturality (\Cref{thm:projection-natural}, as 
  $\FatTaylor{\emptynets} = \{\emptynets\}$) $\Pi \subseteq \FatTaylor{R}$ for some $\MELL$ quasi-proof-structure $R \pnrewrite[\sched] \emptynets$,
  and so  $\Pi \subseteq \Taylor{R}$ since each element of $\Pi$ is without $\maltese$-cells.
  The (\MELL) quasi-proof-structure $R$ is indeed a proof-structure as each element of $\Pi$ is a ($\DiLL_0$) proof-structure of type $\Gamma$, and by Remark \ref{rmk:conclusions} the type of $R$ is $\Gamma$.
  
  	\item Just repeat the same proof as above, paying attention that:
  	\begin{itemize}
  		\item in the ``if'' part,  if $R$ is cut-free, then the scheduling $\sched$ 
  		is cut-free by Lemma~\ref{lemma:cut-free-scheduling};
  		\item in the ``only if'' part, if $\Pi \pnrewrite[\sched] \{\emptynets\}$ for some cut-free scheduling $\sched \colon (\Gamma) \to \emptylists$, then by Lemma~\ref{lemma:cut-free-scheduling} $R$ is cut-free, since $R \pnrewrite[\sched] \emptynets$ according to naturality (\Cref{thm:projection-natural}).
  		\qedhere
  	\end{itemize}
	\end{enumerate}
\end{proof}

Theorem \ref{thm:characterization} means that the action of a scheduling rewriting a set $\Pi$ of $\DiLL_0$ proof-structures to the empty proof-structure can be seen, read in reverse, as a sequence of elementary steps to build a \MELL proof structure that includes $\Pi$ in its Taylor expansion.

\begin{example}
An example of the \kl[action of scheduling on dill]{action of a scheduling} starting from  the singleton of a $\DiLL_0$ proof-structure $\rho$ and ending in \(\{\emptynets\}\) is in
\Cref{fig:longex2} (it is by no means the shortest possible rewriting). 
When replayed backward,  by naturality (\Cref{thm:projection-natural}) it induces a \(\MELL\) proof-structure
\(R\) such that \(\rho\in \Taylor{R}\).
In \Cref{fig:longex2} we represent simultaneously the rewriting from $\Pi$ to $\{\emptynets\}$ and the rewriting from $R$ to $\emptynets$.
\end{example}

\begin{figure}[p]
  \centering
  \vspace{\beforepn}
  \begin{align*}
    \begin{array}{c@{\hspace{1\tabcolsep}}c@{\hspace{1\tabcolsep}}c@{\hspace{1\tabcolsep}}c@{\hspace{1\tabcolsep}}c@{\hspace{1\tabcolsep}}l@{\hspace{1\tabcolsep}}c@{\hspace{1\tabcolsep}}cc}
      \big\{\rho\big\} =
      &\tikzsetnextfilename{images/long-path-poly-1}
        \Biggl\{
        \scalebox{\smallproofnets}
        {\pnet{
        \pnformulae{
        ~\pnf[?b]{\small $\wn\wn \bot$}~~~
        \pnf[!A-oA]{\small$\!\!\!\oc\oc(X^{\!\bot} \!\parr\! X) \quad $}~\\
th        }
        \pnwn[wn]{?b}
        \pncown[cown]{!A-oA}
        \pnsemibox{?b,!A-oA,wn,cown}
      }}
      \Biggr\}
      &\pnrewrite[\BoxR_2]&
      \tikzsetnextfilename{images/long-path-poly-2}
      \Biggl\{
      \scalebox{\smallproofnets}                              
      {\pnet{
      \pnformulae{~\pnf[?b]{\small $\wn\wn \bot$}~~~
      	\pnf[!A-oA]{\small$\!\!\oc(X^{\!\bot} \!\parr\! X) \quad $}~
      }
      \pndaimon[d]{?b,!A-oA}
      \pnsemibox{?b,!A-oA,d}
      }}\Biggr\}
    	&\pnrewrite[\Der_1]&
      \tikzsetnextfilename{images/long-path-poly-3}
      \Biggl\{
      \scalebox{\smallproofnets}                              
      {\pnet{
      \pnformulae{~\pnf[?b]{\small $\wn \bot$}~~~
      \pnf[!A-oA]{\small$\!\!\oc (X^{\!\bot} \!\parr\! X)\quad$}~}
      \pndaimon[d]{?b,!A-oA}
      \pnsemibox{?b,!A-oA,d}
      }}\Biggr\}
    	\\[-0.5cm]
      R \ =
      &\tikzsetnextfilename{images/long-path-mell-1}
      \scalebox{\smallproofnets}{\pnet{
      \pnformulae{~~\pnf[a1]{\small $X^{\!\bot}$}~\pnf[a2]{\small $X$}~\\
      \pnf[b1]{\small $\bot$}~\pnf[b2]{\small $\bot$}}
      \pnaxiom[axa]{a1,a2}
      \pnpar{a1,a2}[A-oA]{\small $X^{\!\bot} \!\parr\! X$}
      \pnbag{}{A-oA}[!A-oA]{\small $\oc (X^{\!\bot} \!\parr\! X)$}[1.3]
      \pnbot[bot1]{b1}
      \pnbot[bot2]{b2}
      \pnexp{}{b1,b2}[wbot]{\small $\wn \bot$}[1.5]]
      \pnbox{axa,b1,b2,A-oA,a1,a2}
      \pnbox{axa,b1,b2,!A-oA,wbot,a1,a2}
      \pnexp{}{wbot}[wwbot]{\small $\wn \wn \bot$}[1.3]
      \pnbag{}{!A-oA}[!!A-oA]{\small $\oc \oc (X^{\!\bot} \!\parr\! X)$}[1.2]
      \pnsemibox{wbot,b1,b2,a1,a2,axa,!!A-oA}
      }}
    	&\pnrewrite[\BoxR_2]&
      \tikzsetnextfilename{images/long-path-mell-2}
      \scalebox{\smallproofnets}{\pnet{
      \pnformulae{~~\pnf[a1]{\small $X^{\!\bot}$}~\pnf[a2]{\small $X$}~\\
      \pnf[b1]{\small $\bot$}~\pnf[b2]{\small $\bot$}}
      \pnaxiom[axa]{a1,a2}
      \pnpar{a1,a2}[A-oA]{\small $X^{\!\bot} \!\parr\! X$}
      \pnbot[bot1]{b1}
      \pnbot[bot2]{b2}
      \pnexp{}{b1,b2}[wbot]{\small $\wn \bot$}[1.5]
      \pnbox{axa,b1,b2,A-oA,a2}
      \pnbag{}{A-oA}[!A-oA]{\small $\oc (X^{\!\bot} \!\parr\! X)$}[1.2]
      \pnexp{}{wbot}[wwbot]{\small $\wn \wn \bot$}
      \pnsemibox{wbot,wwbot,b1,b2,a1,a2,axa}
      }}
    	&\pnrewrite[\Der_1]&
      \tikzsetnextfilename{images/long-path-mell-3}
      \scalebox{\smallproofnets}{\pnet{
      \pnformulae{~~\pnf[a1]{\small $X^{\bot}$}~\pnf[a2]{\small $X$}~\\
      	\pnf[b1]{\small $\bot$}~\pnf[b2]{\small $\bot$}}
      \pnaxiom[axa]{a1,a2}
      \pnpar{a1,a2}[A-oA]{\small $X^{\!\bot} \!\parr\! X$}
      \pnbot[bot1]{b1}
      \pnbot[bot2]{b2}
      \pnexp{}{b1,b2}[wbot]{\small $\wn \bot$}[1.5]
      \pnbox{axa,b1,b2,A-oA,a2}
      \pnbag{}{A-oA}[!A-oA]{\small $\oc (X^{\!\bot} \!\parr\! X)$}[1.2]
            \pnsemibox{wbot,b1,b2,a1,a2,axa}
      }}
      \\
      \hline
      \pnrewrite[\BoxR_2]&
      \tikzsetnextfilename{images/long-path-poly-4}
      \Biggl\{\scalebox{\smallproofnets}{\pnet{
      \pnformulae{~\pnf[?b]{\small $\wn \bot$}~~~\pnf[A]{\small $\!X^\bot \!\parr\! X$}~}
      \pndaimon[d]{?b,A}
      \pnsemibox{?b,d,A}
      }}\Biggr\}
    	&\pnrewrite[\parr_2]&
      \tikzsetnextfilename{images/long-path-poly-5}
      \Biggl\{\scalebox{\smallproofnets}{\pnet{
      \pnformulae{~\pnf[?b]{\small $\wn \bot$}~~
      \pnf[Ab]{\small$X^{\bot}$}~\pnf[A]{\small $X$}~}
      \pndaimon[d]{?b,Ab,A}
      \pnsemibox{?b,d,A}
      }}\Biggr\}
      &\pnrewrite[\Mix_1]&                                
      \tikzsetnextfilename{images/long-path-poly-6}
      \Biggl\{\scalebox{\smallproofnets}{\pnet{
      \pnformulae{~\pnf[?b]{\small $\wn \bot$}~~~
      \pnf[Ab]{\small$X^{\bot}$}~\pnf[A]{\small $X$}~}
      \pndaimon[daib]{?b}
      \pndaimon[daiA]{Ab,A}
      \pnsemibox{daib,?b}
      \pnsemibox{daiA,Ab,A}                     
      }}\Biggr\}\\[-0.5cm]       
      \pnrewrite[\BoxR_2]&
      \tikzsetnextfilename{images/long-path-mell-4}
      \scalebox{\smallproofnets}{\pnet{
      \pnformulae{
      \pnf[b1]{\small $\bot$}~\pnf[b2]{\small $\bot$} ~\pnf[a1]{\small
                          $X^{\bot}$}~\pnf[a2]{\small $X$}~}
      \pnaxiom[axa]{a1,a2}
      \pnpar[A-oA]{a1,a2}{\small $X^\bot \!\parr\! X$}
      \pnbot[bot1]{b1}
      \pnbot[bot2]{b2}
      \pnexp{}{b1,b2}[wbot]{\small $\wn \bot$}
      \pnsemibox{bot1,wbot,a1,a2,axa,A-oA}
      }}
    	&\pnrewrite[\parr_2]&
      \tikzsetnextfilename{images/long-path-mell-5}
      \scalebox{\smallproofnets}{\pnet{
      \pnformulae{
      \pnf[b1]{\small $\bot$}~\pnf[b2]{\small $\bot$} ~\pnf[a1]{\small
                          $X^{\bot}$}~\pnf[a2]{\small $X$}~}
      \pnaxiom[axa]{a1,a2}
      \pnbot[bot1]{b1}
      \pnbot[bot2]{b2}
      \pnexp{}{b1,b2}[wbot]{\small $\wn \bot$}     
      \pnsemibox{bot1,wbot,a1,a2,axa}                   
      }}&\pnrewrite[\Mix_1]&
      \tikzsetnextfilename{images/long-path-mell-6}
      \scalebox{\smallproofnets}{\pnet{
      \pnformulae{
      \pnf[b1]{\small $\bot$}~\pnf[b2]{\small $\bot$}~~~\pnf[a1]{\small
                          $X^{\bot}$}~\pnf[a2]{\small $X$}~}
      \pnaxiom[axa]{a1,a2}
      \pnbot[bot1]{b1}
      \pnbot[bot2]{b2}
      \pnexp{}{b1,b2}[wbot]{\small $\wn \bot$}
      \pnsemibox{bot1,bot2,wbot}                       
      \pnsemibox{axa,a1,a2}
      }}\\
      \hline
      \pnrewrite[\Ax_{3}]&
      \tikzsetnextfilename{images/long-path-poly-7}
      \Biggl\{\scalebox{\smallproofnets}{\pnet{
      \pnformulae{~\pnf[?b]{\small $\wn \bot$}~~\pnfsmall[empty]{}}
      \pndaimon[d]{?b}
      \pnsemibox{d,?b}
      \pnsemibox{empty}
      }}\Biggr\}&\pnrewrite[\Contr_1]
      &\tikzsetnextfilename{images/long-path-poly-8}
      \Biggl\{\scalebox{\smallproofnets}{\pnet{
      \pnformulae{~\pnf[b1]{\small $\wn \bot$}~\pnf[b2]{\small $\wn \bot$}~~\pnfsmall[empty]{}}
      \pndaimon[d]{b1,b2}
      \pnsemibox{d,b1,b2}
      \pnsemibox{empty}
      }}\Biggr\}&\pnrewrite[\Der_2]
      &\tikzsetnextfilename{images/long-path-poly-9}
      \Biggl\{\scalebox{\smallproofnets}{\pnet{
      \pnformulae{~\pnf[b1]{\small $\wn \bot$}~\pnf[b2]{\small $\bot$}~~\pnfsmall[empty]{}}
      \pndaimon[d]{b1,b2}
      \pnsemibox{d,b1,b2}
      \pnsemibox{empty}
      }}\Biggr\}\\[-0.5cm]
      \pnrewrite[\Ax_{3}]&
      \tikzsetnextfilename{images/long-path-mell-7}
      \scalebox{\smallproofnets}{\pnet{
      \pnformulae{
      \pnf[b1]{\small $\bot$}~\pnf[b2]{\small $\bot$}~~\pnfsmall[empty]{}}
      \pnbot[bot1]{b1}
      \pnbot[bot2]{b2}
      \pnexp{}{b1,b2}[wbot]{\small $\wn \bot$}
      \pnsemibox{bot1,bot2,b1,b2,wbot}
      \pnsemibox{empty}
      }}&\pnrewrite[\Contr_1]
      &\tikzsetnextfilename{images/long-path-mell-8}
      \scalebox{\smallproofnets}{\pnet{
      \pnformulae{
      \pnf[b1]{\small $\bot$}~\pnf[b2]{\small $\bot$}~~\pnfsmall[empty]{}}
      \pnbot[bot1]{b1}
      \pnbot[bot2]{b2}
      \pnexp{}{b1}[wb1]{\small $\wn \bot$}
      \pnexp{}{b2}[wb2]{\small $\wn \bot$}
      \pnsemibox{bot1,bot2,b1,b2,wb1,wb2}
      \pnsemibox{empty}
      }}&\pnrewrite[\Der_2]
      &\tikzsetnextfilename{images/long-path-mell-9}
      \scalebox{\smallproofnets}{\pnet{
      \pnformulae{
      \pnf[b1]{\small $\bot$}~\pnf[b2]{\small $\bot$}~~\pnfsmall[empty]{}}
      \pnbot[bot1]{b1}
      \pnbot[bot2]{b2}
      \pnexp{}{b1}[wb1]{\small $\wn \bot$}
      \pnsemibox{bot1,bot2,b1,b2,wb1}
      \pnsemibox{empty}
      }}
      \\
      \hline
      \pnrewrite[\Mix_1]
      &\tikzsetnextfilename{images/long-path-poly-10}
      \Biggl\{\scalebox{\smallproofnets}{\pnet{
      \pnformulae{
      	\pnf[b1]{\small $\wn \bot$}~~\pnf[b2]{\small $\bot$}~~\pnfsmall[empty]{}
      }
      \pndaimon[db1]{b1}
      \pndaimon[db2]{b2}
      \pnsemibox{db1,b1}  
      \pnsemibox{db2,b2}
      \pnsemibox{empty}
      }}\Biggr\}
      &\pnrewrite[\bot_2]
      &\tikzsetnextfilename{images/long-path-poly-11}
      \Biggl\{\scalebox{\smallproofnets}{\pnet{
      \pnformulae{
      	\pnf[b1]{\small $\wn \bot$}~~\pnfsmall[empty2]{}~~\pnfsmall[empty]{}
      }
      \pndaimon[d]{b1}
      \pnsemibox{d,b1}
      \pnsemibox{empty2}
      \pnsemibox{empty}
      }}\Biggr\}&\pnrewrite[\Der_1]
      &\tikzsetnextfilename{images/long-path-poly-12}
      \Biggl\{\scalebox{\smallproofnets}{
      	\pnet{
	      \pnformulae{
	      	\pnf[b1]{\small $\bot$}~~\pnfsmall[empty2]{}~~\pnfsmall[empty]{}
	      }
	      \pndaimon[d]{b1}
	      \pnsemibox{d,b1}
	      \pnsemibox{empty2}
	      \pnsemibox{empty}
      }}\Biggr\}
    	\ \pnrewrite[\bot_1] \ 
      \left\{\,3\emptynet \, \right\}
      \\[-0.5cm]
      \pnrewrite[\Mix_1]&\tikzsetnextfilename{images/long-path-mell-10}
      \scalebox{\smallproofnets}{
      	\pnet{
	      \pnformulae{
	      	\pnf[b1]{\small $\bot$}~~\pnf[b2]{\small $\bot$}~~\pnfsmall[empty]{}
	      }
	      \pnbot[bot1]{b1}
	      \pnbot[bot2]{b2}
	      \pnexp{}{b1}[wb1]{\small $\wn \bot$}
	      \pnsemibox{bot1,b1,wb1}
	      \pnsemibox{bot2,b2}
	      \pnsemibox{empty}
      	}
      }
      &\pnrewrite[\bot_2]
      &\tikzsetnextfilename{images/long-path-mell-11}
      \scalebox{\smallproofnets}{
      	\pnet{
		      \pnformulae{
		      \pnf[b1]{\small $\bot$}~~\pnfsmall[empty2]{}~~\pnfsmall[empty]{}}
		      \pnbot[bot1]{b1}
		      \pnexp{}{b1}[wb1]{\small $\wn \bot$}
		      \pnsemibox{b1,bot1}
		      \pnsemibox{empty2}
		      \pnsemibox{empty}
	      }
      }
    	&\pnrewrite[\Der_1]
      &\tikzsetnextfilename{images/long-path-mell-12}
      \scalebox{\smallproofnets}{
      	\pnet{
		      \pnformulae{
		      	\pnf[b1]{\small $\bot$}~~\pnfsmall[empty2]{}~~\pnfsmall[empty]{}
	      	}
		      \pnbot[bot1]{b1}
		      \pnsemibox{bot1,b1}
		      \pnsemibox{empty2}
		      \pnsemibox{empty}
		    }
	    }
    	\quad 
    	\pnrewrite[\bot_1] 
    	\quad 
     	3\emptynet
    \end{array}
  \end{align*}
  
  Scheduling $\BoxR_2 \, \Der_1  \BoxR_2 \parr_2 \, \Mix_1  \Ax_{3} \, \Contr_1 \, \Der_2  \Mix_1  \bot_2 \, \Der_1 \, \bot_1 
  : (\wn \wn \bot, \oc \oc (X^\bot \parr X)) \to 3\emptylist $
  
  \caption{Action of a scheduling witnessing that $\rho \in \Taylor{R}$.   
  }
  \label{fig:longex2}
\end{figure}

We introduced two different \gluability criteria (Theorems \ref{thm:characterization}.\ref{p:characterization-with-cut} and \ref{thm:characterization}.\ref{p:characterization-cut-free}) because, given a non-empty set $\Pi$ of \emph{cut-free} $\DiLL_0$ proof-structures without $\maltese$-cells, the existence of a solution to the inverse Taylor expansion problem depends on whether we look for it among \MELL proof-structures or among \kl{cut-free} \MELL proof-structures (Example~\ref{ex:glueable-non-cut-free}).

\begin{example}
	\label{ex:glueable-non-cut-free}
	Let $\Sigma = \{\sigma_0\}$ as in \Cref{fig:rewrite-with-cut}. 
	This singleton of a cut-free $\DiLL_0$ proof-structure without $\maltese$-cells is \kl[glueable]{\gluable} (as shown by the action of the scheduling in \Cref{fig:rewrite-with-cut}) but not \kl[cut-free glueable]{cut-free \gluable} (as $\{\sigma_0\} \pnrewrite[\BoxR_1] \{\sigma\}$ and, apart from $\Cut^i$, no \kl{elementary schedulings} apply to $\{\sigma\}$ because of its type).
	Indeed, $\sigma_0$ ``masks'' any information about the content $B$ of the box represented by its $\oc$-cell, apart from its type $\oc X$ which does not allow $B$ to be a cut-free $\MELL$ proof-structure (see also \Cref{sec:inhabitation}).
	\Gluability (with cuts) of $\{\sigma\}$ is related to the fact that there is a \MELL proof-structure of type $X$ \kl{with cuts}, for every \kl{atomic formula} $X$ (see Remark \ref{rmk:non-cut-free}).
	Note that $\sigma_0 \in \Taylor{S_0}$, and $\sigma \in \FatTaylor{S}$ and $\sigma' \in \FatTaylor{S'}$, where $S_0$, $S$ and $S'$ are the \MELL proof-structures with cuts represented in \Cref{fig:rewrite-with-cut}.
\end{example}

\begin{figure}[!t]
	\vspace{\beforepn}
	\centering
		\scalebox{0.8}{
		\begin{tabular}{c@{}c@{}c@{}c@{}c@{}c@{}c@{}c@{}c@{}c@{}c@{}c@{}c@{}c}
			$\Sigma = \Bigg\{\!\!$\!\!
			\scalebox{0.8}{
			\vspace{\beforepn}
				\pnet{
					\pnformulae{\pnf[1']{$\oc X$}}
					\pncown[1]{1'}
			}}
			$\!\!\!\Bigg\}$
			&
			$\pnrewrite[\!\BoxR_1]$
			&
			$\Bigg\{\!\!$
			\!\!\scalebox{0.8}{
						\vspace{\beforepn}
				\pnet{
					\pnformulae{\pnf[1']{$X$}}
					\pndaimon[1]{1'}
			}}
			$\!\!\!\Bigg\}$
			&
			$\pnrewrite[\!\Cut^2\!]$
			&
			$\Bigg\{\!\!$
			\!\!\!\scalebox{0.8}{
							\vspace{\beforepn}
				\pnet{
					\pnformulae{\pnf[1']{$X$}~\pnf[Xp]{$\wn X^\bot$}~\pnf[X]{$\oc X$}}
					\pndaimon[1]{1',Xp,X}
			}}\!
			$\!\!\!\Bigg\}$
			&
			$\pnrewrite[\!\Contr_2]$
			&
			$\Bigg\{\!\!$
			\!\!\!\scalebox{0.8}{
				\vspace{\beforepn}
				\pnet{
					\pnformulae{\pnf[1']{$X$}~\pnf[Xp]{$\wn X^{\!\bot}$}~\pnf[Xp2]{$\wn X^{\!\bot}$}~\pnf[X]{$\oc X$}}
					\pndaimon[1]{1',Xp,Xp2,X}
			}}\!
			$\!\!\!\Bigg\}$
			&
			$\pnrewrite[\!\Mix_2]$
			&
			$\Bigg\{\!\!$
			\!\!\!\scalebox{0.8}{
							\vspace{\beforepn}
				\pnet{
					\pnformulae{\pnf[1']{$X$}~\pnf[Xp]{$\wn X^{\!\bot}\!\!$}~~\pnf[Xp2]{$\wn X^{\!\bot}$}~\pnf[X]{$\oc X\,$}}
					\pndaimon[1]{1',Xp}
					\pndaimon[2]{Xp2,X}
					\pnsemibox{1,Xp}
					\pnsemibox{2,Xp2}
			}}\!
			$\!\!\Bigg\}$
			&
			$\pnrewrite[\!\BoxR_4]$
			&
			$\Bigg\{\!\!$
			\!\!\!\scalebox{0.8}{
				\vspace{\beforepn}
				\pnet{
					\pnformulae{\pnf[1']{$X$}~\pnf[Xp]{$\wn X^{\!\bot}\!\!$}~~\pnf[Xp2]{$\wn X^{\!\bot}$}~\pnf[X]{$X\,$}}
					\pndaimon[1]{1',Xp}
					\pndaimon[2]{Xp2,X}
					\pnsemibox{1,Xp}
					\pnsemibox{2,Xp2}
			}}\!
			$\!\!\Bigg\}$
			&
			$\pnrewrite[\!\Der_2]\,\pnrewrite[\!\Der_3] \, \pnrewrite[\!\Ax_4] \, \pnrewrite[\!\Ax_2]$
			&
			$\big\{\emptynets\big\}$
			\\
			\qquad $\sigma_0$ & & $\sigma$ & & $\sigma'$
		\end{tabular}		
	}

	\vspace{-\baselineskip}
		$S_0 = $
		\scalebox{0.8}{
			\pnet{
				\pnformulae{
					\pnf[X]{$X$}~\pnf[Xb]{$X^{\bot}$}~\pnf[Ab]{$X^{\bot}$}~\pnf[A]{$X$}
				}
				\pnaxiom[ax]{X,Xb}
				\pnaxiom[axA]{A,Ab}
				\pnbox{axA,Ab,A}
				\pnexp{}{Xb,Ab}[a]{$\wn X^\bot$}[1.2]
				\pnbag{}{A}[b]{$\oc X$}[1.2]
				\pncut[cut]{b,a}
				\pnbox{ax,Ab,A,X,Xb,axA,b,a,cut}
				\pnbag{}{X}{$\oc X$}[3]
			}
		}
		\qquad
		$S = $
		\scalebox{0.8}{
			\pnet{
				\pnformulae{
					\pnf[X]{$X$}~\pnf[Xb]{$X^{\bot}$}~\pnf[Ab]{$X^{\bot}$}~\pnf[A]{$X$}
				}
				\pnaxiom{X,Xb}
				\pnaxiom[ax]{A,Ab}
				\pnbox{ax,Ab,A}
				\pnexp{}{Xb,Ab}[a]{$\wn X^\bot$}[1.2]
				\pnbag{}{A}[b]{$\oc X$}[1.2]
				\pncut{b,a}
			}
		}
		\qquad
		$S' = $
		\scalebox{0.8}{
			\pnet{
				\pnformulae{
					\pnf[X]{$X$}~\pnf[Xb]{$X^{\bot}$}~\pnf[Ab]{$X^{\bot}$}~\pnf[A]{$X$}
				}
				\pnaxiom{X,Xb}
				\pnaxiom[ax]{A,Ab}
				\pnbox{ax,Ab,A}
				\pnexp{}{Xb,Ab}[a]{$\wn X^\bot$}[1.2]
				\pnbag{}{A}[b]{$\oc X$}[1.2]
			}
		}

	\vspace{-\baselineskip}
	\caption{Example of \gluability that is not cut-free \gluability.}
	\label{fig:rewrite-with-cut}
\end{figure}

\begin{example}
  \label{ex:not-coherent}
  The three $\DiLL_0$ proof-structures $\rho_1, \rho_2, \rho_3$ below 
  are not \gluable as a whole, but are \gluable two by two (this is a slight variant of the example in \cite[pp.~244-246]{Tasson:2009}). 
  In fact, there is no \MELL proof-structure whose Taylor expansion contains $\rho_1, \rho_2, \rho_3$, but any pair of them is in the Taylor expansion of some \MELL proof-structure. 
  
  \begin{center}
    $\rho_1$ 
    \qquad\qquad\qquad\qquad\qquad\qquad
		$\rho_2$ 
    \qquad\qquad\qquad\qquad\qquad\qquad
		$\rho_3$ 
	
		\vspace*{-1.5\baselineskip}
    \scalebox{0.65}{
      \pnet{
        \pnformulae{
          \pnf[1]{$\mathbf{1}$}~\pnf[2]{$\mathbf{1}$}~\pnf[3]{$\mathbf{1}$}~
          \pnf[4]{$\mathbf{1}$}~
          \pnf[5]{$\bot$}~\pnf[6]{$\bot$}~\pnf[7]{$\bot$}~
          \pnf[8]{$\bot$}~\pnf[9]{$\bot$}
        }
        \pnone{1}
        \pnone{2}
        \pnone{3}
        \pnone{4}
        \pnbot{5}
        \pnbot{6}
        \pnbot{7}
        \pnbot{8}
        \pnbot{9}
        \pnbag{}{1,2}{$\oc \mathbf{1}$}
        \pnbag{}{3}{$\oc \mathbf{1}$}
        \pnbag{}{4}{$\oc \mathbf{1}$}
        \pnexp{}{5}{$\wn \bot$}
        \pnexp{}{6,7}{$\wn \bot$}
        \pnexp{}{8,9}{$\wn \bot$}
      }}
    \qquad 
    \tikzsetnextfilename{images/tasson-example-2}
    \scalebox{0.65}{
      \pnet{
        \pnformulae{
          \pnf[1]{$\mathbf{1}$}~\pnf[2]{$\mathbf{1}$}~\pnf[3]{$\mathbf{1}$}~
          \pnf[4]{$\mathbf{1}$}~
          \pnf[5]{$\bot$}~\pnf[6]{$\bot$}~\pnf[7]{$\bot$}~
          \pnf[8]{$\bot$}~\pnf[9]{$\bot$}
        }
        \pnone{1}
        \pnone{2}
        \pnone{3}
        \pnone{4}
        \pnbot{5}
        \pnbot{6}
        \pnbot{7}
        \pnbot{8}
        \pnbot{9}
        \pnbag{}{1}{$\oc \mathbf{1}$}
        \pnbag{}{2,3}{$\oc \mathbf{1}$}
        \pnbag{}{4}{$\oc \mathbf{1}$}
        \pnexp{}{5,6}{$\wn \bot$}
        \pnexp{}{7}{$\wn \bot$}
        \pnexp{}{8,9}{$\wn \bot$}
      }}
    \qquad
    \tikzsetnextfilename{images/tasson-example-3}
    \scalebox{0.65}{
      \pnet{
        \pnformulae{
          \pnf[1]{$\mathbf{1}$}~\pnf[2]{$\mathbf{1}$}~\pnf[3]{$\mathbf{1}$}~
          \pnf[4]{$\mathbf{1}$}~
          \pnf[5]{$\bot$}~\pnf[6]{$\bot$}~\pnf[7]{$\bot$}~
          \pnf[8]{$\bot$}~\pnf[9]{$\bot$}
        }
        \pnone{1}
        \pnone{2}
        \pnone{3}
        \pnone{4}
        \pnbot{5}
        \pnbot{6}
        \pnbot{7}
        \pnbot{8}
        \pnbot{9}
        \pnbag{}{1}{$\oc \mathbf{1}$}
        \pnbag{}{2}{$\oc \mathbf{1}$}
        \pnbag{}{3,4}{$\oc \mathbf{1}$}
        \pnexp{}{5,6}{$\wn \bot$}
        \pnexp{}{7,8}{$\wn \bot$}
        \pnexp{}{9}{$\wn \bot$}
      }
    }
  \end{center}
  Indeed, for any $i,j \in \{1,2,3\}$ with $i \neq j$, \(\rho_i\) and \(\rho_j\) are in the Taylor expansion of \(R_{ij}\) below.
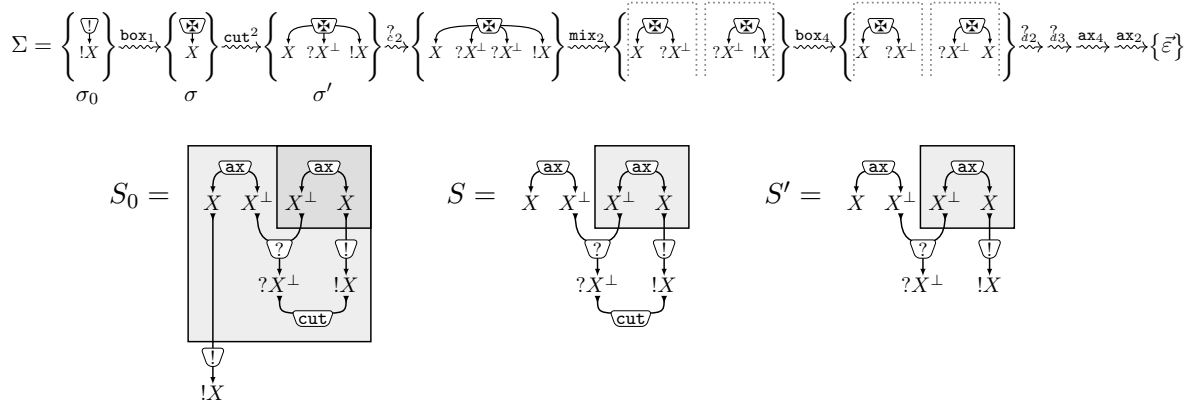
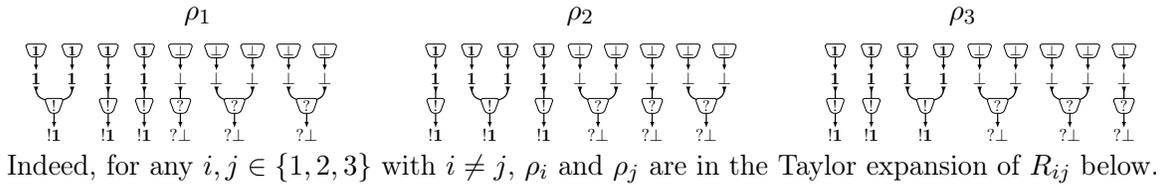
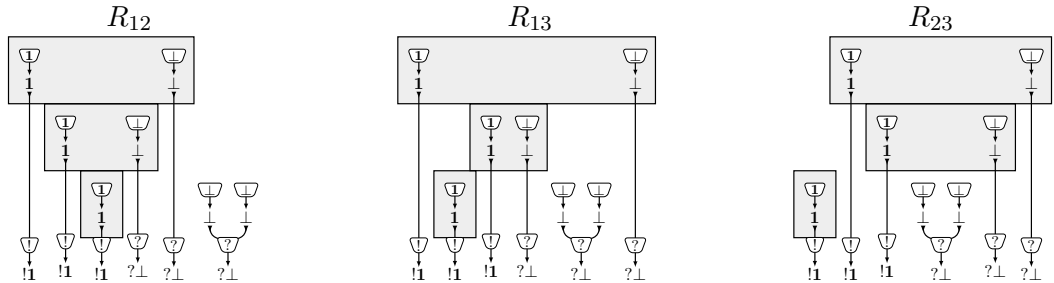
  \begin{center}
  	$R_{12}$ 
  	\qquad\qquad\qquad\qquad\qquad\qquad
  	$R_{13}$ 
  	\qquad\qquad\qquad\qquad\qquad\qquad
  	$R_{23}$ 
  	
  	\vspace*{-1.5\baselineskip}
    \scalebox{0.65}{
      \pnet{
        \pnformulae{
          \pnf[1]{$\mathbf{1}$}~~~~\pnf[5]{$\bot$}\\\\
          ~\pnf[2]{$\mathbf{1}$}~~\pnf[4]{$\bot$}\\\\
          ~~\pnf[3]{$\mathbf{1}$}~
          ~~\pnf[6]{$\bot$}~
          \pnf[7]{$\bot$}
        }
        \pnone[o1]{1}
        \pnone[o2]{2}
        \pnone[o3]{3}
        \pnbot{4}
        \pnbot{5}
        \pnbot{6}
        \pnbot{7}
        \pnbag{}{1}{$\oc \mathbf{1}$}[3.75]
        \pnbag{}{2}{$\oc \mathbf{1}$}[2.3]
        \pnbag{}{3}{$\oc \mathbf{1}$}
        \pnexp{}{4}{$\wn \bot$}[2.3]
        \pnexp{}{5}{$\wn \bot$}[3.75]
        \pnexp{}{6,7}{$\wn \bot$}
        \pnbox{o1,1,5}
        \pnbox{o2,2,4}
        \pnbox{o3,3}
      }
    }
  	\qquad\qquad
  	\scalebox{0.65}{
  		\pnet{
  			\pnformulae{
  				\pnf[1]{$\mathbf{1}$}~~~~~~\pnf[7]{$\bot$}\\\\
  				~~\pnf[3]{$\mathbf{1}$}~\pnf[4]{$\bot$}~~\\\\
  				~\pnf[2]{$\mathbf{1}$}~~
  				~\pnf[5]{$\bot$}~\pnf[6]{$\bot$}
  			}
  			\pnone[o1]{1}
  			\pnone[o2]{2}
  			\pnone[o3]{3}
  			\pnbot{4}
  			\pnbot{5}
  			\pnbot{6}
  			\pnbot{7}
  			\pnbag{}{1}{$\oc \mathbf{1}$}[3.75]
  			\pnbag{}{2}{$\oc \mathbf{1}$}
  			\pnbag{}{3}{$\oc \mathbf{1}$}[2.3]
  			\pnexp{}{4}{$\wn \bot$}[2.3]
  			\pnexp{}{5,6}{$\wn \bot$}
  			\pnexp{}{7}{$\wn \bot$}[3.75]
  			\pnbox{o1,1,7}
  			\pnbox{o2,2}
  			\pnbox{o3,3,4}
  		}
  	}
  	\qquad\qquad
  	\scalebox{0.65}{
  		\pnet{
  			\pnformulae{
  				~\pnf[2]{$\mathbf{1}$}~~~~~\pnf[7]{$\bot$}\\\\
  				~~\pnf[3]{$\mathbf{1}$}~~~\pnf[6]{$\bot$}\\\\
  				\pnf[1]{$\mathbf{1}$}~~~\pnf[4]{$\bot$}
  				~\pnf[5]{$\bot$}~
  			}
  			\pnone[o1]{1}
  			\pnone[o2]{2}
  			\pnone[o3]{3}
  			\pnbot{4}
  			\pnbot{5}
  			\pnbot{6}
  			\pnbot{7}
  			\pnbag{}{1}{$\oc \mathbf{1}$}
  			\pnbag{}{2}{$\oc \mathbf{1}$}[3.75]
  			\pnbag{}{3}{$\oc \mathbf{1}$}[2.3]
  			\pnexp{}{4,5}{$\wn \bot$}
  			\pnexp{}{6}{$\wn \bot$}[2.3]
  			\pnexp{}{7}{$\wn \bot$}[3.75]
  			\pnbox{o1,1}
  			\pnbox{o2,2,7}
  			\pnbox{o3,3,6}
  		}
  	}
  \end{center}
\end{example}

Example \ref{ex:not-coherent} can be generalized to finite sets of any cardinality: for any \(n > 1\), there is a finite set $\Pi$ of \(n+1\) $\DiLL_0$ proof-structures that is not \gluable as a whole, but all subsets of $\Pi$ of at most \(n\) elements are \gluable. 
This means that, for every $n > 1$, \gluability of a set $\Pi$ of $\DiLL_0$ proof-structures is not reducible to test that all subsets of $\Pi$ with at most $n$ elements satisfy some $n$-ary coherence relation, unless (obviously) $\Pi$ is a finite set with at most $n$ elements.
As explained in \Cref{sect:intro}, this difficulty in the inverse Taylor expansion problem is typical of $\MELL$ and does not arise in the $\lambda$-calculus.

\begin{remark}[infinite sets]
  \label{rmk:infinite}
The \gluability criterion (\Cref{thm:characterization}) is not limited to finite sets of \(\DiLL_0\) proof-structures, it holds for infinite sets too. 
This raises the question of the relationship between \gluability and finite \gluability, namely: when every finite subset of a given infinite set of $\DiLL_0$ proof-structures is \gluable, is the infinite set itself \gluable too?
The answer is negative.
  Let \(\Pi\) be an \emph{infinite} set of \(\DiLL_0\) proof-structures such that \emph{every finite} subset of \(\Pi\) is \gluable. 
  There are two cases for $\Pi$.
  \begin{itemize}

  \item \(\Pi\) itself is \gluable. It is typically the case when \(\Pi\) is exactly the
    Taylor expansion of a \(\MELL\) proof-structure $R$ with at least one box, such as in the following example.
    \begin{center}
    	\vspace{\beforepn}
    	$\Pi = \Bigg\{
			\scalebox{0.8}{
			\pnet{
				\pnformulae{\pnf[1']{$\one$}~\pnf{$\overset{n}{\dots}$}~\pnf[1'']{$\one$}}
				\pnone[1]{1'}
				\pnone[1bis]{1''}
				\pnbag{}{1',1''}{$\oc \one$}
			}}
    	\mathrel{\bigg|} n \in \Nat \Bigg\}
    	\  \pnrewrite[\BoxR_1] \ 
    	\bigg\{
    		\scalebox{0.8}{
	    	\pnet{
	    		\pnformulae{\pnf[1']{$\one$}}
	    		\pnone[1]{1'}
	    	}}
    	\bigg\}
    	\ \pnrewrite[\One_1] \ 
    	\{\emptynet\}
    	$
    	\qquad\qquad\quad
    	$R =
    	\scalebox{0.8}{
    	\pnet{
    		\pnformulae{\pnf[1']{$\one$}}
    		\pnone[1]{1'}
    		\pnbox{1,1'}
    		\pnbag{}{1'}{$\oc \one$}[1.2]
    	}}
    	$
    \end{center}
   
    We might say that
    \(\Pi\) is \emph{infinite in width}: infinity is related to the number of
    copies chosen for the boxes of $R$, a number that has to be finite but can
    be arbitrarily large.
    
  \item \(\Pi\) is not \gluable. 
  This is the case in the following set (we draw three elements only).

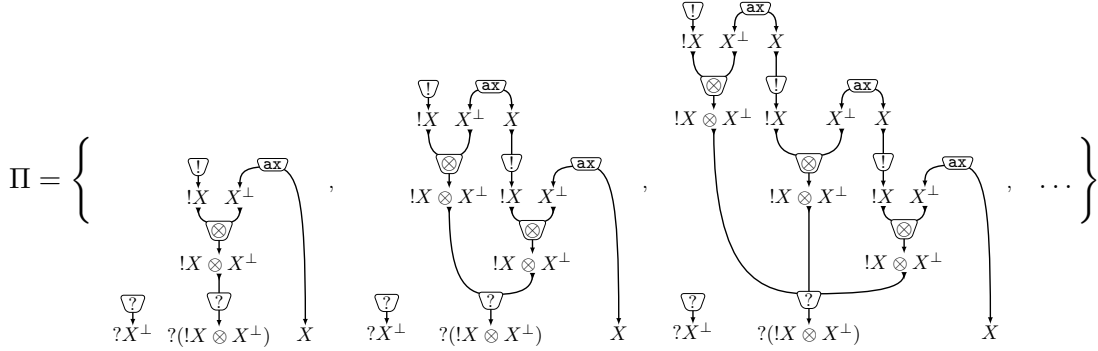
\begin{center}
  \vspace{\beforepn}
  $\Pi = \Bigg\{$ \!\!
  \scalebox{\smallproofnets}{
    \raisebox{-1.4cm}{
      \pnet{
        \pnformulae{~~\pnf[ba]{$\oc X$}~\pnf[na]{$X^{\bot}$}\\\\\\
          \\
          \pnf[na']{$\wn X^{\bot}$}~~~~~\pnf[a]{$X$}
      	}
      	\pnwn{na'}
        \pncown{ba}
        \pnaxiom{na,a}
        \pntensor{ba,na}[bana]{$\oc X \otimes X^{\bot}$}
        \pnexp{}{bana}{$\wn (\oc X \otimes X^{\bot})$}[1.1]
      }
    },
    \quad
    \raisebox{-0.7cm}{
      \pnet{
        \pnformulae{~\pnf[ba]{$\oc X$}~\pnf[na]{$X^{\bot}$}~\pnf[a]{$X$}\\\\
          ~~~~\pnf[na1]{$X^{\bot}$}\\\\\\\\
          \pnf[na']{$\wn X^{\bot}$}~~~~~~\pnf[a1]{$X$}
        }
        \pnwn{na'}
        \pncown{ba}
        \pnaxiom{na,a}
        \pntensor{ba,na}[bana]{$\oc X \otimes X^{\bot}$}[1.2]
        \pnbag{}{a}[bal]{$\oc X$}[1.2]
        \pnaxiom{na1,a1}
        \pntensor{bal,na1}[balna1]{$\oc X \otimes X^{\bot}$}
        \pnexp{}{bana,balna1}{$\wn (\oc X \otimes X^{\bot})$}[1.1]
      }
    },
    \quad
    \pnet{
      \pnformulae{\pnf[ba2]{$\oc X$}~\pnf[na2]{$X^{\bot}$}~\pnf[a2]{$X$}\\\\
        ~~~~\pnf[na]{$X^{\bot}$}~\pnf[a]{$X$}\\\\
        ~~~~~~\pnf[na1]{$X^{\bot}$}\\\\\\\\
      \pnf[na']{$\wn X^{\bot}$}~~~~~~~~\pnf[a1]{$X$}}
      \pnwn{na'}
      \pnbag{}{a2}[ba]{$\oc X$}[1.2]
      \pncown{ba2}
      \pnaxiom{na2,a2}
      \pntensor{ba2,na2}[bana2]{$\oc X \otimes X^{\bot}$}[1.2]
      \pnaxiom{na,a}
      \pntensor{ba,na}[bana]{$\oc X \otimes X^{\bot}$}[1.2]
      \pnbag{}{a}[bal]{$\oc X$}[1.2]
      \pnaxiom{na1,a1}
      \pntensor{bal,na1}[balna1]{$\oc X \otimes X^{\bot}$}
      \pnexp{}{bana,balna1,bana2}{$\wn (\oc X \otimes X^\bot)$}[1.1]
      }, 
    }
  	\ \dots $\Bigg\}$
    \end{center}
  Every finite subset of $\Pi$ is \gluable, but the whole infinite set $\Pi$ is not.
  We might say that \(\Pi\) is \emph{infinite in depth}: in $\Pi$, for every $n > 0$ there is a $\DiLL_0$ proof-structure with $n$ $\oc$-cells.
  Any finite subset $\Pi_n$ of $\Pi$ ($n$ is the maximal number of $\oc$-cells in the elements of $\Pi_n$) is included in the Taylor expansion of a $\MELL$ proof-structure with \kl[depth of proof-structure]{depth} $n$, and indeed the action of the scheduling $\nu_n^1$ defined in Example \ref{ex:church} rewrites $\Pi_n$ to $\{\emptynets\}$.\footnotemark
  \footnotetext{\label{note:church}In fact, if $\pi_h$ is the element of such a $\Pi$ with exactly $h$ $\oc$-cells, $\{\pi_h\} \pnrewrite[\nu_n^1] \{\emptynets\}$ for all $n \geqslant h$.} 
  However, the whole $\Pi$ can only been seen as the Taylor expansion of an ``infinite'' proof-structure with infinitely deep nested boxes. 
  Such infinite proof-structures are not in the syntax of $\MELL$. 
  An ongoing research  relates them to non-well-founded proofs for fixed-point~logics~\cite{Infinets1,Infinets2}.
  \end{itemize}
\end{remark}

In the inverse Taylor expansion problem, both kinds of infinity---in width and in depth---arise already in the $\lambda$-calculus, they are not specific to $\MELL$. 
	An instance of a set of resource $\lambda$-terms that is infinite in width is the Taylor expansion of the $\lambda$-term $xy$.
	The set $\Pi = \{\lambda f. \lambda x. f [\,],\allowbreak \
	     \lambda f. \lambda x. f [f [\,] ],\allowbreak \
	     \lambda f. \lambda x. f [f [f [\,] ]], \
	     \dots\}$ 
	of resource $\lambda$-terms is an example of infinity in depth: every finite subset of $\Pi$ is included in the Taylor expansion of some $\lambda$-term (a Church numeral), but $\Pi$ is not included in the Taylor expansion of any $\lambda$-term.

\subsection{Inhabitation}
\label{sec:inhabitation}
We can characterize type inhabitation in $\MELL$ proof-structures. 
We first put the question of type inhabitation for $\MELL$ proof-structures in a proper way.

\begin{remark}[inhabitation with cuts is trivial]
  \label{rmk:non-cut-free}
  For every \kl{atomic} formula $X$ there is a \MELL proof-structure of type $X$ \kl{with cuts}, the one called $S$ in \Cref{fig:rewrite-with-cut}.
  More generally, for every $\MELL$ formula $A$ (\resp~list $\Gamma$ of \MELL formulas), there is a \MELL proof-structure of type $A$ (\resp~$\Gamma$) possibly with cuts. 
  Indeed, for $A$ just ``$\eta$-expand'' $S$ in \Cref{fig:rewrite-with-cut}, by replacing each $\ax$-cell in $S$ with the ``canonical'' \MELL proof-structure of type $A, A^\bot$. 
   For $\Gamma = \allowbreak(A_1, \dots, A_n)$, repeat the same procedure for each $A_i$ separately, and juxtapose the \MELL proof-structures to yield a unique \MELL proof-structure $R$ of type $\Gamma$ (take $R = \emptynet$ if $\Gamma = \emptylist$).
  
  Therefore, type inhabitation for \MELL proof-structures is trivial, if $\cut$-cells are allowed. 
\end{remark} 

Type inhabitation for \MELL proof-structures is a non-trivial problem only if we restrict to \emph{\kl{cut-free}} \MELL proof-structures. 
Indeed, given two propositional variables $X \neq Y$, there is no \emph{cut-free} \MELL proof-structure of type $X$ or $Y, X$.
So, the problem is the following: 
\begin{description}
	\item[Data] a list $\Gamma$ of \MELL formulas;
	\item[Question] does there exist a \kl{cut-free} \MELL proof-structure of type $\Gamma$?
\end{description}
We study the \emph{type inhabitation problem for cut-free \MELL proof-structures} via schedulings, following the idea that a \kl{scheduling} $\sched \colon (\Gamma) \to \emptylists$ ``encodes'' a \MELL proof-structure of type $\Gamma$.

Let us see the behavior of \kl[cut-free schedulings]{\emph{cut-free} schedulings}. 
For some list $\Gamma$ of \MELL formulas there is no cut-free scheduling $\sched \colon (\Gamma) \to \emptylists$.
For instance, if $X \neq Y$ are propositional variables and we exclude $\Cut^i$, then there are no \kl{elementary schedulings} from $(X)$, while the only other elementary schedulings from $(Y,X)$ are $\Mix_1$ and $\Ex_1$.
The non-existence of a cut-free-scheduling from $(X)$ or $(Y,X)$ to $\emptylists$ goes hand in hand with the non-existence of a cut-free $\MELL$ proof-structure of type $X$ or $Y,X$, by Proposition \ref{prop:unwinding-to-empty} and Lemma~\ref{lemma:cut-free-scheduling}.

Conversely, for any list $\Gamma$ of \MELL formulas, a \MELL proof-structure $R(\sched)$ of type $\Gamma$ is associated with every scheduling $\sched \colon (\Gamma) \to \emptylists$: such a $R(\sched)$ is built by reading $\sched$ backward, it exists and is uniquely determined because the action of $\sched$ is \emph{co-functional} (Lemma \ref{prop:unwinding-co-functional}).

\begin{definition}[building from scheduling]
	\label{def:building}
	Let $\Gamma \neq \emptylist$ be a list of lists of \MELL formulas, and $\sched \colon \Gamma \to n\emptylist$ be a  scheduling with $n \!>\! 0$.
	The \MELL quasi-proof-structure $R(\sched)$ is defined by induction on the \kl[scheduling length]{length} of $\sched$: if $\sched$ is \kl[empty scheduling]{empty}, $R(\sched) = n \emptynet$;
		otherwise $\sched = a \schedtwo$ for some elementary scheduling $a$, and $R(\sched)$ is the only \MELL quasi-proof-structure $R$ such that $R \pnrewrite[a] R(\schedtwo)$.
\end{definition}

Clearly, for any scheduling $\sched \colon \Gamma \to \emptylists$, $R(\sched)$ is of type $\Gamma$ and $R(\sched) \pnrewrite[\sched] \emptynets$, and $R(\sched)$ is cut-free if and only if so is $\sched$ (Lemma \ref{lemma:cut-free-scheduling}).
Actually, cut-free schedulings are a way to characterize the (lists of) \MELL formulas that are the type of some \kl{cut-free} \MELL proof-structure.

\begin{thm}[type inhabitation for cut-free \MELL proof-structures]
	\label{thm:inhabitation}
	Let $\Gamma$ be a list of \MELL formulas.
	There is a \kl{cut-free} \MELL proof-structure of type $\Gamma$ if and only if there is a  \kl{cut-free scheduling} $\sched \colon (\Gamma) \to \emptylists$.
\end{thm}

\begin{proof}
	If $\sched \colon (\Gamma) \to \emptylists$ is a cut-free scheduling, $R = R(\sched)$ is a \MELL quasi-proof-structure of type $\Gamma$. As $\Gamma$ is a list of \MELL formulas, $R$ is a proof-structure.
	By Lemma~\ref{lemma:cut-free-scheduling}, $R$ is cut-free.
	
	Conversely, if $R$ is a cut-free \MELL proof-structure of type $\Gamma$, then 
	$R \pnrewrite[\sched] \emptynets$ for some cut-free scheduling $\sched \colon (\Gamma) \to \emptylists$, by Proposition \ref{prop:unwinding-to-empty} and Lemma~\ref{lemma:cut-free-scheduling}.
\end{proof}

	Lemma \ref{prop:unwinding-co-functional}, Definition \ref{def:building} and the proof of Theorem \ref{thm:inhabitation} say that, given a list $\Gamma$ of \MELL formulas, any cut-free scheduling from  $(\Gamma)$ to $\emptylists$ can be seen, when read backward, as a sequence of elementary steps to build a cut-free \MELL proof-structure of type $\Gamma$~from~scratch.
	
	Consider now the \emph{procedure} $\mathcal{P}_1$ below, given a list $\Gamma$ of \MELL formulas as input:
	\begin{enumerate}
		\item for all $n \geqslant 0$, take all cut-free schedulings from $(\Gamma)$ of \kl[scheduling length]{length} $n$ (there are finitely many);
		\item as soon as a list of the form $\emptylists$ is reached by a cut-free scheduling, accept $\Gamma$.
	\end{enumerate}

\Cref{thm:inhabitation} guarantees that such a \emph{nondeterministic} procedure accepts $\Gamma$ if and only if there is a cut-free \MELL proof-structure of type $\Gamma$, and any cut-free scheduling $\sched$ that allows the procedure to accept $\Gamma$ also constructs a witness $R(\sched)$.
Therefore, we can conclude:

\begin{fact}[semi-decidability]
	Procedure $\mathcal{P}_1$ semi-decides the type inhabitation problem for cut-free \MELL proof-structures.
\end{fact}

Procedure $\mathcal{P}_1$ need not halt. 
Indeed, in general, there is no bound on the length of cut-free schedulings, and the elementary scheduling $\Contr_i$ can create cut-free schedulings of arbitrary length with target $\emptylists$, such as $\schedtwo_n$ below---of length $3n+1$---for every $n \geqslant 0$.
$$
\schedtwo_n = \underbrace{\Contr_1 \cdots \Contr_1}_{n \text{ times}} \underbrace{\Mix_{1} \Weak_1 \cdots \Mix_{1} \Weak_1}_{n \text{ times}} \Weak_1 \ \colon \, (\wn X) \to (n+1)\emptylist
$$

Consider now the \emph{procedure} $\mathcal{P}_2$ below, a slight variant of $\mathcal{P}_1$ for the same input $\Gamma$:  

\begin{enumerate}
	\item for all $n \geqslant 0$, take all cut-free schedulings from $(\Gamma)$ of \kl[scheduling length]{length} $n$ that do not contain the elementary scheduling $\Contr_{i}$ (there are finitely many);
	\item as soon as a list of the form $\emptylists$ is reached by one of those schedulings, accept $\Gamma$.
\end{enumerate}

Procedure $\mathcal{P}_2$ always halts. 
Indeed, define the following \emph{sizes}: 
\begin{itemize}
	\item the size $\Size{A}$ of a \MELL formula $A$ is the number of \kl[atomic formula]{atoms}, \kl{connectives} and \kl{units} occurring in $A$, \ie $\Size{X} = \Size{\One} = 1$, $\Size{A \otimes B} =  \Size{A} + \Size{B} + 1$, $\Size{\oc A} = \Size{A} + 1$ and $\Size{A^\perp} = \Size{A}$;

	\item the size of a list $(A_1, \dots, A_n)$ of \MELL formulas is the sum $\sum_{i=i}^n\Size{A_i}$;
	\item the size of a (non-empty) list $\Gamma$ of lists of \MELL formulas is $\Size{\Gamma} = (|\Gamma|, \mathsf{n}(\Gamma))$ with the lexicographic order, where $|\Gamma|$ is the size of the \kl{flattening} of $\Gamma$, and $\mathsf{n}(\Gamma)$ is the finite multiset of the lengths of each list in $\Gamma$, with the multiset order.
\end{itemize}

One has $\Size{\Gamma} > \Size{\Gamma'}$ for any \kl{elementary scheduling} $a \colon \Gamma \to \Gamma'$ except
\begin{itemize}
	\item when $a = \Cut^{i}$ since $\Size{\Gamma} < \Size{\Gamma'}$, but $\mathcal{P}_2$  considers cut-free schedulings only;
	\item when $a = \Ex_{i}$ since $\Size{\Gamma} = \Size{\Gamma'}$, but infinite chains of consecutive $\Ex_{i}$ are irrelevant, because for every $\Gamma$, sufficiently long schedulings $\Ex_{i_1} \dots \Ex_{i_n}$ are from $\Gamma$ to $\Gamma$;\footnotemark
	\footnotetext{\label{note:exc-schedule}Indeed, for every list $\Gamma \neq \emptylist$ of lists of \MELL formulas, there is $N_\Gamma \geqslant 0$ such that, for every scheduling $\sched \colon \Gamma \to \Delta$ only made of $\Ex_{i}$, there is a scheduling $\schedtwo \colon \Gamma \to \Delta$ only made of $\Ex_{i}$ and of length $\leqslant N_\Gamma$.}
	
	\item when $a = \Contr_{i}$ since $\Size{\Gamma} < \Size{\Gamma'}$, but $\mathcal{P}_2$ only considers schedulings without $\Contr_{i}$.
\end{itemize}

The nondeterministic procedure $\mathcal{P}_2$ decides the type inhabitation problem in a subset of \MELL proof-structures.
We say that a $\wn$-cell with a least two inputs is a \emph{strict contraction}.
Note that if $a = \Contr_{i}$ and $R \pnrewrite[a] R'$ is not \kl[nullary splitting for elementary scheduling]{nullary $\Contr$-splitting}, then $R$ has a strict contraction.

\begin{proposition}[decidability without strict contractions]
	\label{fact:decide-without-contractions}
	Procedure $\mathcal{P}_2$ decides the type inhabitation problem for cut-free $\MELL$ proof-structures \emph{without strict contractions}, that is, 
	\begin{description}
		\item[Data] a list $\Gamma$ of \MELL formulas;
		\item[Question] is there a cut-free \MELL proof-structure of type $\Gamma$ without strict contractions?
	\end{description}
\end{proposition}

\begin{proof}
	Let $\Gamma$ be a list of \MELL formulas. 
	Termination of $\mathcal{P}_2$ on input $\Gamma$ is proved above.
	
	If $\mathcal{P}_2$ accepts $\Gamma$, then there is a cut-free scheduling $\sched \colon (\Gamma) \to \emptylists$ that does not contain $\Contr_{i}$.
	According to Definition \ref{def:building}, it is immediate to prove that $R(\sched)$ has no strict contractions.
	
	If there is a cut-free \MELL proof-structure $R$ of type $(\Gamma)$ without strict contractions, then there is a  scheduling $\sched \colon (\Gamma) \to \emptylists$ such that $R \pnrewrite[\sched] \emptynets$ without nullary $\Contr$-splittings (Proposition \ref{prop:unwinding-to-empty}). 
	Since $R$ has no strict contractions, then $\sched$ does not contain any $\Contr_i$, so $\mathcal{P}_2$ accepts $\Gamma$. 
\end{proof}

Type inhabitation for cut-free \MELL proof-structures depends on whether strict contractions are allowed or not: the list $\oc X, \wn X^\perp, \oc X$ is accepted by $\mathcal{P}_1$ but rejected by $\mathcal{P}_2$, indeed the only cut-free \MELL proof-structure of type $\oc X, \wn X^\perp, \oc X$ has a strict contraction.

\begin{remark}[on decidability of $\MELL$]
  Our characterization of type inhabitation for cut-free $\MELL$ proof-structures (Theorem \ref{thm:inhabitation}) can be 
  seen in relation to the question of deciding if a \(\MELL\) formula is provable: indeed, a
	\MELL formula \(A\) is provable if and only if there is a \emph{correct} cut-free $\MELL$ proof-structure of type $A$ (for some notion of correctness).

  In the propositional setting, $\MLL$ provability is decidable and
  \(\mathbf{NP}\)-complete \cite{Kanovich94},   while 
  \(\LL\) provability is undecidable  \cite{Lincoln1992}; 
  a longstanding open problem is whether provability in \(\MELL\) is  decidable or not.
	Our result offers a fresh perspective on it.
	
	Take again procedure $\mathcal{P}_1$ above 
	with the same input $\Gamma$, and replace its second step by:
	\begin{enumerate}
		\setcounter{enumi}{1}
		\item if a list of the form $\emptylists$ is reached by a cut-free scheduling $\sched$, construct the cut-free \MELL proof-structure $R = R(\sched)$ of type $\Gamma$; 
		\item check if $R$ 
		fulfills a suitable correctness criterion that characterizes all and only the \MELL proof-structures corresponding to a derivation in
		\MELL sequent calculus, \eg \cite{Tortora03};
		\item as soon as a correct (cut-free) \MELL proof-structure is found, accept $\Gamma$.
	\end{enumerate}
	As explained above, there may be infinitely many cut-free schedulings  $\sched \colon (\Gamma) \to \emptylists$, and also infinitely many cut-free \MELL proof-structures of type $\Gamma$. 
	But they are finite in number in the case without strict contractions.
	As correctness of a \MELL proof-structure is decidable, the procedure above \emph{semi-decides} provability in \MELL (is a given sequent provable in \MELL sequent calculus?) and \emph{decides} provability in \MELL \emph{without contractions} (is a given sequent provable in \MELL sequent calculus without the contraction rule?). 
	Our procedure restricts its search to cut-free objects, without loss of generality by cut-elimination.
	Semi-decidability of \MELL provability and decidability of \MELL provability without contractions are not new results; 
	another (and simpler) proof relies on the fact that in \MELL sequent calculus every inference rule other than contraction and cut, read bottom-up, decreases  the size of a sequent, for a suitable notion of size. 
	What is new is the analogy between the type inhabitation problem for cut-free \MELL proof-structures and the provability problem in \MELL.

  Summing up, (strict) contraction is the only reason that makes hard to prove whether type inhabitation for cut-free $\MELL$ proof-structures, as well as $\MELL$ provability, are decidable or not.
  This seems to suggest that \emph{correctness} (the fact that a $\MELL$ proof-structure is a proof of a $\MELL$ formula) does not play an essential role in the question if $\MELL$ provability is decidable or not. 
  However, for the time being, we are not able to reduce decidability of provability in \MELL to decidability of type inhabitation for cut-free \MELL proof-structures. 	
\end{remark}


\section{Decidability of the finite inverse Taylor expansion problem}
\label{sec:decidability}

Consider the decision problem below, called the \emph{finite inverse Taylor expansion problem}.

\begin{description}
	\item[Data] a \emph{finite} set $\Pi$ of $\DiLL_0$ proof-structures without $\maltese$-cells;
	\item[Question] does there exist a \MELL proof-structure $R$ such that $\Pi \subseteq \Taylor{R}$?
\end{description} 

By Theorem~\ref{thm:characterization}, 
a positive answer is equivalent to the \gluability of \(\Pi\), thus the question amounts to decide the \kl[glueable]{\gluability} of a finite set $\Pi$ of $\DiLL_0$ proof-structures without $\maltese$-cells.
In this section, we show that the problem is actually \emph{decidable}, using our rewrite rules.
The challenge here is to restrict the search space for rewritings to guarantee termination.

Recall that the action of $\BoxR_{i}$ is the only rewrite rule in $\DiLL_0$ that can introduce $\maltese$-cells (\Cref{fig:empty-box}).
We identify $\maltese$-cells that have the same number of outputs, with the same type.

\begin{remark}[$\maltese$-cell]
Let $\Sigma, \Pi$ be sets of $\DiLL_0$ quasi-proof-structures such that every element of $\Sigma$ is without $\maltese$-cells, $\Sigma \pnrewrite[\sched] \Pi$ and $\Pi$ is \gluable.
The presence of $\maltese$-cells in $\Pi$  witnesses the absence of information on the content of some box of some element of $\Sigma$. 
However, in order that $\Pi$ is \gluable, either this lack of information in some element of $\Pi$ is provided by another element of $\Pi$, or the information is plainly missing. 
Indeed, let $\Pi$ be a set of $\DiLL_0$ proof-structures, and $i$ be a conclusion of some $\rho \in\Pi$ and the output of a $\maltese$-cell of $\rho$; 
we have two cases (Definition \ref{def:daimon-nonDaimonSched} below relies on this distinction):
\begin{enumerate}
	\item either $i$ is the output of a $\ell$-cell of \emph{some} $\DiLL_0$ quasi-proof-structure $\pi \in\Pi$, with $\ell\neq\maltese$;
	\item or $i$ is the output of a $\maltese$-cell of \emph{every} $\DiLL_0$ quasi-proof-structure $\pi \in\Pi$; in this case the component of every $\pi \in\Pi$ having $i$ among its conclusions consists of the same $\maltese$-cell. 
\end{enumerate}
\end{remark}

Let $n>0$ and $\Gamma = (\Gamma_{\!1}; \dots; \Gamma_{\!n})$ be a list of lists of \MELL formulas. 
	We denote by $\maltese_\Gamma$ the $\DiLL_0$ quasi-proof-structure $(\rho_{\Gamma_{\!1}}, \dots, \rho_{\Gamma_{\!n}})$ of type $\Gamma$, where $\rho_{\Gamma_{\!i}}$ is the daimon of type $\Gamma_{\!i}$ if $\Gamma_{\!i} \neq \emptylist$, and $\rho_{\Gamma_{\!i}} = \emptynet$ if $\Gamma_{\!i} = \emptylist$.
	Lemma \ref{lemma:action-daimons} below states that every $\{\maltese_\Gamma\}$ is \gluable: 
	an example for $\Gamma = (X)$ with $X$ \kl{atomic} is in \Cref{fig:rewrite-with-cut} (see the daimon $\sigma$ there).
A set $\Pi$ of $\DiLL_0$ quasi-proof-structures is \intro[full]{$\maltese$-full} if $\Pi = \{\maltese_\Gamma\}$ for some $\Gamma$.
Note that any $\{\emptynets\}$ is $\maltese$-full.

An elementary scheduling $\Cut^i$ is \emph{with conclusion} if $\Gamma_{\!k} \neq \emptylist$ in its definition in \Cref{fig:elementary-schedulings}.

\begin{lemma}[\gluability of daimons]
	\label{lemma:action-daimons}
	Every $\maltese$-full set of $\DiLL_0$ quasi-proof-structures is \gluable.
	More precisely, let $\Gamma$ be a non-empty list of lists of \MELL formulas.
	Then,
		\begin{enumerate}
			\item\label{p:action-daimons-scheduling} there is a scheduling $\sched \colon \Gamma \to \emptylists$ such that every $\Cut^i$ in $\sched$ is with conclusion;
			\item\label{p:action-daimons-glueable} if $\sched \colon \Gamma \to \emptylists$ and every $\Cut^i$ in $\sched$ is with conclusion, then $\{\maltese_{\Gamma}\} \pnrewrite[\sched] \{\emptynets\}$.
		\end{enumerate}
\end{lemma}

\begin{proof} 
	By Points \ref{p:action-daimons-scheduling} and \ref{p:action-daimons-glueable}, every $\maltese$-full set of $\DiLL_0$ quasi-proof-structures is \gluable.
	\begin{enumerate}		
		\item For every \MELL formula $A$, we define by induction on $A$ a scheduling $\sched_A \colon (A) \to \emptylists$ where every $\Cut^i$ in $\sched_A$ is with conclusion.
		We shall use the fact that for every \MELL formula $A$ and lists $\Delta, \Gamma$ of \MELL formulas, given a scheduling $\sched_A \colon (A) \to \emptylists$, one has $\sched_A \colon (A, \Delta; \Gamma) \to (\Delta; \Gamma)$ and $\sched_A \colon (\emptylist; A, \Gamma) \to (\emptylist; \Gamma)$.
		Let $\schedtwo_A = \Cut^2 \Contr_2 \Mix_{2} \BoxR_{4} \Der_2 \Der_3 \colon \allowbreak (A) \to (A, A^\perp; A^\perp, A)$ and note that $\Cut^2$ in $\schedtwo_A$ is with conclusion.
		\begin{itemize}
			\item If $A = X$ (and similarly if $A = X^\perp$), then $\sched_A = \schedtwo_A \Ax_{4} \Ax_{2}$.
			\item If $A = \One$ then $\sched_A = \One_1$; if $A = \bot$ then $\sched_A = \bot_1$.
			\item If $A = B \otimes C$ (similarly if $A = B \parr C$), $\sched_A = \schedtwo_A \otimes_4 \parr_3 \parr_2 \otimes_1 \sched_{B} \sched_{C} \sched_{B^\perp} \sched_{C^\perp} \sched_{B^\perp} \sched_{C^\perp} \sched_{B} \sched_{C}$.
			\item If $A = \oc B$ (and similarly if $A = \wn B$), then $\sched_{A} = \schedtwo_{A} \BoxR_4 \Der_3 \BoxR_1 \Der_2 \sched_{B} \sched_{B^\perp}  \sched_{B^\perp} \sched_{B}$.
		\end{itemize}
		Let $\Gamma = (A_1, \dots, A_{i_1}; \cdots; A_{i_{n\!-\!1}\!+\!1}, \dots, A_{m})$ be a non-empty list of lists of \MELL formulas. 
		In the scheduling $\sched_{A_1\!} \cdots \sched_{A_{i_1}\!} \cdots \sched_{A_{i_{n\!-\!1}\!+\!1}\!} \cdots \sched_{A_{m}} \colon \Gamma \to \emptylists$ every $\Cut^i$ is with conclusion.
		
		\item By induction on the length $k \geqslant 0$ of the scheduling $\sched = a_1 \!\cdots a_k$, where $a_j$ is an elementary scheduling for all $1 \leqslant j \leqslant k$. 
			It is crucial that every $a_j = \Cut^i$ is with conclusion: otherwise, for a daimon $\rho'$ there is no $\rho$ such that $\{\rho\} \pnrewrite[\!\Cut^i\!] \{\rho'\}$ (see \Cref{fig:cut-daimon}).
			\qedhere
	\end{enumerate}
\end{proof}

Therefore, by Lemma \ref{lemma:action-daimons}, to prove that a set $\Pi$ of $\DiLL_0$ proof-structures is \gluable, it suffices to show that  $\Pi$ rewrites to a \kl[full]{$\maltese$-full} set $\Pi'$.
A more subtle question is to prove that if $\Pi$ cannot rewrite to a \kl[full]{$\maltese$-full} set $\Pi'$, then $\Pi$ is not \gluable (Lemma \ref{lemma:CanonicalScheduling} below).

\begin{definition}[$\maltese$-/$\lnot\maltese$-actions, canonicity]
	\label{def:daimon-nonDaimonSched}
Let $\Pi$ be a set of $\DiLL_0$ quasi-proof-structures with the same conclusions.
Suppose that $\Pi\pnrewrite[a]\Pi'$ where $a \neq \Cut^{i}$ (resp.\ $a = \Cut^{i}$) is an elementary scheduling that 
acts on the conclusion $i$ of all elements of $\Pi$ (\resp $\Pi'$). If for every $\rho \in\Pi$ (resp.\ $\rho'\in\Pi'$) the conclusion $i$ of $\rho$ (\resp $\rho'$) is the output of a $\maltese$-cell, then 
we say that $\Pi \pnrewrite[a] \Pi'$ is a \intro[daimon action]{$\maltese$-action}. 
If $\Pi \pnrewrite[\sched]\Pi'$ and for every elementary scheduling $a$ such that $\Pi \pnrewrite[\sched]\Pi'=\Pi \pnrewrite[\sched_{1}]\Pi_{1}\pnrewrite[a]\Pi_{2}\pnrewrite[\sched_{2}]\Pi'$ 
one has that $\Pi_1 \pnrewrite[a] \Pi_2$ is (resp.\ is not) a $\maltese$-action, we say that $\Pi\pnrewrite[\sched]\Pi'$ is a \emph{$\maltese$-action} (\resp \intro[nondaimon action]{$\neg\maltese$-action}), written $\Pi\pnrewrite[\maltese]\Pi'$ (\resp $\Pi\pnrewrite[\!\neg\maltese]\Pi'$). 

If $\Pi\pnrewrite[\sched]\{\emptynets\} = \Pi\pnrewrite[\!\neg\maltese]\,\pnrewrite[\maltese]\{\emptynets\}$, 
then $\Pi\pnrewrite[\sched]\{\emptynets\}$ is \intro{canonical} and $\Pi$ is \intro[canonically glueable]{canonically \gluable}.
\end{definition}

So, given a set $\Pi_0$ of $\DiLL_0$ quasi-proof-structures and a $\maltese$-action $\Pi_0 \pnrewrite[\!a_1] \allowbreak\dots \pnrewrite[\!a_n] \Pi_n$, for all $1\leqslant j \leqslant n$ the elementary scheduling $a_j$ acts on a conclusion of a daimon in \emph{all} elements of $\Pi_{j-1}$.
	On the other hand, in a $\lnot\maltese$-action $\Pi_0 \pnrewrite[\!a_1] \allowbreak\dots \pnrewrite[\!a_n] \Pi_n$, for all $1\leqslant j \leqslant n$ the elementary scheduling $a_j$ acts on the output of a $\ell$-cell with $\ell \neq \maltese$ in \emph{some} element of~$\Pi_{j-1}$.

\begin{remark}[fullness]
	\label{rmk:full}
	Let $\Pi$ be a set of $\DiLL_0$ quasi-proof-structures, and $\sched$ be a scheduling.
	By straightforward induction on the length of $\sched$, we can prove the facts below. 
	\begin{itemize}
		\item If $\Pi$ is $\maltese$-full and $\Pi \pnrewrite[\sched] \Pi'$ for some set $\Pi'$, then $\Pi \pnrewrite[\sched] \Pi'$ is a $\maltese$-action and $\Pi'$ is $\maltese$-full.
		\item If $\Pi \pnrewrite[\sched] \{\emptynets\}$, then $\Pi \pnrewrite[\sched] \{\emptynets\}$ is a \kl[daimon action]{$\maltese$-action} if and only if $\Pi$ is $\maltese$-\kl{full}. 
		Therefore, $\Pi\pnrewrite[\sched]\{\emptynets\}$ is \kl{canonical} if and only if $\Pi\pnrewrite[\sched]\{\emptynets\} = \Pi \pnrewrite[\lnot\maltese] \Pi' \pnrewrite[\maltese] \{\emptynets\}$ for some \kl[full]{$\maltese$-full} $\Pi'$.
	\end{itemize}
\end{remark}

The following lemma relies on the fact that $\maltese$-actions can always be delayed.

\begin{lemma}[postponement of $\maltese$-actions]
	\label{lemma:CanonicalScheduling}
	Every \kl[glueable]{\gluable} set of $\DiLL_0$ quasi-proof-structures is \kl[canonically glueable]{canonically \gluable}.
\end{lemma}

\begin{proof}
Immediate consequence of the following fact: if $\Pi_{1}, \Pi_2, \Pi_3$ are sets of $\DiLL_0$ quasi-proof-structures such that $\Pi_{1}\pnrewrite[\!a_1]\Pi_{2}$ is a $\maltese$-action and $\Pi_{2}\pnrewrite[\!a_2]\Pi_{3}$ is not a $\maltese$-action for some elementary schedulings $a_1, a_2$, then there are elementary schedulings $c_1,c_2$ and a set $\Sigma$ of $\DiLL_0$ quasi-proof-structures such that $\Pi_{1}\pnrewrite[\!c_1] \Sigma$ is not a $\maltese$-action and $\Sigma \pnrewrite[\!c_2]\Pi_{3}$ is a $\maltese$-action. 
Notice that $c_1$ (\resp $c_2$) is ``essentially'' the same as $a_2$ (\resp $a_1$).
\end{proof}

We have defined nullary $\Contr$-splitting actions on \MELL quasi-proof-structures in \Cref{sect:rules}. This notion can be easily extended to $\DiLL_0$.
Let $\Pi$ be a set of $\DiLL_0$ quasi-proof-structures. 
	Given an elementary scheduling $a$, $\Pi \pnrewrite[a] \Pi'$ is \intro[dill0 nullary splitting]{nullary $\Contr$-splitting} if $a = 
	\Contr_i$ and, for every $\rho \in \Pi$, it splits the $\wn$-\kl{cell} of $\rho$ with output $i$ and $h_\rho \geqslant 0$ inputs into two $\wn$-cells, one $\wn$-cell with output $i$ (\resp $i+1$) and $h_\rho$ inputs, the other $\wn$-cell with output $i+1$ (\resp $i$) and $0$ inputs.
Given a scheduling $\sched$ such that $\Pi\pnrewrite[\sched]\Pi'$, \intro[number of nullary splittings]{$\Contr^{0}(\Pi \pnrewrite[\sched] \Pi')$} is the number of nullary $\Contr$-splitting actions in $\Pi\pnrewrite[\sched]\Pi'$, that is, the number of elementary schedulings $a = \Contr_i$ such that $\Pi \pnrewrite[\sched] \Pi' = \Pi \pnrewrite[\sched_{1}] \Pi_{1} \pnrewrite[a] \Pi_{2} \pnrewrite[\sched_{2}] \Pi'$ and $\Pi_{1} \pnrewrite[a]\Pi_{2}$ is nullary $\Contr$-splitting. 
	We say that:
	\begin{itemize}
		\item $\Pi\pnrewrite[\sched]\Pi'$ is \emph{without nullary $\Contr$-splittings} if $\Contr^{0}(\Pi \pnrewrite[\sched] \Pi') = 0$; \emph{has nullary $\Contr$-splittings} otherwise;
		\item $\Pi$ is \emph{canonically \gluable without nullary $\Contr$-splittings} if there is a scheduling $\sched$ such that $\Pi \pnrewrite[\sched] \{\emptynets\}$ is canonical and without nullary $\Contr$-splittings.
	\end{itemize}

For instance, in Example \ref{ex:church}, $\{\pi\} \pnrewrite[\Contr_2] \{\pi'\}$ is nullary $\Contr$-splitting.
And $\{\sigma_1, \sigma_2\} \pnrewrite[\Contr_2] \Sigma$ is nullary $\Contr$-splitting if $\Sigma = \{\sigma_1', \sigma_2'\}$, but it is not nullary $\Contr$-splitting if $\Sigma = \{\sigma_1', \sigma_2''\}$, where $\sigma_1, \sigma_2, \sigma_1', \sigma_2', \sigma_2''$ are the $\DiLL_0$ proof-structures below.
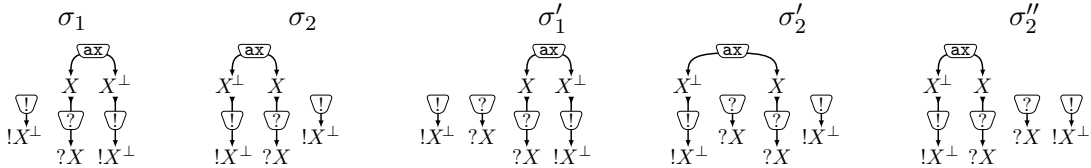
\begin{center}
	$\sigma_1
	\qquad\qquad\qquad\quad
	\sigma_2
	\qquad\qquad\qquad\quad\ \
	\sigma_1'
	\qquad\qquad\qquad\quad\
	\sigma_2'
	\qquad\qquad\qquad\quad
	\sigma_2''$
	
	\vspace{\beforepn}
		\scalebox{0.8}{
		\pnet{
			\pnformulae{~\pnf[na]{$X$}~\pnf[ba]{$X^{\bot}$}
				\\
				\pnf[ba']{$\oc X^{\!\bot}\,$}
			}
			\pncown{ba'}
			\pnaxiom{na,ba}
			\pnbag{}{ba}[bba]{$\,\oc X^{\bot}$}
			\pnexp{}{na}{$\wn X $}[1]
		}
	}
	\qquad
	\scalebox{0.8}{
		\pnet{
			\pnformulae{
				\pnf[ba']{$X^{\!\bot}\,$}~\pnf[na]{$X$}
				\\
				~~\pnf[ba]{$\oc X^{\bot}$}
			}
			\pncown{ba}
			\pnaxiom{na,ba'}
			\pnbag{}{ba'}[bba']{$\,\oc X^{\bot}$}
			\pnexp{}{na}{$\wn X $}[1]
		}
	}
	\qquad
	\scalebox{0.8}{
		\pnet{
			\pnformulae{~~\pnf[na]{$X$}~\pnf[ba]{$X^{\bot}$}
				\\
				\pnf[ba']{$\oc X^{\!\bot}$}~\pnf[na']{$\wn X$}
			}
			\pnwn{na'}
			\pncown{ba'}
			\pnaxiom{na,ba}
			\pnbag{}{ba}[bba]{$\,\oc X^{\bot}$}
			\pnexp{}{na}{$\wn X $}[1]
		}
	}
	\qquad
	\scalebox{0.8}{
		\pnet{
			\pnformulae{
				\pnf[ba']{$X^{\!\bot}$}~~\pnf[na']{$X$}
				\\
				~\pnf[na]{$\wn X$}~~\pnf[ba]{$\oc X^{\bot}$}
			}
			\pnwn{na}
			\pncown{ba}
			\pnaxiom{na',ba'}
			\pnbag{}{ba'}[bba']{$\,\oc X^{\bot}$}
			\pnexp{}{na'}{$\wn X $}[1]
		}
	}
	\qquad
	\scalebox{0.8}{
		\pnet{
			\pnformulae{
				\pnf[ba']{$X^{\!\bot}$}~\pnf[na']{$X$}
				\\
				~~\pnf[na]{$\wn X$}~\pnf[ba]{$\oc X^{\bot}$}
			}
			\pnwn{na}
			\pncown{ba}
			\pnaxiom{na',ba'}
			\pnbag{}{ba'}[bba']{$\,\oc X^{\bot}$}
			\pnexp{}{na'}{$\wn X $}[1]
		}
	}
\end{center}

\begin{remark}[need for nullary splitting]
	\label{rmk:nullary-splitting}
	For $\pi_h$ ($h > 0$) as in \Cref{note:church} on p.~\pageref{note:church}, one has $\{\pi_1,\pi_2\} \allowbreak \pnrewrite[\sched] \allowbreak\{\emptynets\}$ only with schedulings $\sched$ whose action on $\{\pi_1\}$ has nullary $\Contr$-splittings. 
	So, nullary $\Contr$-splittings are somehow ``necessary'' in $\DiLL_0$, differently from \MELL where they are superfluous for normalization (Proposition \ref{prop:unwinding-to-empty}).
	To make diagram \eqref{eq:naturality-square} commute for naturality (\Cref{thm:projection-natural}), nullary $\Contr$-splitting actions in $\DiLL_0$ must be mimicked by nullary $\Contr$-splitting actions on \MELL, hence the action of $\Contr_{i}$ in \MELL must allow nullary $\Contr$-splittings. 
\end{remark}

Akin to \MELL, nullary $\Contr$-splitting actions in $\DiLL_0$ are superfluous as means of destruction.
They may be needed for \emph{some} elements in a set $\Pi$ of $\DiLL_0$ quasi-proof-structures, but not for \emph{all} of them, thus nullary $\Contr$-splitting actions are useless for the set $\Pi$ as a whole.

\begin{lemma}[splitting in $\DiLL_0$]
	\label{lemma:NoContrCanonicalScheduling}
	Every \kl[canonically glueable]{canonically \gluable} set $\Pi$ of $\DiLL_0$ quasi-proof-structures is
	canonically \gluable \kl[dill0 nullary splitting]{without nullary $\Contr$-splittings}.
\end{lemma}

\begin{proof}
	By Remark \ref{rmk:full}, it suffices to prove that for every canonically \gluable set $\Pi$ of $\DiLL_0$ quasi-proof-structures, there is a scheduling $\sched$ such that $\Pi\pnrewrite[\sched]\Pi'$ is a $\lnot\maltese$-action without $\Contr$-splittings, for some \kl[full]{$\maltese$-full} $\Pi'$.
	Indeed, by Lemma \ref{lemma:action-daimons}, there is a scheduling $\schedtwo$ such that $\Pi' \pnrewrite[\schedtwo] \{\emptynets\}$, it is a \kl[daimon action]{$\maltese$-action} by Remark \ref{rmk:full} and is clearly without nullary $\Contr$-splittings.
	
	Given a $\lnot\maltese$-action $\Pi\pnrewrite[\sched]\Pi'$ where $\Pi'$ is $\maltese$-full, we prove by induction on $\kl[number of nullary splittings]{\Contr^{0}(\Pi \pnrewrite[\sched] \Pi')} \in \Nat$ that it is possible to build a $\neg\maltese$-action $\Pi\pnrewrite[\sched']\Pi''$ without nullary $\Contr^{0}$-splittings for some $\maltese$-full $\Pi''$. 
	Clearly, it is enough to show that if $\kl[number of nullary splittings]{\Contr^{0}(\Pi \pnrewrite[\sched] \Pi')} > 0$ then we can find a $\maltese$-full $\Pi''$ and a $\neg\maltese$-action $\Pi\pnrewrite[\sched']\Pi''$ such that $\Contr^{0}(\Pi \pnrewrite[\sched] \Pi') > \Contr^{0}(\Pi \pnrewrite[\sched'] \Pi'')$.
	
	Let $\Pi \pnrewrite[\sched]\Pi' = \Pi \pnrewrite[\sched_{1}] \Pi_{1} \pnrewrite[a]\Pi_{2} \pnrewrite[\sched_{2}]\Pi'$ be a $\neg\maltese$-action, where $\Pi'$ is $\maltese$-full, $a = \Contr_{k}$ and $\Pi_{1} \pnrewrite[a]\Pi_{2}$ is nullary $\Contr$-splitting. 
	Since the $\wn$-\kl{cells} with $0$ inputs of $\Pi_{2}$ created by $\Pi_1 \pnrewrite[a] \Pi_2$ must disappear in $\Pi'$ (which contains only $\maltese$-cells), there are only two cases for $\Pi_{2} \pnrewrite[\sched_{2}] \Pi'$:
	\begin{enumerate}
		\item\label{item:1NoContrCanonicalScheduling}
		either $\Pi_{2}\pnrewrite[\sched_{2}]\Pi' = \Pi_{2} \pnrewrite[\sched_{21}] \Pi_{3} \pnrewrite[\Weak_i] \Pi_{4} \pnrewrite[\sched_{22}] \Pi'\!$, where $\Pi_{3} \pnrewrite[\Weak_i] \Pi_{4}$ erases the $\wn$-cells without inputs and their unique output (a conclusion of each element of $\Pi_3$) created by $\Pi_{1} \pnrewrite[a]\Pi_{2}$;
		
		\item\label{item:2NoContrCanonicalScheduling}
		or $\Pi_{2} \pnrewrite[\sched_{2}] \Pi' = \Pi_{2} \pnrewrite[\sched_{21}] \Pi_{3} \pnrewrite[\!\BoxR_{i}]\Pi_{4} \pnrewrite[\sched_{22}]\Pi'$, where the component of every element of $\Pi_{3}$ containing the $\wn$-cell without inputs created by $\Pi_{1} \pnrewrite[a]\Pi_{2}$ is transformed by $\Pi_{3} \pnrewrite[\!\BoxR_{i}]\Pi_{4}$ into a $\maltese$-cell, according to the rewrite rule in \Cref{fig:empty-box} applied to all the elements of $\Pi_{3}$.
	\end{enumerate}
	
	In Case~\ref{item:1NoContrCanonicalScheduling}, just erase the actions $\pnrewrite[a]$ and $\pnrewrite[\Weak_i]$ from $\Pi \pnrewrite[\sched] \Pi'$: let $\Pi\pnrewrite[\sched'] \Pi' = \Pi \pnrewrite[\sched_{1}] \allowbreak\Pi_{1} \allowbreak\pnrewrite[\schedtwo_{21}] \allowbreak \Pi_{4} \allowbreak\pnrewrite[\sched_{22}] \Pi'$, where $\Pi_{1} \pnrewrite[\schedtwo_{21}] \Pi_{4}$ is (essentially) the same as $\Pi_2 \pnrewrite[\sched_{21}] \Pi_3$ except that the name of the conclusions on which the elementary actions apply have changed (since the action $\Pi_1 \pnrewrite[a] \Pi_2$ is no more present). 
	Clearly, $\Pi \pnrewrite[\sched'] \Pi'$ is a $\neg\maltese$-action and  $\Contr^{0}(\Pi \pnrewrite[\sched] \Pi') > \Contr^{0}(\Pi \pnrewrite[\sched'] \Pi')$.
	
	In Case~\ref{item:2NoContrCanonicalScheduling}, erase the action $\pnrewrite[a]$ and slightly modify the action $\pnrewrite[\BoxR_{i}]$: indeed, let $\Pi\pnrewrite[\sched']\Pi'' = \Pi \pnrewrite[\sched_{1}] \Pi_{1} \pnrewrite[\schedtwo_{21}] \Pi'_{3} \pnrewrite[\BoxR_{i\!-\!1}] \Pi'_{4} \pnrewrite[\schedtwo_{22}] \Pi''$, where 
	\begin{itemize}
		\item
		$\Pi_{1} \pnrewrite[\schedtwo_{21}] \Pi'_{3}$ is the same as $\Pi_2 \pnrewrite[\sched_{21}] \Pi_3$ except that the name of the conclusions on which the elementary actions apply have changed (since the action $\Pi_1 \pnrewrite[a] \Pi_2$ is no more present);
		
		\item	the set $\Pi'_{3}$ is $\Pi_{3}$ with one conclusion and one $\wn$-cell with $0$ inputs less in each of its elements;
		
		\item	the action $\Pi'_{3} \pnrewrite[\BoxR_{i\!-\!1}] \Pi'_{4}$ applies the rewrite rule in Figure~\ref{fig:empty-box} to the conclusion $i-1$ of every element in $\Pi'_{3}$, which corresponds to the conclusion $i$ of every element in $\Pi_{3}$;
		
		\item	the $\maltese$-cell of each element in $\Pi'_{4}$ created by $\Pi'_{3} \pnrewrite[\BoxR_{i\!-\!1}] \Pi'_{4}$ has one output less than the $\maltese$-cell of each element in $\Pi_{4}$ created by $\Pi_{3} \pnrewrite[\BoxR_{i}] \Pi_{4}$;
		the action $\Pi'_{4} \pnrewrite[\schedtwo_{22}] \Pi''$ and the set $\Pi''$ are obtained from $\Pi_{4} \pnrewrite[\sched_{22}] \Pi'$ and $\Pi'$ accordingly (in particular $\Pi''$ is $\maltese$-full).
	\end{itemize}
	Clearly, even in Case~\ref{item:2NoContrCanonicalScheduling}, $\Pi \pnrewrite[\sched'] \Pi''$ is a $\neg\maltese$-action and  $\Contr^{0}(\Pi \pnrewrite[\sched] \Pi') > \Contr^{0}(\Pi \pnrewrite[\sched'] \Pi'')$.
\end{proof}

By Lemmas \ref{lemma:CanonicalScheduling} and \ref{lemma:NoContrCanonicalScheduling}, a set $\Pi$ of $\DiLL_0$ quasi-proof-structures is \gluable if and only if it is canonically \gluable without nullary $\Contr$-splittings.
This restricts the search space for rewritings to decide if $\Pi$ is \gluable, making \gluability a  \emph{decidable} problem in the finite case.

\begin{proposition}[decidability of finite \gluability]
	\label{prop:decisionFiniteDai}
The problem below is decidable.

\begin{description}
	\item[Data] a finite set $\Pi$ of $\DiLL_0$ quasi-proof-structures;
	\item[Question] is $\Pi$ \gluable?
\end{description}
\end{proposition}

\begin{proof}
	We first define some notions that will be used in the proof.
	Akin to the size of a \MELL quasi-proof-structure used in the proof of Proposition~\ref{prop:unwinding-to-empty}, the \emph{size} $\Size{\Pi}$ of a finite set $\Pi$ of $\DiLL_0$ quasi-proof-structures is the quadruple $(p,q,r,s)$~where:	
	\begin{itemize}
		\item $p$ is the finite multiset of the number of inputs of each $\wn$-cell with at least one input in each element of $\Pi$ (hence, $p$ is a finite multiset of positive integers);
		\item $q$ is the number of cells in $\Pi$ different from $\wn$-cells without inputs and from $\maltese$-cells;
		\item $r$ is the finite multiset $[k_1, \dots, k_n]$ where $n$ is the number of conclusions in each element of $\Pi$ and, for all $1 \leqslant i \leqslant n$, $k_i$ is the number of $\rho \in \Pi$ whose conclusion $i$ is the output of a $\wn$-cell with at least one input;
		\item $s$ is the finite multiset of the number of \kl{conclusions} of each \kl{component} of each element~of~$\Pi$.
	\end{itemize}
	Finite multisets are well-ordered as usual, quadruples are well-ordered lexicographically.
	
	We say that a non-$\maltese$-full set $\Pi$ of $\DiLL_0$ quasi-proof-structures \emph{safely rewrites to} $\Pi'$ if $\Pi \pnrewrite[\sched] \Sigma \pnrewrite[a] \Pi'$ for some set $\Sigma$, some scheduling $\sched$ only made of an arbitrary number (possibly $0$) of $\Ex_{i}$ such that $\Pi \pnrewrite[\sched] \Sigma$ is a $\lnot\maltese$-action, and some elementary scheduling $a \neq \Ex_{i}$ such that $\Sigma \pnrewrite[a] \Pi'$ is neither a $\maltese$-action nor \kl[dill0 nullary splitting]{nullary $\Contr$-splitting}.

	For every finite set $\Pi$ of $\DiLL_0$ quasi-proof-structures, exactly one of the following holds:
	\begin{enumerate}
		\item\label{item:One} either $\Pi$ is $\maltese$-full (this includes the case $\Pi = \{\emptynets\}$) and then it is \gluable by Lemma \ref{lemma:action-daimons};
		
		\item\label{item:Two} or $\Pi$ is not $\maltese$-full, and every scheduling $\sched$ such that $\Pi \pnrewrite[\sched] \Pi'$ is a $\lnot\maltese$-action for some $\Pi'$ is only made of an arbitrary number\,---\,possibly $0$\,---\,of $\Ex_{i}$ (this includes the case when no elementary scheduling applies to $\Pi$, in particular when two elements of $\Pi$ have different types);
		then, by definition, $\Pi$ is \gluable if $\Pi = \emptyset$, and it is not \gluable otherwise;
		
		\item\label{item:Three}
		or $\Pi$ is not $\maltese$-full, and there are $h > 0$ sets $\Pi_{1}, \dots, \Pi_h$ such that $\Pi$ safely rewrites to $\Pi_j$ for all $1 \leqslant j \leqslant h$, and if $\Pi$ safely rewrites to $\Pi'$ then $\Pi' = \Pi_{j}$ for some $1 \leqslant j \leqslant h$; 
		then, $\Pi$ is canonically \gluable without nullary $\Contr$-splittings (which is equivalent to be \gluable, by Lemmas \ref{lemma:CanonicalScheduling} and \ref{lemma:NoContrCanonicalScheduling}) if and only if so is at least one among $\Pi_1, \dots, \Pi_{h}$.
	\end{enumerate}

	Let us prove that when neither Case~\ref{item:One} nor Case~\ref{item:Two} holds, we are in Case \ref{item:Three}.
	If neither Case~\ref{item:One} nor Case~\ref{item:Two} holds, $\Pi$ is not $\maltese$-full and $\Pi \pnrewrite[\sched] \Sigma \pnrewrite[a] \Pi'$ for some sets $\Sigma, \Pi'$, some scheduling $\sched$ only made of $\Ex_{i}$ such that $\Pi \pnrewrite[\sched] \Sigma$ is a $\lnot\maltese$-action, and some elementary scheduling $a \neq \Ex_{i}$ such that $\Sigma \pnrewrite[a] \Pi'$ is not a $\maltese$-action. 
	If $a = \Contr_{i}$ and $\Sigma \pnrewrite[a] \Pi'$ is nullary $\Contr$-splitting, then $i$ is the output of a $\wn$-cell in every element of $\Sigma$, and so there is a set $\Pi''$ such that  $\Sigma \pnrewrite[c] \Pi''$ is neither a $\maltese$-action nor nullary $\Contr$-splitting,  with $c \in \{\Contr_{i}, \Weak_i,\Der_i\}$ (if $c = \Contr_{i}$, $\Sigma \pnrewrite[c] \Pi''$ just splits the inputs of the $\wn$-cell with output $i$ in some element of $\Sigma$ in a way different from $\Sigma \pnrewrite[a] \Pi'$).
	Therefore, $\Pi$ safely rewrites to a finite set of $\DiLL_0$ quasi-proof-structures.
	Clearly, it is impossible that $\Pi$ safely rewrites to infinitely many sets, since each element of $\Pi$ has a finite number of conclusions.
	Hence, Case \ref{item:Three} holds.
	
	The \emph{nondeterministic} procedure to decide if a finite set $\Pi$ of $\DiLL_0$ quasi-proof-structures is \gluable consists of repeating the step described in Case \ref{item:Three},  until all the sets $\Pi_1, \dots, \Pi_n$ which $\Pi$ rewrites to are of the form described in Case \ref{item:One} or \ref{item:Two}: $\Pi$ is \gluable if and only if so is at least one among $\Pi_1, \dots, \Pi_n$. 
	The procedure terminates because (recall the size above):
	\begin{itemize}
		\item if $\Pi \pnrewrite[\Ex_{i}] \Pi'$ then $\Size{\Pi} = \Size{\Pi'}$ but infinite chains of $\Ex_{i}$ are irrelevant (see \Cref{note:exc-schedule});
		\item if $\Pi$ safely rewrites to $\Pi'$ then $\Size{\Pi} > \Size{\Pi'}$ (according to the lexicographic order). 
		\qedhere
	\end{itemize} 
\end{proof}

The size used in the proof of Proposition \ref{prop:decisionFiniteDai} is well defined only if the set $\Pi$ of $\DiLL_0$ quasi-proof-structures is finite. 
The \emph{finiteness} hypothesis in deciding \gluability of $\Pi$ is necessary, as we have seen in Remark \ref{rmk:infinite}: the infinite in depth $\Pi$ shown there is not \gluable and there are infinite rewritings from $\Pi$ that are canonical and without nullary $\Contr$-splittings.

\begin{theorem}[decidability]
	\label{thm:decisionFinite}
The finite inverse Taylor expansion problem is decidable.
\end{theorem}

\begin{proof}
	Let $\Pi$ be a finite set of $\DiLL_0$ proof-structures without $\maltese$-cells.
	By Proposition~\ref{prop:decisionFiniteDai}, we can decide if $\Pi$ is \gluable. 
	By \Cref{thm:characterization}, since there is no $\maltese$-cell in $\Pi$, this is equivalent to deciding if $\Pi \subseteq \Taylor{R}$ for some \MELL proof-structure $R$.
\end{proof}

The nondeterministic procedure to decide the finite inverse Taylor expansion problem is given by the proof of decidability of finite \gluability (Proposition \ref{prop:decisionFiniteDai}) and is \emph{constructive}: if a finite set of $\DiLL_0$ proof-structures without $\maltese$-cells is included in the Taylor expansion $\Taylor{R}$ of some \MELL proof-structure $R$, our procedure allows us to construct one of these $R$.

\section{Non-atomic axioms}
\label{sect:nonatomic}

Our (quasi-)proof-structures (Definitions \ref{def:proof-structure} and \ref{def:quasi-proof-structure}, \Cref{fig:resource-cells}) deal with \emph{atomic axioms} only: the two outputs of any $\Ax$-cell are \kl{dual} atomic formulas.
We can relax the definitions and allow also the presence of \emph{non-atomic axioms}: the outputs of any $\ax$-cell are typed by \kl{dual} \MELL formulas, not necessarily atomic.
We can extend our results to this more general setting, with some technical complications. 
The only trouble is with \emph{exponential axioms}, where the outputs of an $\ax$-cell are typed by \MELL formulas of the form $\oc A^\bot, \wn A$. 
Indeed, consider the daimons $\sigma$ and $\sigma'$ and the $\MELL$ proof-structure $S'$ below.
\vspace{\beforepn}
\begin{center}
	$\sigma = $
  \scalebox{0.8}{
    \pnet{
      \pnformulae{
      \pnf[!A1]{$\oc A^{\bot}$}~\pnf[!A2]{$\oc A^{\bot}$}~\pnf[?A]{$\wn A$}
      }
      \pndaimon[d]{!A1,!A2,?A}[1][.2]
    }
  } 
  \qquad
  $\sigma' = $
  \scalebox{0.8}{
    \pnet{
      \pnformulae{
	\pnf[!A1]{$\oc A^{\bot}$}~\pnf[!A2]{$\oc A^{\bot}$}~\pnf[?A1]{$\wn A$}~\pnf[?A2]{$\wn A$}
      }
      \pndaimon[d]{!A1,!A2,?A1,?A2}
    }
  }
  \qquad
  $S' = $
  \scalebox{0.8}{
    \pnet{
      \pnformulae{
	\pnf[!A1]{$\oc A^{\bot}$}~\pnf[!A2]{$\oc A^{\bot}$}~\pnf[?A1]{$\wn A$}~\pnf[?A2]{$\wn A$}
      }
      \pnaxiom{!A1,?A2}[1.5]
      \pnaxiom{!A2,?A1}
    }
  }
\end{center}
We have that $\{\sigma\} \pnrewrite[\Contr_3] \{\sigma'\} \subseteq \FatTaylor{S'}$.
	But no elementary scheduling $\Contr_3$ can be applied backwards to $S'$ (because of the lack of a $\wn$-cell), which \emph{breaks naturality} (Theorem \ref{thm:projection-natural}).
	This is actually due to the fact that, with exponential axioms, co-functionality (Lemma \ref{prop:unwinding-co-functional}) fails for the
rewrite rule $\pnrewrite[\Contr_i]$: read from right to left, $\pnrewrite[\Contr_i]$ is a functional but not total relation.

The solution we propose here
to deal with non-atomic axioms asks for a technical refinement of only few notions, 
the rest is unchanged and goes smoothly.
	The definitions of \kl{module}, \kl{proof-structure} and \kl{quasi-proof-structure} (Definitions \ref{def:proof-structure} and \ref{def:quasi-proof-structure}) change~as~follows.
	\begin{enumerate}
		\item In \kl{modules}, non-atomic axioms are allowed, except the exponential ones. 
		We add a new type of cells, the $\ocax$-cells, with no inputs and two outputs of type $\oc A^\bot$ and $A$. 
		\item\label{p:no-bangaxiom-alone} In \kl{proof-structures} (and hence in \kl{quasi-proof-structures}), we require that, for every $\ocax$-cell whose outputs have type $\oc A^\bot$ and $A$, the output $o$ of type $A$ is an input of a $\wn$-cell (more precisely, $o$ is an \kl{half} of an \kl{edge} whose other half is the input of a $\wn$-cell).  
	\end{enumerate} 

\noindent
	The requirement ``\kl{atomic}'' in the \kl{elementary scheduling} $\Ax_i$ in \Cref{fig:elementary-schedulings} is replaced by ``non-\kl{exponential}'', and we add the new \kl{elementary scheduling} $\OcAx_i$ below to the list in \Cref{fig:elementary-schedulings},
	\begin{center}
		\small
	$(\Gamma_1; \cdots ; \Gamma_{k\!-\!1}; \FlagType(i), \FlagType(i\!+\!1); \Gamma_{k\!+\!1}; \cdots; \Gamma_n)
	\xrightarrow{\!\OcAx_{i}\!}
	(\Gamma_1; \cdots ; \Gamma_{k\!-\!1}; \emptylist; \Gamma_{k\!+\!1} ; \cdots; \Gamma_n)$  if $\FlagType(i) = \wn A = \FlagType(i\!+\!1)^\bot$
	\end{center}
	whose \kl[action of elementary scheduling]{action}---also in its daimoned version---is like the hypothesis in \Cref{fig:hypothesis,fig:hypothesis-daimon}:

	\vspace{-2\baselineskip}
	\begin{center}
		\scalebox{0.8}{
			\pnet{
				\pnformulae{
					\pnf[i']{$A$}
					\\
					~~\pnf[j]{$i\!+\!1$}
				}
				\pnexp{}{i'}[i]{$i$}
				\pnaxiombang[ax]{i',j}
				\pnsemibox{ax,i',j}
			}
			\ $\pnrewrite[\OcAx_{i}]$ \
			\pnet{
				\pnformulae{
					\pnfsmall[empty]{}
				}      
				\pnsemibox{empty}
			}
			\qquad\qquad\qquad
			\pnet{
				\pnformulae{
					\pnf[i]{$i$}~\pnf[j]{$i\!+\!1$}
				}
				\pndaimon[ax]{i,j}
				\pnsemibox{ax,i,j}
			} 
			\ $\pnrewrite[\OcAx_{i}]$ \
			$\bigg\{
			\pnet{
				\pnformulae{
					\pnfsmall[empty]{}
				}      
				\pnsemibox{empty}
			}
			\bigg\}$
		}
	\end{center}

The way the \kl[action of elementary scheduling]{action of an elementary scheduling} is defined (Definitions \ref{def:unwinding} and \ref{def:paths-on-polyadic}) automatically rules out the possibility that a quasi-proof-structure rewrites to a \kl{module} that is not a quasi-proof-structure.
For instance, the elementary scheduling $\Der_1$ does not \kl{apply} to the (\MELL and $\DiLL_0$) proof-structure $Q$ below, because it would yield a module $Q'$ that is not a quasi-proof-structure (see \Cref{p:no-bangaxiom-alone} above).

\begin{center}
	\vspace{\beforepn}
	$Q = $
	\scalebox{0.8}{
		\pnet{
			\pnformulae{
				\pnf[i']{$A$}~\pnf[j]{$\oc A^\bot$}
			}
			\pnexp{}{i'}[i]{$\wn A$}
			\pnaxiombang[ax]{i',j}
		} 
	}
	$ \ \not\pnrewrite[\Der_{1}] \ $
	\scalebox{0.8}{
		\pnet{
			\pnformulae{
				\pnf[i']{$A$}~\pnf[j]{$\oc A^\bot$}
			}
			\pnaxiombang[ax]{i',j}
		}
	}
	$= Q'$
\end{center}

Apart from the changes we mentioned, all definitions and 
claims in \Cref{sect:proof-nets,sect:taylor,sect:rules,sec:naturality,sec:glueable,sec:decidability} remain unaffected and valid (just add $\OcAx_{i}$ to the list of elementary schedulings that can apply to $R$ in Lemmas \ref{lem:cell-or-mix} and \ref{lemma:contraction}.\ref{p:cell-or-mix-der}).
The counterexample to naturality shown at the beginning of this section does not apply here, 
as $S'$ is not a \MELL quasi-proof-structure (because of its exponential axioms). 
In our framework with non-atomic axioms, the only \MELL proof-structures containing $\sigma'$ in their \kl{filled Taylor expansion} are the ones below (up to the order of their conclusions), and the elementary scheduling $\Contr_3$ \kl{applies} backward~to~each~of~them. 
	
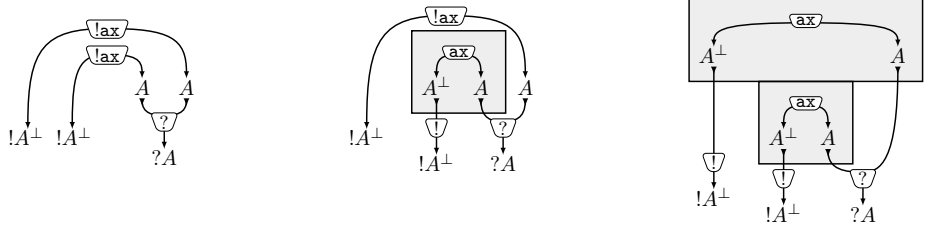
\begin{center}
		\vspace{\beforepn}
	 \scalebox{0.8}{
		\pnet{
			\pnformulae{
				\\
				~~~\pnf[?A1]{$A$}~\pnf[?A2]{$A$}
				\\
				\pnf[!A1]{$\oc A^{\bot}\,$}~\pnf[!A2]{$\,\oc A^{\bot}$}
			}
			\pnaxiombang{!A1,?A2}[1.8]
			\pnaxiombang{!A2,?A1}
			\pnexp{}{?A1,?A2}{$\wn A$}
		}
		\qquad\qquad\qquad
		\pnet{
			\pnformulae{
				\\
				~~\pnf[!A2]{$A^{\bot}$}~\pnf[?A1]{$A$}~\pnf[?A2]{$A$}
				\\
				\pnf[!A1]{$\oc A^{\bot}$}
			}
			\pnaxiombang{!A1,?A2}[2.2]
			\pnaxiom[ax]{!A2,?A1}
			\pnexp{}{?A1,?A2}{$\wn A$}[1.1]
			\pnbag{}{!A2}{$\oc A^\bot$}[1.1]
			\pnbox{ax,!A2,?A1}
		}
		\qquad\qquad\qquad
		\pnet{
			\pnformulae{
				\pnf[!A1]{$A^{\bot}$}~~~~~\pnf[?A1]{$A$}
				\\\\
				~~\pnf[!A2]{$A^{\bot}$}~\pnf[?A2]{$A$}
			}
			\pnaxiom[ax1]{!A1,?A1}
			\pnaxiom[ax2]{!A2,?A2}
			\pnexp{}{?A1,?A2}{$\wn A$}[1.1]
			\pnbag{}{!A2}{$\oc A^\bot$}[1.1]
			\pnbag{}{!A1}{$\oc A^\bot$}[2.2]
			\pnbox{ax2,!A2,?A2}
			\pnbox{ax1,!A1,?A1}
		}
	}
\end{center}

The way we propose here to deal with non-atomic axioms is more elegant and less \textit{ad hoc} than the one we presented in \cite[Section 7]{CSL2020}. 
In the latter, we allowed the outputs of a $\maltese$-cell to be inputs of a $\wn$-cell, and we changed the action of the elementary scheduling $\Contr_i$ for $\DiLL_0$ quasi-proof-structures in the daimoned case (\Cref{fig:contraction-daimon}), by explicitly requiring the presence of a $\wn$-cells above the conclusion $i$.
Consequently, in \cite{CSL2020} we had to redefine the filled Taylor expansion in the non-atomic case, and thus to recheck that naturality holds for \emph{all} the rewrite rules.
Here instead we do not have to redefine the filled Taylor expansion and we have to check naturality (\Cref{thm:projection-natural}) in the non-atomic case only for the action~of~$\OcAx_i$.

As another possible solution to deal with non-atomic axioms, instead of using generalized $\wn$-cells in Definition \ref{def:proof-structure} (with an arbitrary number of inputs of type $A$ and one output of type $\wn A$, as in \cite{DanosRegnier95}), we could have used different kinds of $\wn$-cells for dereliction (one input of type $A$, one output of type $\wn A$), contraction (two inputs of type $\wn A$, one output of type $\wn A$), and weakening (no inputs, one output of type $\wn A$), as in \cite{Girard:1987}.
Such a choice should be supported by a rework of the elegant definition of the Taylor expansion of a \MELL proof-structure via the notion of pullback (Definitions \ref{def:proto} and \ref{def:taylor}), since it collapses dereliction, contraction and weakening in a generalized $\wn$-cell.

\section{Conclusions and perspectives}
\label{sec:conclusions}

\subsubsection*{Daimons for empty boxes}

Our \gluability criterion (\Cref{thm:characterization}) solves the inverse Taylor expansion problem in an ``asymmetric'' way: we characterize the sets of $\DiLL_0$ proof-structures without $\maltese$-cells that are included in the Taylor expansion of some \MELL proof-structure, but in general $\DiLL_0$ proof-structures might contain $\maltese$-cells (while \MELL proof-structures  cannot, see Definition \ref{def:proof-structure}).
\kl{Daimons} and \kl{emptyings} are needed to get a natural transformation 
(via the \emph{\kl{filled}} Taylor expansion, \Cref{thm:projection-natural}), which is the main ingredient to prove our \gluability criterion. 
But we are interested in frameworks without $\maltese$-cells.
This asymmetry is technically inevitable, 
due to the fact that a \gluable set of $\DiLL_0$ proof-structures might not contain any information on the content of some box, when they take 0 copies of~it.

\subsubsection*{Comparison with Pagani and Tasson \cite{PaganiTasson2009}}

Pagani and Tasson's solution \cite{PaganiTasson2009} to the inverse Taylor expansion problem is less
general than ours, because it characterizes \emph{finite} sets of $\DiLL_0$
proof-structures that are included in the Taylor expansion of some
\emph{cut-free} \MELL proof-structure with \emph{atomic axioms}. 
We do not have these limitations (finite, cut-free, atomic axioms), our characterization yields a decision procedure (\Cref{thm:decisionFinite}) in the finite case as in \cite{PaganiTasson2009}, and we can smoothly adapt our criterion to the cut-free case (\Cref{thm:characterization}.\ref{p:characterization-cut-free}, which is not trivial as explained in Example \ref{ex:glueable-non-cut-free}).
In \cite{PaganiTasson2009}, the restrictions to cut-free 
and to atomic axioms aim to simplify their presentation and might be overcome. 
But the restriction to \emph{finite} sets of $\DiLL_0$ proof-structures is somehow essential in their approach. 
Roughly, their machinery takes a finite set $\Pi $ of $\DiLL_0$ proof-structures as input,  juxtaposes its elements in a unique graph $\pi$ and then runs a rewriting by means of some tokens that go through the whole $\pi$ and try to merge its components in a $\MELL$ proof-structure whose Taylor expansion includes $\Pi$, if any. 
If $\Pi$ were an infinite set, $\pi$ would be an infinite graph and, apart from the technical intricacies of dealing with infinite objects, it would be impossible 
to provide a characterization in that case.
In particular, their approach could not distinguish whether $\Pi$ is infinite in width (which can be in the Taylor expansion of some \MELL proof-structure) or infinite in depth (which cannot, see Remark \ref{rmk:infinite}).
Our approach, instead, defines a rewrite relation on a single---finite---element of $\Pi$ and extends it to the whole (possibly infinite) $\Pi$ by requiring that the same rewrite rule applies to each element of $\Pi$ (Definition \ref{def:paths-on-polyadic}). 
Thus, we can accommodate the case where $\Pi$ is infinite.

We believe that our rewriting rules are also simpler than the ones in
\cite{PaganiTasson2009} and rely on a more abstract and less \textit{ad hoc}
property (naturality), that allows us to prove also semi-decidability of
another problem: \emph{type inhabitation} for cut-free \MELL proof-structures.

Finally, Pagani and Tasson's solution of the inverse Taylor expansion problem is affected by another limitation, even though not particularly emphasized in \cite{PaganiTasson2009}: 
their Theorem 2 (analogous to the ``only if'' part of our \Cref{thm:characterization}) assumes not only that their rewriting starting from a set of $\DiLL_0$ proof-structures terminates 
but also that it ends on a \MELL proof-structure, according to their definition of \MELL proof-structure.
This is limiting when in the set of $\DiLL_0$ proof-structures there is no information about the content of a box.
For instance, consider the singletons $\Pi$ and $\Pi'$ of $\DiLL_0$ proof-structures below.
\begin{center}
	\vspace{\beforepn}
	\small
	$\Pi = \bigg\{
	\scalebox{0.8}{
	\pnet{
		\pnformulae{\pnf[X]{$\oc \one$}}
		\pncown[cown]{X}
	}}
	\bigg\}$
	\qquad\qquad\qquad
	$\Pi' = \bigg\{
	\scalebox{0.8}{
	\pnet{
		\pnformulae{\pnf[1]{$\oc X$}}
		\pncown[cown]{1}
	}}
	\bigg\}$
	\qquad\qquad\qquad
		$R =
		\scalebox{0.8}{
	\pnet{
		\pnformulae{\pnf[1']{$\one$}}
		\pnone[1]{1'}
		\pnbox{1,1'}
		\pnbag{}{1'}{$\oc \one$}[1.2]
	}}$
\end{center}
Pagani and Tasson's rewrite rules do not distinguish the two singletons, each one is included in the Taylor expansion of cut-free \MELL proof-structures with an ``empty box'', due to the lack of information.
So their notion of \MELL proof-structure is wider and non-standard (because it allows the presence of empty boxes).
On the contrary, our rewrite rules distinguish $\Pi$ and $\Pi'$: 
via the action of cut-free schedulings, the former can rewrite to $\{\emptynet\}$ and is included in the Taylor expansion of the \MELL proof-structure $R$ above, the latter cannot rewrite to $\{\emptynet\}$ and is not part of the Taylor expansion of any cut-free \MELL proof-structure. 
Summing up, our characterization follows the standard notion of \MELL proof-structures (unlike \cite{PaganiTasson2009}) and is more fine-grained and informative than the one in \cite{PaganiTasson2009}.

\subsubsection*{The \texorpdfstring{\(\lambda\)}{lambda}-calculus, connectedness and coherence}

Our rewriting system and \gluability criterion might help to prove that a binary coherence relation can solve the inverse Taylor expansion problem for \MELL proof-structures fulfilling some geometric property related to connectedness, despite the impossibility for the full \MELL fragment. Such a coherence would extend the coherence criterion for resource $\lambda$-terms. Note that our \gluability criterion is actually an extension of the criterion for resource $\lambda$-terms. Indeed, in the case of the $\lambda$-calculus, there are three rewrite steps, corresponding to abstraction, application and variable (which can be encoded in our rewrite steps), and coherence is defined inductively: if a set of resource
$\lambda$-terms is coherent, then any set of resource $\lambda$-term that rewrites to it is also coherent.

Presented in this way, the main difference between the $\lambda$-calculus and ``connected'' \MELL
(concerning the inverse Taylor expansion problem) would not be because of the rewriting
system, but because the structure of any resource $\lambda$-term uniquely determines the rewriting
path, while for $\DiLL_0$ proof-structures we have to quantify existentially over all possible
paths. This is an unavoidable consequence of the fact that proof-structures do not have a
tree-structure, contrary to $\lambda$-terms.

Moreover, it is possible to match and mix different rewritings. 
Indeed, consider three $\DiLL_0$ proof-structures pairwise \gluable: 
proving that they are \gluable as a whole amounts to computing a rewriting from the three rewritings witnessing their pairwise \gluability. 
Our rewriting system has been designed with that mixing-and-matching operation in mind, in the particular case where the boxes are connected. This is reminiscent of \cite{FSCD2016}, where we also showed that a certain property enjoyed by the \(\lambda\)-calculus can be extended to proof-structures, provided they are connected inside boxes. 
We leave it as future work.

\subsubsection*{Functoriality and naturality}

Our functorial point of view on proof-structures might unify many results. Let us
cite two of them.
\begin{itemize}
\item A sequent calculus proof of \(\vdash \Gamma\) can be translated to a
  path from the empty sequent to \(\Gamma\). 
  This could be the starting point for the formulation of a new correctness criterion.
\item The category \(\Scheduling\) can be extended with a higher
  structure---transforming it from a category into a 2-category---which allows cut-elimination to
  be represented as a 2-arrow. The functors \(\MELLFunctor\) and
  \(\PPolyPN\) can also be extended to 2-functors, so as to prove 
	via naturality that cut-elimination and the Taylor expansion commute.
\end{itemize}

\subsubsection*{Acknowledgments} The authors are grateful to Olivier Laurent, Lionel Vaux and the anonymous reviewers for their insightful comments.


\bibliographystyle{alphaurl}
\bibliography{biblio.bib}

\begin{thebibliography}{LMSS92}

\bibitem[B{\'{e}}c98]{Bechet98}
Denis B{\'{e}}chet.
\newblock Minimality of the correctness criterion for multiplicative proof
  nets.
\newblock {\em Mathematical Structures in Computer Science}, 8(6):543--558,
  1998.

\bibitem[BHP13]{DBLP:conf/csl/BoudesHP13}
Pierre Boudes, Fanny He, and Michele Pagani.
\newblock A characterization of the {Taylor} expansion of lambda-terms.
\newblock In {\em Computer Science Logic 2013 {(CSL} 2013)}, volume~23 of {\em
  LIPIcs}, pages 101--115. Schloss Dagstuhl - Leibniz-Zentrum fuer Informatik,
  2013.
\newblock \href {https://doi.org/10.4230/LIPIcs.CSL.2013.101}
  {\path{doi:10.4230/LIPIcs.CSL.2013.101}}.

\bibitem[BM07]{Borisov2008}
Dennis~V. Borisov and Yuri~I. Manin.
\newblock {\em Generalized Operads and Their Inner Cohomomorphisms}, volume 265
  of {\em Progress in Mathematics}, pages 247--308.
\newblock Birkh{\"a}user Basel, Basel, 2007.
\newblock \href {https://doi.org/10.1007/978-3-7643-8608-5_4}
  {\path{doi:10.1007/978-3-7643-8608-5_4}}.

\bibitem[BM20]{BarbarossaManzonetto20}
Davide Barbarossa and Giulio Manzonetto.
\newblock Taylor subsumes {S}cott, {B}erry, {K}ahn and {P}lotkin.
\newblock {\em Proc. {ACM} Program. Lang.}, 4({POPL}):1:1--1:23, 2020.
\newblock \href {https://doi.org/10.1145/3371069} {\path{doi:10.1145/3371069}}.

\bibitem[Bou93]{Boudol}
G{\'{e}}rard Boudol.
\newblock The lambda-calculus with multiplicities (abstract).
\newblock In {\em 4th International Conference on Concurrency Theory ({CONCUR}
  '93)}, volume 715 of {\em Lecture Notes in Computer Science}, pages 1--6.
  Springer, 1993.
\newblock \href {https://doi.org/10.1007/3-540-57208-2\_1}
  {\path{doi:10.1007/3-540-57208-2\_1}}.

\bibitem[Bou09]{Boudes:2009}
Pierre Boudes.
\newblock Thick subtrees, games and experiments.
\newblock In {\em Typed Lambda Calculi and Applications, 9th International
  Conference ({TLCA} 2009)}, volume 5608 of {\em Lecture Notes in Computer
  Science}, pages 65--79. Springer, 2009.
\newblock \href {https://doi.org/10.1007/978-3-642-02273-9\_7}
  {\path{doi:10.1007/978-3-642-02273-9\_7}}.

\bibitem[CV18]{ChouquetVaux18}
Jules Chouquet and Lionel Vaux.
\newblock An application of parallel cut elimination in unit-free
  multiplicative linear logic to the taylor expansion of proof nets.
\newblock In {\em 27th {EACSL} Annual Conference on Computer Science Logic,
  {CSL} 2018}, volume 119 of {\em LIPIcs}, pages 15:1--15:17. Schloss Dagstuhl,
  2018.
\newblock \href {https://doi.org/10.4230/LIPIcs.CSL.2018.15}
  {\path{doi:10.4230/LIPIcs.CSL.2018.15}}.

\bibitem[dC16]{Carvalho16}
Daniel de~Carvalho.
\newblock The relational model is injective for multiplicative exponential
  linear logic.
\newblock In {\em 25th Annual Conference on Computer Science Logic (CSL 2016)},
  volume~62 of {\em LIPIcs}, pages 41:1--41:19. Schloss Dagstuhl -
  Leibniz-Zentrum fuer Informatik, 2016.
\newblock \href {https://doi.org/10.4230/LIPIcs.CSL.2016.41}
  {\path{doi:10.4230/LIPIcs.CSL.2016.41}}.

\bibitem[dC18]{DBLP:journals/lmcs/Carvalho18}
Daniel de~Carvalho.
\newblock Taylor expansion in linear logic is invertible.
\newblock {\em Logical Methods in Computer Science}, 14(4), 2018.
\newblock \href {https://doi.org/10.23638/LMCS-14(4:21)2018}
  {\path{doi:10.23638/LMCS-14(4:21)2018}}.

\bibitem[dCT12]{CarvalhoTortora}
Daniel de~Carvalho and Lorenzo {Tortora de Falco}.
\newblock The relational model is injective for multiplicative exponential
  linear logic (without weakenings).
\newblock {\em Annals of Pure and Applied Logic}, 163(9):1210--1236, 2012.
\newblock \href {https://doi.org/10.1016/j.apal.2012.01.004}
  {\path{doi:10.1016/j.apal.2012.01.004}}.

\bibitem[DPS21]{Infinets2}
Abhishek De, Luc Pellissier, and Alexis Saurin.
\newblock Canonical proof-objects for coinductive programming: infinets with
  infinitely many cuts.
\newblock In {\em 23rd International Symposium on Principles and Practice of
  Declarative Programming ({PPDP} 2021)}, pages 7:1--7:15. {ACM}, 2021.

\bibitem[DR95]{DanosRegnier95}
Vincent Danos and Laurent Regnier.
\newblock {Proof-nets and the Hilbert Space}.
\newblock In {\em Proceedings of the Workshop on Advances in Linear Logic},
  pages 307--328. Cambridge University Press, 1995.

\bibitem[DS19]{Infinets1}
Abhishek De and Alexis Saurin.
\newblock {Infinets: The parallel syntax for non-wellfounded proof-theory}.
\newblock In {\em {TABLEAUX 2019 - 28th International Conference on Automated
  Reasoning with Analytic Tableaux and Related Methods}}, volume 11714 of {\em
  Lecture Notes in Computer Science}. Springer, 2019.
\newblock \href {https://doi.org/10.1007/978-3-030-29026-9\_17}
  {\path{doi:10.1007/978-3-030-29026-9\_17}}.

\bibitem[Ehr02]{Ehrhard:2002}
Thomas Ehrhard.
\newblock {On K{\"o}the sequence spaces and linear logic}.
\newblock {\em Mathematical Structures in Computer Science}, 12(5):579--623,
  2002.
\newblock \href {https://doi.org/10.1017/S0960129502003729}
  {\path{doi:10.1017/S0960129502003729}}.

\bibitem[Ehr05]{Ehrhard:2005}
Thomas Ehrhard.
\newblock Finiteness spaces.
\newblock {\em Mathematical Structures in Computer Science}, 15(4):615--646,
  2005.
\newblock \href {https://doi.org/10.1017/S0960129504004645}
  {\path{doi:10.1017/S0960129504004645}}.

\bibitem[ER03]{Ehrhard:2003}
Thomas Ehrhard and Laurent Regnier.
\newblock {The differential lambda-calculus}.
\newblock {\em Theoretical Computer Science}, 309(1-3):1--41, 2003.
\newblock \href {https://doi.org/10.1016/S0304-3975(03)00392-X}
  {\path{doi:10.1016/S0304-3975(03)00392-X}}.

\bibitem[ER06a]{EhrhardRegnier06}
Thomas Ehrhard and Laurent Regnier.
\newblock {B{\"{o}}hm Trees, Krivine's Machine and the Taylor Expansion of
  Lambda-Terms}.
\newblock In {\em Second Conference on Computability in Europe (CiE 2006)},
  volume 3988 of {\em Lecture Notes in Computer Science}, pages 186--197.
  Springer, 2006.
\newblock \href {https://doi.org/10.1007/11780342\_20}
  {\path{doi:10.1007/11780342\_20}}.

\bibitem[ER06b]{DBLP:journals/tcs/EhrhardR06}
Thomas Ehrhard and Laurent Regnier.
\newblock Differential interaction nets.
\newblock {\em Theoretical Computer Science}, 364(2):166--195, 2006.
\newblock \href {https://doi.org/10.1016/j.tcs.2006.08.003}
  {\path{doi:10.1016/j.tcs.2006.08.003}}.

\bibitem[ER08]{Ehrhard:2008}
Thomas Ehrhard and Laurent Regnier.
\newblock {Uniformity and the Taylor expansion of ordinary lambda-terms}.
\newblock {\em Theoretical Computer Science}, 403(2-3):347--372, 2008.
\newblock \href {https://doi.org/10.1016/j.tcs.2008.06.001}
  {\path{doi:10.1016/j.tcs.2008.06.001}}.

\bibitem[Gir87]{Girard:1987}
Jean-Yves Girard.
\newblock {Linear logic}.
\newblock {\em Theoretical Computer Science}, 50(1):1--101, 1987.
\newblock \href {https://doi.org/10.1016/0304-3975(87)90045-4}
  {\path{doi:10.1016/0304-3975(87)90045-4}}.

\bibitem[GPT16]{FSCD2016}
Giulio Guerrieri, Luc Pellissier, and Lorenzo {Tortora de Falco}.
\newblock Computing connected proof(-structure)s from their {T}aylor expansion.
\newblock In {\em 1st International Conference on Formal Structures for
  Computation and Deduction ({FSCD} 2016)}, volume~52 of {\em LIPIcs}, pages
  20:1--20:18. Schloss Dagstuhl - Leibniz-Zentrum fuer Informatik, 2016.
\newblock \href {https://doi.org/10.4230/LIPIcs.FSCD.2016.20}
  {\path{doi:10.4230/LIPIcs.FSCD.2016.20}}.

\bibitem[GPT19]{Wollic}
Giulio Guerrieri, Luc Pellissier, and Lorenzo {Tortora de Falco}.
\newblock Proof-net as graph, {Taylor} expansion as pullback.
\newblock In {\em Logic, Language, Information, and Computation - 26th
  International Workshop (WoLLIC 2019)}, volume 11541 of {\em Lecture Notes in
  Computer Science}, pages 282--300. Springer, 2019.
\newblock \href {https://doi.org/10.1007/978-3-662-59533-6\_18}
  {\path{doi:10.1007/978-3-662-59533-6\_18}}.

\bibitem[GPT20]{CSL2020}
Giulio Guerrieri, Luc Pellissier, and Lorenzo {Tortora de Falco}.
\newblock Glueability of resource proof-structures: inverting the {Taylor}
  expansion.
\newblock In {\em 28th {EACSL} Annual Conference on Computer Science Logic,
  {CSL} 2020}, volume 152 of {\em LIPIcs}, pages 24:1--24:18. Schloss Dagstuhl,
  2020.
\newblock \href {https://doi.org/10.4230/LIPIcs.CSL.2020.24}
  {\path{doi:10.4230/LIPIcs.CSL.2020.24}}.

\bibitem[Kan94]{Kanovich94}
Max~I. Kanovich.
\newblock The complexity of {H}orn fragments of linear logic.
\newblock {\em Ann. Pure Appl. Log.}, 69(2-3):195--241, 1994.
\newblock \href {https://doi.org/10.1016/0168-0072(94)90085-X}
  {\path{doi:10.1016/0168-0072(94)90085-X}}.

\bibitem[KMP20]{KerinecMP18}
Emma Kerinec, Giulio Manzonetto, and Michele Pagani.
\newblock Revisiting call-by-value {B}{\"{o}}hm trees in light of their
  {T}aylor expansion.
\newblock {\em Log. Methods Comput. Sci.}, 16(3), 2020.
\newblock \href {https://doi.org/10.23638/LMCS-16(3:6)2020}
  {\path{doi:10.23638/LMCS-16(3:6)2020}}.

\bibitem[Laf90]{DBLP:conf/popl/Lafont90}
Yves Lafont.
\newblock Interaction nets.
\newblock In {\em Seventeenth Annual {ACM} Symposium on Principles of
  Programming Languages ({POPL} 1990)}, pages 95--108. {ACM} Press, 1990.
\newblock \href {https://doi.org/10.1145/96709.96718}
  {\path{doi:10.1145/96709.96718}}.

\bibitem[LMSS92]{Lincoln1992}
Patrick Lincoln, John Mitchell, Andre Scedrov, and Natarajan Shankar.
\newblock Decision problems for propositional linear logic.
\newblock {\em Annals of Pure and Applied Logic}, 56(1):239 -- 311, 1992.
\newblock \href {https://doi.org/10.1016/0168-0072(92)90075-B}
  {\path{doi:10.1016/0168-0072(92)90075-B}}.

\bibitem[MP07]{MazzaPagani07}
Damiano Mazza and Michele Pagani.
\newblock The separation theorem for differential interaction nets.
\newblock In {\em Logic for Programming, Artificial Intelligence, and
  Reasoning, 14th International Conference ({LPAR} 2007)}, volume 4790 of {\em
  Lecture Notes in Computer Science}, pages 393--407. Springer, 2007.
\newblock \href {https://doi.org/10.1007/978-3-540-75560-9\_29}
  {\path{doi:10.1007/978-3-540-75560-9\_29}}.

\bibitem[MPV18]{MazzaPellissierVial18}
Damiano Mazza, Luc Pellissier, and Pierre Vial.
\newblock Polyadic approximations, fibrations and intersection types.
\newblock {\em {PACMPL}}, 2({POPL}):6:1--6:28, 2018.
\newblock \href {https://doi.org/10.1145/3158094} {\path{doi:10.1145/3158094}}.

\bibitem[MR14]{ManzonettoR14}
Giulio Manzonetto and Domenico Ruoppolo.
\newblock Relational graph models, {T}aylor expansion and extensionality.
\newblock {\em Electronic Notes in Theoretical Compututer Science},
  308:245--272, 2014.
\newblock \href {https://doi.org/10.1016/j.entcs.2014.10.014}
  {\path{doi:10.1016/j.entcs.2014.10.014}}.

\bibitem[Pag09]{Pagani09}
Michele Pagani.
\newblock The cut-elimination theorem for differential nets with promotion.
\newblock In {\em Typed Lambda Calculi and Applications, 9th International
  Conference, ({TLCA} 2009)}, volume 5608 of {\em Lecture Notes in Computer
  Science}, pages 219--233. Springer, 2009.
\newblock \href {https://doi.org/10.1007/978-3-642-02273-9\_17}
  {\path{doi:10.1007/978-3-642-02273-9\_17}}.

\bibitem[Par92]{Parigot1992}
Michel Parigot.
\newblock $\lambda\mu$-calculus: An algorithmic interpretation of classical
  natural deduction.
\newblock In {\em Logic Programming and Automated Reasoning,International
  Conference ({LPAR} '92)}, volume 624 of {\em Lecture Notes in Computer
  Science}, pages 190--201. Springer, 1992.
\newblock \href {https://doi.org/10.1007/BFb0013061}
  {\path{doi:10.1007/BFb0013061}}.

\bibitem[PT09]{PaganiTasson2009}
Michele Pagani and Christine Tasson.
\newblock The inverse {Taylor} expansion problem in linear logic.
\newblock In {\em 24th Annual Symposium on Logic in Computer Science ({LICS}
  2009)}, pages 222--231. {IEEE} Computer Society, 2009.
\newblock \href {https://doi.org/10.1109/LICS.2009.35}
  {\path{doi:10.1109/LICS.2009.35}}.

\bibitem[PTV16]{DBLP:conf/fossacs/PaganiTV16}
Michele Pagani, Christine Tasson, and Lionel Vaux.
\newblock Strong normalizability as a finiteness structure via the {Taylor}
  expansion of {\(\lambda\)}-terms.
\newblock In {\em Foundations of Software Science and Computation Structures -
  19th International Conference ({FOSSACS} 2016)}, volume 9634 of {\em Lecture
  Notes in Computer Science}, pages 408--423. Springer, 2016.
\newblock \href {https://doi.org/10.1007/978-3-662-49630-5\_24}
  {\path{doi:10.1007/978-3-662-49630-5\_24}}.

\bibitem[Ret97]{Retore97}
Christian Retor{\'{e}}.
\newblock A semantic characterisation of the correctness of a proof net.
\newblock {\em Mathematical Structures in Computer Science}, 7(5):445--452,
  1997.
\newblock \href {https://doi.org/10.1017/S096012959700234X}
  {\path{doi:10.1017/S096012959700234X}}.

\bibitem[Tas09]{Tasson:2009}
Christine Tasson.
\newblock {\em {S{\'e}mantiques et Syntaxes Vectorielles de la Logique
  Lin{\'e}aire}}.
\newblock PhD thesis, Universit\'e Paris Diderot, France, December 2009.
\newblock URL: \url{https://tel.archives-ouvertes.fr/tel-00440752}.

\bibitem[TdF03]{Tortora03}
Lorenzo Tortora~de Falco.
\newblock Additives of linear logic and normalization - part {I:} a
  (restricted) church-rosser property.
\newblock {\em Theor. Comput. Sci.}, 294(3):489--524, 2003.
\newblock \href {https://doi.org/10.1016/S0304-3975(01)00176-1}
  {\path{doi:10.1016/S0304-3975(01)00176-1}}.

\bibitem[Tra09]{Tranquilli09}
Paolo Tranquilli.
\newblock Confluence of pure differential nets with promotion.
\newblock In {\em Computer Science Logic, 23rd international Workshop, {CSL}
  2009, 18th Annual Conference of the EACSL}, volume 5771 of {\em Lecture Notes
  in Computer Science}, pages 500--514. Springer, 2009.
\newblock \href {https://doi.org/10.1007/978-3-642-04027-6\_36}
  {\path{doi:10.1007/978-3-642-04027-6\_36}}.

\bibitem[Tra11]{Tranquilli11}
Paolo Tranquilli.
\newblock Intuitionistic differential nets and lambda-calculus.
\newblock {\em Theor. Comput. Sci.}, 412(20):1979--1997, 2011.
\newblock \href {https://doi.org/10.1016/j.tcs.2010.12.022}
  {\path{doi:10.1016/j.tcs.2010.12.022}}.

\bibitem[Vau17]{DBLP:conf/csl/Vaux17}
Lionel Vaux.
\newblock Taylor expansion, lambda-reduction and normalization.
\newblock In {\em 26th {EACSL} Annual Conference on Computer Science Logic
  ({CSL} 2017)}, volume~82 of {\em LIPIcs}, pages 39:1--39:16. Schloss Dagstuhl
  - Leibniz-Zentrum fuer Informatik, 2017.
\newblock \href {https://doi.org/10.4230/LIPIcs.CSL.2017.39}
  {\path{doi:10.4230/LIPIcs.CSL.2017.39}}.

\end{thebibliography}

\end{document}